\documentclass[11pt]{article}
\usepackage[margin=1in]{geometry}

\usepackage{amsmath}
\usepackage{amsfonts}
\usepackage{amssymb}
\usepackage{amsthm}
\usepackage{indentfirst}
\usepackage{graphicx}
\usepackage{float}

\usepackage{algorithm}
\usepackage{algorithmicx}
\usepackage[noend]{algpseudocode}

\usepackage[utf8]{inputenc}

\usepackage{multirow}

\usepackage{parskip}
\usepackage{hyperref}
\usepackage{cleveref}
\usepackage[dvipsnames]{xcolor}

\newtheorem{theorem}{Theorem}[section]

\newtheorem{lemma}{Lemma}[section]
\newtheorem{corollary}{Corollary}[section]
\newtheorem{definition}{Definition}[section]
\newtheorem{remark}{Remark}[section]
\newtheorem{claim}{Claim}[section]

\DeclareMathOperator*{\argmin}{argmin}
\newcommand{\comment}[1]{}

\usepackage[textsize=tiny]{todonotes}

\title{Improved Product-state Approximation Algorithms for Quantum Local Hamiltonians}
\author{Thiago Bergamaschi\thanks{UC Berkeley. Email: \texttt{thiagob@berkeley.edu}.}}
\date{\today}

\begin{document}

\maketitle

\begin{abstract}
    The ground state energy and the free energy of Quantum Local Hamiltonians are fundamental quantities in quantum many-body physics, however, it is QMA-Hard to estimate them in general. In this paper, we develop new techniques to find classical, additive error product-state approximations for these quantities on certain families of Quantum $k$-Local Hamiltonians. Namely, those which are either dense, have low threshold rank, or are defined on a sparse graph that excludes a fixed minor, building on the methods and the systems studied by Brandão and Harrow, Gharibian and Kempe, and Bansal, Bravyi and Terhal.
    
    We present two main technical contributions. First, we discuss a connection between product-state approximations of local Hamiltonians and combinatorial graph property testing. We develop a series of \textit{weak Szemer\'edi regularity} lemmas for $k$-local Hamiltonians, built on those of Frieze and Kannan and others. We use them to develop \textit{constant time} sampling algorithms, and to characterize the `vertex sample complexity' of the Local Hamiltonian problem, in an analog to a classical result by Alon, de la Vega, Kannan and Karpinski. Second, we build on the information-theoretic product-state approximation techniques by Brandão and Harrow, extending their results to the free energy and to an asymmetric graph setting. We leverage this structure to define families of algorithms for the free energy at low temperatures, and new algorithms for certain sparse graph families.
\end{abstract}
\newpage

\section{Introduction}

The mean-field approximation is a popular heuristic in quantum many-body physics, in which product-states are used as an ansatz for generic quantum states. The low-energy states of quantum systems may be highly entangled objects, and possibly exponentially more complex than simple (unentangled) product states. This often makes computing properties of these low-energy states classically intractable. From a complexity-theoretic point of view, the mean-field approach casts these quantum problems that are in the complexity class QMA \cite{Kitaev2002ClassicalAQ}, into problems in NP, since product-states have a polynomial-size description and can act as classical, efficiently verifiable certificates. However, in the absence of a hardness-of-approximation result for QMA \cite{scottmanifesto, Arad2011ANA, Aharonov2009TheDL, Hastings2013TrivialLE} and assuming QMA$\neq$NP, it is generally unknown if the ground states of quantum systems can even have `good' approximations with succinct classical descriptions, let alone if we can compute or approximate them efficiently.

In this work, we develop a series of classical algorithms to efficiently find mean-field approximations for quantum systems described by \textit{local Hamiltonians}, and we develop new techniques to show that good mean-field approximations exist for fairly general classes of these systems. A local Hamiltonian corresponds to a sparse matrix $H \in \mathbb{C}^{d^n\times d^n}$ which is exponentially large in the number $n$ of quantum particles (or qu\textit{d}its), and can be described as a sum over `local' terms $H  = \sum_{e\in E}h_e$ defined by some hypergraph $G=([n],E)$. $H$ is said to be $k$-local if each hyperedge $e\in E$ is a $k$-tuple of vertices in $[n]$, see section \ref{subsection-background} for more details. 

It is well known that the existence of product-state approximations to $H$ is very sensitive to the structure of the underlying interaction graph $G$. In a seminal result, \cite{Brando2013ProductstateAT} proved that so long as $H$ has bounded interaction strengths, and is defined on a graph $G$ of high degree or small expansion, then there exists a product state which approximates the ground state energy of $H$ up to an additive error $\epsilon\cdot m$ (scaling with the number of edges or `interactions' $m$ of $H$). Their results can be interpreted as rigorous proofs of accuracy of the mean-field approximation to the ground state energy of certain systems, and they opened the door to classical approximation schemes to find these `good' mean-field solutions efficiently. One of the main focuses of this work is to relax certain assumptions on the structure of the interaction graphs $G$, to extend the scope of their algorithms and existence statements.

The second main focus of this work is to study the structure and classical computation of properties of quantum systems in thermal equilibrium. The Helmholtz Free Energy $F(\beta)$ of a Quantum Local Hamiltonian $H$ at a given temperature $\beta^{-1}$ arises as an approximate counting analog to the ground state energy, as it reveals the degeneracy of the ground state (the number of QMA witnesses), the density of states of the Hamiltonian, as well as the existence of phase transitions. Quantitatively, $F(\beta)$ can be described as the optimum of a maximum entropy program:

\begin{equation}
    F(\beta) \equiv \min_{\rho\geq 0, \|\rho\|_1=1} f(\rho) = \min_{\rho\geq 0, \|\rho\|_1=1} \text{Tr}[H\rho] - S(\rho)/\beta
\end{equation}

Where the optimizer $\rho\propto e^{-\beta H}$ of the program above is called the Gibbs state of $H$. The computational complexity, and in particular the hardness of approximation of $F(\beta)$ is similarly not comprehensively understood. While QMA-Hard to estimate in general due to a reduction to the `low temperature' limit, and exactly computable in polynomial time using a \#P oracle \cite{Brown2011ComputationalDO}, it would seem there is much to uncover regarding the computational tradeoffs between error and temperature \cite{Bravyi2021OnTC}.

\subsection{Our Main Contributions}

In this section we overview our main contributions, which we present formally and in more detail in section \ref{subsection-results}. \\

\textbf{Rigorous Mean-Field Approximations and Guarantees in NP}

Our first contributions concern improvements and extensions to the existence statements by \cite{Brando2013ProductstateAT}. Their methods had roots in the information-theoretic techniques by \cite{Raghavendra2012ApproximatingCW} and \cite{Barak2011RoundingSP}, developed in the context of approximating CSPs using the Lasserre Heirarchy. Informally, we show how to use their self-decoupling arguments to construct mixed states which are tensor products of single-particle mixed states, which approximate the Free Energy up to an additive error. We view these results as rigorous proofs of accuracy for the mean-field approximation to the Free Energy of Quantum Local Hamiltonians, and they imply that approximating the Free Energy of dense Hamiltonians up to an extensive error (scaling with the number of edges) is in NP.

\begin{theorem} 
Fix $d= O(1), \epsilon>0$, and an inverse temperature $\beta$. Let $H = \sum_{e\in E}h_e$ be a $2$-Local Hamiltonian on $n$ qudits of local dimension $d$, and $m = \Omega(n/\epsilon^3)$ interactions each of strength $\|h_e\|_{\infty} \leq 1$. Then, there exists a product state $\sigma_\beta =\otimes_{u\in [n]} \sigma_{u}, \sigma_u\in \mathbb{C}^{d\times d}$ such that

\begin{equation}
    F(\beta)\leq f(\sigma_\beta) = \text{Tr}[H\sigma_\beta] - S(\sigma_\beta)/\beta \leq F(\beta) + \epsilon\cdot  m
\end{equation}
\end{theorem}

We emphasize two important points about the result above. First and foremost, the existence of approximations to $F(\beta)$ in NP implies that we can now use classical approximation schemes to search for optimal mean-field approximations to the free energy, and they will also be good approximations to the `entangled value' of $F(\beta)$. As we discuss later, this enables us to import practically all the previous machinery of approximation schemes for the ground state energy, to the Free Energy, developing novel algorithms for many quantum systems and improving on recent results. 

The second point of emphasis is that the result above holds at all temperatures $\beta^{-1}$. In this fashion, we are able to bypass the `low temperature bottleneck' of many approximation schemes for the Free Energy which constrain approaches in previous work, such as the polynomial interpolation method \cite{Barvinok2016CombinatoricsAC} or Markov Chain Monte Carlo methods. We present a comprehensive comparison with previous work and the scope of our techniques for thermal systems in section \ref{subsection-related}. \\

\textbf{Hamiltonian Regularity Lemmas, Approximation Algorithms and Property Testing} 

From an algorithmic point of view, our main contribution is a connection between product state approximations and graph property testing. We discuss quantum analogs of the \textit{weak Szemer\'edi regularity lemmas} for dense graphs, hyper-graphs and low-threshold rank graphs \cite{Frieze1999QuickAT,Alon2002RandomSA, Gharan2013ANR}, developed in the context of additive approximation schemes for Max-Cut and Max-kCSPs. At their heart lies a powerful combinatorial characterization of these systems, Szemer\'edi's celebrated regularity lemma  \cite{Szemerdi1975RegularPO}, which states that dense graphs can be approximately decomposed into unions of complete bipartite graphs. We develop natural, constructive generalizations of these results for Quantum Local Hamiltonians, by combining our new product state approximations with multi-coloured versions of known weak regularity results, leading to improved approximation algorithms and novel structural characterizations of local Hamiltonians. Our central result in this vein is an additive error approximation scheme for dense $k$-Local Hamiltonians, which runs in \textit{constant time}:

\begin{theorem} \label{results-gsptasdense}
Fix $d, k=O(1)$, $\epsilon > 0$, and let $H = \sum_e h_e$ be a $k$-Local Hamiltonian on $n$ qudits of local dimension $d$ and bounded strength interactions $\|h_e\|_\infty\leq 1$. Then, there exists a randomized algorithm which runs in time $2^{\text{poly}(1/\epsilon)}$, and with probability $.99$ returns an estimate for the ground state energy of $H$ accurate up to an additive error of $\epsilon \cdot n^{k}$. 
\end{theorem}

We report our sampling algorithms, including that in theorem \ref{results-gsptasdense}, in the \textit{probe model of computation} introduced by \cite{Goldreich1998PropertyTA}. In a nutshell, the time complexity measured above corresponds to the number of queries to a description of $H$, see section \ref{subsection-background} for more details. In the body, we show how these ideas can be used to develop improvements in runtime from $n^{\text{poly}(1/\epsilon)}$ to $\text{poly}(n, 1/\epsilon) + 2^{\text{poly}(1/\epsilon)}$ or $2^{\text{poly}(1/\epsilon)}$ for a wide range of problems on Quantum Local Hamiltonians, such as approximation schemes for the ground state energy, the Free Energy, and for Hamiltonians defined on low threshold rank graphs. 

\section{Technical Overview}

\subsection{Background and Notation}
\label{subsection-background}

\noindent \textbf{Linear Algebra and Matrix Norms} Given an $w\times w$ matrix $A$ we refer to $\|A\|_p$ as the Schatten $p$-norm of $A$, the $L_p$ norm of the singular values of $A$, and we refer to $|A|_p$ as the $L_p$ norm of the $w^2$-dimensional vectorization of $A$. The graph decompositions are phrased in terms of the \textit{cut norm} $\|A\|_C$ introduced by \cite{Frieze1999QuickAT}, defined by

\begin{equation}
    A^+ = \max_{S_1, S_2\subseteq [w]} \sum_{i\in S_1, j\in S_2}A_{ij} \text{ and } \|A\|_C = \max(A^+, (-A)^+)
\end{equation}

\noindent where we have $\|A\|_C\leq \|A\|_{\infty\rightarrow  1} = \sup_{x\neq  0} \frac{|Ax|_1} {|x|_\infty}\leq 4\cdot \|A\|_C$.\\

\noindent \textbf{Asymptotic Notation} For any function $f(n)$ we refer to the asymptotic notation $\tilde{O}(f(n)) = O(f(n)\text{polylog}(f(n)))\leq  c_1\cdot f(n) \log^{c_2} f(n)$ for a choice of real positive constants $c_1,  c_2$. \\

\noindent \textbf{Local Hamiltonians} We denote a $k$-Local Hamiltonian on $n$ qudits of local dimension $d$ via a $d^n\times d^n$ Hermitian matrix, which can be expressed as a sum of local interactions $H = \sum_{e\in E} h_e$. By `local', we simply mean that each summand $h_e = H_e\otimes \mathbb{I}_{V\setminus e}$ acts non-trivially only on $k$ particles at a time, as indicated by each $k$-tuple $e = (u_1\cdots u_k)$ in a set of hyper-edges $E$. In this manner, we can specify any Local Hamiltonian `instance' simply by specifying the $d^k\times d^k$ submatrices of each local term. If $d, k=O(1)$, then the input has a polynomial-sized description in $n$. For notational convenience, we often omit the trivial support $\mathbb{I}_{V\setminus e}$. The ground state energy and the ground state of $H$ are its minimum eigenvalue and corresponding eigenvector, and the \textit{variational} minimum energy of $H$ is the minimum energy of $H$ among all product states $\min_{\rho = \otimes \rho_u}\text{Tr}[H\otimes_u \rho_u]$ with $\rho_u\in \mathbb{C}^{d\times d}$ and $\rho_u\geq 0,\text{Tr}_u[\rho_u]=1$. \\

\noindent \textbf{Interaction Graphs} We refer to the `Interaction Graph' $G=([n], E)$ of a 2-Local Hamiltonian $H$ as the graph with undirected edges $e = (u, v)\in E$ whenever the particles $u, v$ interact non-trivially in $H$. That is, whenever the spectral norm is non-zero $\|H_e\|_\infty \neq 0$. By expressing each $d^2\times d^2$ Hermitian matrix $H_{u, v} =  \sum_{i,j\in [d^2]}H^{i, j}_{u, v} \cdot  \sigma^i_u\otimes \sigma^j_v$ in an orthogonal basis decomposition, and grouping all the interactions with the same basis $i, j$, we refer to the $i, j$ `Pauli Graph' as the subgraph of $G$ induced on all the \textit{directed} edges $e = (u, v)$ with non-zero $H^{i, j}_{u, v} = d^{-2}\text{Tr}[H_{u, v} \sigma^i_u\otimes \sigma^j_v]$, with weighted adjacency matrix $J^{ij}=\{H^{i, j}_{u, v}\}_{u, v\in [n]}$. We note that the matrices $J^{ij}$ are degenerate, since $J^{ij} = (J^{ji})^T$, but we often brush over this issue via a handshaking argument. If we are given a density matrix $\rho = \otimes \rho_u$ which is a product of single qudit density matrices with a basis decomposition $\rho_u = d^{-1}\sum_i \alpha^i_u\cdot \sigma^i$, then the energy of $\rho$, $\text{Tr}[H\rho]$ is a polynomial over the real variables $\alpha$:
\begin{equation}
    \sum_{(u,v)\in E} \text{Tr}[H_{u,v}\rho_u\otimes\rho_v] = d^{-2}\sum_{(u,v)\in E}\sum_{i, j\in [d^2]} H^{i, j}_{u, v} \alpha^i_u\cdot \alpha^j_v = (2d^2)^{-1} \sum_{i, j\in [d^2]} \sum_{u\neq v\in [n]}J^{ij}_{uv}\alpha^i_u\cdot \alpha^j_v
\end{equation}

\noindent \textbf{Model of Computation} We report our sampling algorithms in the \textit{probe model of computation} introduced by \cite{Goldreich1998PropertyTA} in the context of graph property testing. That is, we assume we can sample a uniformly random vertex or hyper-edge in $O(1)$ time (or `probes'). Formally, fixed a $k$-Local Hamiltonian `instance' $H = \sum_{e\in E}H_e\otimes \mathbb{I}_{V\setminus e}$, for any $k$-tuple of vertices/hyper-edge $e = (u_1\cdots u_k),  u_i\in [n]$, we assume we can query the (constant-sized) $d^k\times d^k$ sub-matrix $H_e$ in $O(1)$ time. We emphasize that since our goal is often a sublinear time algorithm, we always enforce that our algorithms output estimates for the energy (or free energy), and \textit{implicit} descriptions of product states. If requested, these implicit descriptions can always be expanded into $n$-qudit product states in an additional polynomial time. \\

\noindent \textbf{Extensive Errors} We refer to an additive approximation scheme, or an additive error, for a problem on a Hamiltonian $H$ (or graph $G$), as \textit{extensive} if the notion of error scales with the number of interactions of $H$ (resp., edges of $G$). For instance, an $\epsilon\cdot m$ additive approximation scheme for Max-Cut on graphs of $m$ edges is an `extensive' error approximation scheme for constant $\epsilon$. We pay particular attention to this limit of additive errors since the PCP theorem \cite{Arora1992ProofVA} ensures it is NP-Hard to approximate Max-Cut up to said error in general.

\subsection{Our Results}
\label{subsection-results}

\textbf{Approximation Guarantees in NP}

The first of our results are rigorous proofs of accuracy of the mean-field approximation on Quantum $k$-Local Hamiltonians. We argue the existence of product states, or products of single-particle mixed states, which provide additive error approximations to the ground state energy and the free energy of these systems. We build on the information-theoretic techniques by \cite{Brando2013ProductstateAT}, presenting an extension to the free energy and modestly refining their techniques on generic (hyper-) graphs.

\begin{theorem} \label{results-psa}
Fix $k, d = O(1)$. Let $H = \sum_{e\in E}h_e$ be a $k$-Local Hamiltonian on $n$ qudits of local dimension $d$, and $m$ interactions each of strength $\|h_e\|_{\infty} \leq 1$. Then, there exists a product state $\psi\rangle = \otimes_{u\in [n]} |\psi_u\rangle,|\psi_u\rangle \in \mathbb{C}^d$ such that
\begin{equation}
    \langle \psi|H|\psi\rangle \leq \min_{\phi}  \langle \phi|H|\phi\rangle + O(n^{\frac{k-1}{3}}m^{2/3})
\end{equation}
\end{theorem}

 In the body, we prove more general versions of the theorem above sensitive to the matrix of interaction strengths of $H$. Theorem \ref{results-psa} matches the previous results in \cite{Brando2013ProductstateAT} whenever the Hamiltonian is defined on $D$-regular or dense graphs $m = \Omega(n^k)$, and generalizes these statements to an asymmetric setting on Hamiltonians which are dense on average. In the setting of Theorem \ref{results-psa}, whenever $m = \Omega(n^{k-1}/\epsilon^{3})$, approximating the ground state energy of $H$ up to additive error $\epsilon\cdot m$ is in the complexity class NP, as the product state has a polynomial size description and acts as a classical witness. While these optimal product states may be NP-Hard to find in the worst case, there are many examples where one can approximate these solutions efficiently.

To extend both these information-theoretic ideas and algorithmic applications to the free energy, we need further insights on the structure of these product state approximations. We discuss in section \ref{section-existence} how the `entanglement-breaking' procedure of \cite{Brando2013ProductstateAT}, not only approximately preserves the energy, but in fact also increases the entropy as well. When applied to the Gibbs state, we show one can carefully extract a tensor product of single particle mixed-states which is a good approximation to the free energy. We formalize this statement in Theorem \ref{results-psafe},

\begin{theorem} \label{results-psafe}
Fix $k, d= O(1)$, and an inverse temperature $\beta$. Let $H = \sum_{e\in E}h_e$ be a $k$-Local Hamiltonian on $n$ qudits of local dimension $d$, and $m$ interactions each of strength $\|h_e\|_{\infty} \leq 1$. Then, there exists a product state $\sigma_\beta =\otimes_{u\in [n]} \sigma_{u}, \sigma_u\in \mathbb{C}^{d\times d}$ such that

\begin{equation}
    f(\sigma_\beta) = \text{Tr}[H\sigma_\beta] - S(\sigma_\beta)/\beta \leq F + O(n^{\frac{k-1}{3}}m^{2/3})
\end{equation}
\end{theorem}

We emphasize that the statement above implies a product state approximation exists at all temperatures $\beta^{-1}$ (and recovers the ground state approximation at $T = 0$), and moreover uses very little of the underlying graph structure apart from the \textit{average} dense condition. \\

\textbf{Hamiltonian Weak Regularity Lemmas}

We develop an approach to designing approximations algorithms for Local Hamiltonians based on \textit{weak Szemer\'edi regularity lemmas}, which are approximate decompositions to graphs, matrices, and tensors \cite{Szemerdi1975RegularPO, Frieze1999QuickAT, Alon2002RandomSA, Gharan2013ANR}. 

The idea behind this construction lies in a powerful tool in extremal combinatorics. In his celebrated regularity lemma, \cite{Szemerdi1975RegularPO} proved that any dense graph can be approximated by a union of a constant number of complete bipartite graphs. However, the number of partitions grew very fast with the intended quality of approximation. \cite{Frieze1999QuickAT} developed a constructive decomposition under a weaker notion of approximation, what they refered to as a ``weak" regularity lemma. Concretely, they prove that any real $n\times n$ matrix with bounded entries can be decomposed into a sum of $O(1/\epsilon^2)$ cut matrices (complete bipartite graphs), up to an error $\epsilon\cdot n^2$ in the cut norm. Moreover, \cite{Frieze1999QuickAT} prove that one can in fact construct such a ``cut decomposition" implicitly in time polynomial in $1/\epsilon$, which enables them to devise constant time sampling-based approximation schemes for many problems on dense graphs. 

We define a natural adaptation of their results to a quantum setting, by constructing an approximate decomposition $H_D$ of a Local Hamiltonian $H$ which is a sum over complete, bipartite, sub-Hamiltonians. The structure of $H_D$ can be understood as a `multi-colored' matrix cut decomposition, as essentially we apply the cut decomposition by \cite{Frieze1999QuickAT} to each term in a basis decomposition of $H$. For concreteness, let $H = \sum_{u, v}h_{u, v}$ be a $2$-Local Hamiltonian on qubits, and let us consider re-writing its Pauli basis decomposition below. We suppress the identity terms $\otimes \mathbb{I}_{V\setminus\{u, v\}}$ on the qubits that each interaction acts trivially on.

\begin{equation}
    H = \sum_{(u, v)\in E}h_{u, v} = \sum_{(u, v)\in E}\sum_{i, j\in\{I, X, Y, Z\}} h^{i, j}_{u, v}\sigma_{u}^i\otimes \sigma_{v}^j = \sum_{i, j\in\{I, X, Y, Z\}}\sum_{u< v} h^{i, j}_{u, v}\sigma_{u}^i\otimes \sigma_{v}^j
\end{equation}

We associate each pair of indices $i, j \in\{I, X, Y, Z\}$ to a color, and consider the $n\times n$ real valued weighted adjacency matrix $J^{ij} =\{h^{i, j}_{u, v}\}_{u,v\in [n]}$ of the $i, j$ `Pauli Graph'. By applying the cut decomposition by \cite{Frieze1999QuickAT} to each of these $16$ matrices $J^{ij}$, we construct an approximate decomposition of $H$ into roughly $16\cdot O(1/\epsilon^2)$ complete bipartite sub-Hamiltonians. In this context, a `complete bipartite sub-Hamiltonian' is defined by two Pauli matrices, (say, $X, Y$), two subsets $S, T\subset [n]$ (which, for now, we assume to be disjoint), and an interaction strength $\alpha\in \mathbb{R}$, and can be expressed as $\alpha \sum_{u\in S, v\in T}X_u\otimes Y_v$. 

In the body we argue that the approximation guarantees in the cut norm are precisely what we need to ensure that for any product state $\sigma = \otimes_u \sigma_u$,  the energy of $\sigma$ under $H$ or $H_D$ are close: $\text{Tr}[H\sigma]\approx \text{Tr}[H_D\sigma]$. By further combining this product state regularity with our asymmetric product state approximations, we prove a stronger property on the spectra of $H_D$:

\begin{lemma} [Informal\label{results-reg}]
Fix $d, k= O(1)$ and a constant $\epsilon>0$, and let $H = \sum_e h_e$ be a $k$-Local Hamiltonian on $n$ qudits of local dimension $d$ and $m$ interactions of strength bounded by $\|h_e\|_\infty\leq 1$. Then, there exists a decomposition $H_D= \sum_i^s D^{(i)}$ of $H$ into $s=O(1/\epsilon^2)$ complete bipartite sub-Hamiltonians such that
\begin{equation}
    \|H-H_D\|_\infty \leq \epsilon \cdot n^{k/2}m^{1/2}
\end{equation}
\end{lemma}

\textbf{Additive Error Approximation Schemes}

Leveraging the structure of the Hamiltonian regularity Lemma \ref{results-reg} in combination with the product state approximation toolkit enables us to devise a series of approximation schemes for Quantum Local Hamiltonians. We follow the ideas of \cite{Frieze1999QuickAT, Alon2002RandomSA, Gharan2013ANR} in establishing LP relaxations to Max Cut and other Max CSPs, and we develop an SDP relaxation scheme for finding the minimal energy product state of a Local Hamiltonian. These ideas enable us to devise an efficient additive error approximation scheme for dense Hamiltonians,

\begin{theorem} [Theorem \ref{results-gsptasdense}, restatement]
Fix $d, k=O(1)$ and $\epsilon > 0$. Let $H = \sum_e h_e$ be a $k$-Local Hamiltonian on $n$ qudits of local dimension $d$, and $m$ interactions of bounded strength $\|h_e\|_\infty\leq 1$. There exists a randomized algorithm which runs in time $2^{\tilde{O}(1/\epsilon^{2k-2})}$ in the probe model of computation, and with probability $.99$ computes an estimate for the ground state energy of $H$ accurate up to an additive error of $\epsilon \cdot n^{k/2}\sqrt{m}$. 
\end{theorem}

We note that $n^{k/2}\sqrt{m} \geq m$, and thus in polynomial or sublinear time this approximation scheme only provides a non-trivial guarantee when the hyper-graph is dense, $m = \Omega(n^k/\log^c n)$ for some small positive constant $c$. However, it provides an improvement over the $n^{O(1/\epsilon^2)}$ time algorithms by \cite{Gharibian2011ApproximationAF} and \cite{Brando2013ProductstateAT} in this additive error regime. On the other hand, a simple explicit variant of this result provides a sub-exponential time approximation algorithm whenever $m = \omega(n^{k-1}\log n)$:

\begin{theorem} 
In the context of Theorem \ref{results-gsptasdense}, there exists a randomized algorithm which runs in time $\tilde{O}(n^k)\cdot 2^{\tilde{O}(n^k/\epsilon^2 m)}$ and with high probability computes an estimate for the ground state energy of $H$ accurate up to an additive error of $\epsilon \cdot  m$. 
\end{theorem}

Concretely, the key idea behind these relaxations is that for any product state $\sigma = \otimes_u \sigma_u$, the energy of $\sigma$ on the cut decomposition $H_D$ is a simple function of the average magnetization of a small number of subsets of the $n$ qudits. To illustrate how this enables a relaxation scheme, consider a single complete bipartite sub-Hamiltonian, such as $H_{S, T} = \sum_{u\in S, v\in T} X_u\otimes Y_v$. The energy of $\sigma$ on $H_{S, T}$ is

\begin{equation}
    \text{Tr}[H_{S, T}\sigma] =  \sum_{u\in S, v\in T} \text{Tr}_{u, v}[ X_u \sigma_u\otimes Y_v \sigma_v\big] = \bigg(\sum_{u\in S} \text{Tr}[ X_u \sigma_u]\bigg)\cdot \bigg(\sum_{v\in T} \text{Tr}[ Y_v \sigma_v]\bigg),
\end{equation}

\noindent simply the product of the average $X$ direction magnetization of $S\subset [n]$ with the average $Y$ magnetization of $T$. If we fix a `guess' $r, c\in [-n, n]$, one can introduce affine constraints on the single particle density matrices $\sigma_u$, constraining their average magnetizations to lie within a $\pm \gamma \cdot n$ range of the guess $r, c$:
\begin{gather}
    r - \gamma \cdot n \leq \sum_{u\in S} \text{Tr}[ X_u \sigma_u]\leq r + \gamma \cdot n \\
    c - \gamma \cdot n \leq \sum_{v\in T} \text{Tr}[ Y_v \sigma_v]\leq c + \gamma \cdot n 
\end{gather}

Then we are guaranteed that any product state $\sigma$ which is feasible for the constraints above must have energy in a range around the guess: $|\text{Tr}[H_{S, T}\sigma]- r\cdot c|\leq (2\cdot \gamma+\gamma^2) \cdot n^2$. In this manner, one can discretize over the space of `guesses' $(r, c)$ and define an overlapping set of convex constraints on the description of the product states $\sigma$, such that every product state is feasible for at least one set of constraints. Approximating the ground state energy among product states ultimately reduces to checking the feasibility of a constant number of SDPs, one for each guess of $r, c$, and outputting whichever gives us the smallest energy estimate. 

Using the techniques by \cite{Gharan2013ANR}, we can extend these insights to the setting of symmetric $2$-Local Hamiltonians defined on graphs of low threshold rank. They proved that the weak regularity results of \cite{Frieze1999QuickAT} could be extended to low-threshold rank graphs, by constructing a cut decomposition of a low rank approximation to the normalized adjacency matrix of these graphs. While in the appendix we formalize approximation algorithms for generic symmetric Hamiltonians (on low threshold rank graphs), perhaps the most faithful extension of this result to the quantum setting would be its application to approximating the Quantum Max Cut \cite{Gharibian2019AlmostOC, Parekh2021BeatingRA, Parekh2021ApplicationOT, Parekh2022AnOP}. Given an undirected graph $G = (V, E)$,  the ``Quantum Max-Cut" corresponds to the maximum eigenvalue of the Hamiltonian 

\begin{equation}
    H = \frac{1}{2}\sum_{e\in E}\bigg(\mathbb{I}_u\otimes \mathbb{I}_v  - X_u\otimes X_v- Y_u\otimes Y_v- Z_u\otimes Z_v\bigg) \otimes \mathbb{I}_{V\setminus\{u, v\}}
\end{equation}

If $A$ is the adjacency matrix of $G$ and $D$ the diagonal matrix of degrees, the $\delta$-SOS threshold rank $t_\delta(A)$ of $A$ is the number of eigenvalues of the normalized adjacency matrix $D^{-1/2}AD^{-1/2}$ which are outside of the range $[-\delta, \delta]$. We prove

\begin{theorem}\label{results-qmc}
    Fix $\epsilon, \delta > 0$. Let $G = (V, E)$ be a graph on $n$ vertices and $m$ edges with adjacency matrix $A$ and threshold rank $t\equiv t_{\epsilon/2}(A)$. Then, there exists an algorithm which finds an $\epsilon\cdot m + O(n^{1/3}m^{2/3})$ additive error approximation to the Quantum Max Cut of $G$ in time $\text{poly}(n, 1/\epsilon, t) + 2^{\tilde{O}(t/\epsilon^2)}$.
\end{theorem}

For instance, sparse $D$-regular random graphs have $\Theta(D^{-1/2})$-SOS threshold rank $1$. In this manner, for any constant $\epsilon$ and if $D = \Omega(1/\epsilon^{3})$, then one can compute an $\epsilon\cdot m$ approximation to the Quantum Max Cut of a $D$-regular random graph in polynomial time.

A series of works \cite{Jain2018TheMA, Jain2018TheVS, Jain2019MeanfieldAC} showed that the matrix weak regularity lemma \cite{Frieze1999QuickAT} could be used to approximate the free energy of Ising Models, and to give interesting structural results on the quality of the mean-field approximation and the `vertex sample complexity' of these systems. They observed that the maximum entropy program subject to the linear relaxation constraints described above, reveals properties of the Gibbs distribution and enables an additive error approximation to the free energy at all temperatures. By combining these ideas with the Hamiltonian regularity Lemma \ref{results-reg} and Theorem \ref{results-psafe} on product state approximations to the free energy, we develop a series of additive error approximation schemes for the free energy of Quantum Local Hamiltonians. The first of which is a constant time approximation scheme, which provides an additive error guarantee in a low temperature regime.

\begin{theorem} \label{results-variationalfe}
Fix $k, d = O(1)$, and $\epsilon, \delta > \omega(n^{-1/(2k-2)})$ and an inverse temperature $\beta > 0$, and let $H$ be a $k$-Local Hamiltonian on $n$ qudits of local dimension $d$ and $m$ bounded strength interactions. Then, there exists an algorithm that runs in time $2^{\tilde{O}(\epsilon^{2-2k})}\cdot O(\delta^{-2})$ in the probe model of computation, that returns an estimate to the free energy accurate up to an additive error of $\epsilon n^{k/2}m^{1/2} + \delta n /\beta$ and is correct with probability $.99$. 
\end{theorem}

We emphasize that the free energy is a convex program regularized by temperature, and thereby our approximation schemes often incur a tradeoff between combinatorial errors and thermal (temperature dependent) errors. In the low temperature regime, whenever $\beta = \Omega(n^{1-k/2}m^{-1/2})$, the algorithm above recovers the behavior of the ground state energy approximation scheme, and is largely temperature independent. However, as the temperature increases and surpasses the threshold, the leading source of error becomes the thermal error $\delta n/\beta$. In our second algorithm, we show that an explicit approach significantly improves this thermal error dependence, at the cost of a polynomial runtime.

\begin{theorem} \label{results-variationalfeexplicit}
Fix $k, d = O(1)$, and $\epsilon, \delta > 0$ and an inverse temperature $\beta > 0$, and let $H$ be a $k$-Local Hamiltonian on $n$ qudits of local dimension $d$ and $m$ bounded strength interactions. Then, there exists an algorithm that runs in time $2^{\tilde{O}(\epsilon^{-2})}\cdot \tilde{O}(n^k \log 1/\delta)$, that returns an estimate to the free energy accurate up to an additive error of $\epsilon n^{k/2}m^{1/2} + \delta n /\beta$ and is correct with high probability. 
\end{theorem}

\textbf{The Vertex Sample Complexity}

The Regularity Lemma \ref{lemma-hregularity} enables us to derive an insightful structural statement for Local Hamiltonians. Namely, the definition of a `vertex sample complexity' for Local Hamiltonians of bounded interaction strengths, in an analogy to the vertex sample complexity of Max-kCSPs of \cite{Alon2002RandomSA} and \cite{Andersson2002PropertyTF}. They showed that the restriction of any Max-kCSP to a uniformly random sample of $\text{poly}(1/\epsilon)$ variables, sufficed to estimate the maximum number of satisfiable clauses up to an additive error of $\epsilon\cdot n^k$. We develop a generalization of this result to Quantum Local Hamiltonians, by combining the Hamiltonian regularity lemma with some extensions to the proof techniques by \cite{Alon2002RandomSA} to SDPs.

\begin{theorem} \label{results-vsc}
    Fix $d, k=O(1)$ and $\epsilon > 0$, and let $H$ be a $k$-local Hamiltonian on $n$ qudits of local dimension $d$ and $m$ bounded interaction strengths. Let $Q\subset [n]$ be a uniformly random sample of $q = \Omega(\epsilon^{-6} \log 1/\epsilon)$ of those qudits, and let $H_Q$ be the sum of interactions with support contained entirely in $Q$. Then, with probability $0.99$,
    
    \begin{equation}
        \bigg|\min_{\rho}\text{Tr}[H\rho]-\frac{n^k}{q^k}\min_{\rho_Q}\text{Tr}[H_Q\rho_Q]\bigg|\leq \epsilon \cdot n^k
    \end{equation}
    
\end{theorem}

We rely crucially on the guarantee of product state approximations to Quantum Local Hamiltonians in this regime of additive error. Indeed, one of the directions of the statement above is quite intuitive for both classical and quantum systems: If the ground state energy of $H$ is low, then the ground state energy of the restriction $H_Q$ can't be much higher than the estimate. This is since the reduced density matrix $\rho_Q = \text{Tr}_{V\setminus Q}[\psi]$ of the ground state $\psi$ of $H$, probably also has low energy $\text{Tr}[H_Q\rho_Q]\approx \frac{q^2}{n^2}\cdot \text{Tr}[H\psi]$, and the true ground state energy of $H_Q$ can only be lower than that. 

In the converse, however, lies an interesting `semi-classical' characterization of this additive error regime. Note that if the ground state energy of $H$ is `high', then in particular there doesn't exist any product states with low energy on $H$. Using the proof techniques in \cite{Alon2002RandomSA}, we show this implies the existence of a certain succinctly describable classical certificate to this product-state `infeasibility', which we sample from to prove the absence of product states with low energy on $H_Q$. Here is where we require the asymmetric product state approximations of Theorem \ref{results-psa}: for sufficiently large $Q$, the absence of low energy product states for $H_Q$ must imply a high ground state energy for $H_Q$. In this sense, the ground state energy of $H_Q$ can't be much lower than its estimate either. 

As a straightforward corollary to this structural result, now we can easily devise an algorithm which provides an additive error guarantee by exactly diagonalizing the Hamiltonian $H_Q$ on $q = \tilde{O}(\epsilon^{-6})$ vertices in time $2^{\tilde{O}(1/\epsilon^6)}$. However, we can in fact do slightly better, simply by applying the additive error, product state approximation algorithm by \cite{Gharibian2011ApproximationAF} to the subsample:

\begin{corollary}
    Fix $d, k=O(1)$ and $\epsilon > 0$, and let $H$ be a $k$-Local Hamiltonian on $n$ qudits of local dimension $d$ and $m$ bounded interaction strengths. There exists a randomized algorithm which runs in time $2^{\tilde{O}(\epsilon^{-2})}$, and with probability $.99$ outputs an estimate to the ground state energy accurate up  to an additive error of $\epsilon\cdot n^k$.
\end{corollary}

Aside from the improved dependence on $k$ in the exponent, this result may seem to only subtly differ from that in Theorem \ref{results-gsptasdense}. However, we emphasize that Theorem \ref{results-gsptasdense} requires an exponential number in $1/\epsilon$ of samples of vertices, whereas Theorem \ref{results-vsc} guarantees a polynomial number suffices. \\

\textbf{Approximation Schemes on Graphs that exclude a Fixed Minor}

Finally, for our last main contribution, we develop novel singly-exponential time algorithms for sparse, $2$-Local Hamiltonians defined on graphs that exclude a fixed minor. Formally, the family of $h$-minor free graphs are all the graphs $G$ that can not produce another (smaller) graph $h$, by deleting edges and vertices and by contracting edges \cite{Robertson1986GraphMI}. Planar graphs, and bounded genus graphs (such as toriods) are among the interesting special cases of these classes. Our approach builds on previous work by \cite{Bansal2009ClassicalAS} and \cite{Brando2013ProductstateAT} on planar graphs, using more general combinatorial decompositions \cite{Demaine2005AlgorithmicGM} and improving on their `quantum-to-classical' mappings. We show how such $2$-Local Hamiltonians can be approximately understood as classical Max $k$-CSPs defined on the high degree vertices in the graph, and develop a dynamic programming algorithm to solve it using a simple hyper-dimensional version of a tree decomposition. Our first result for these systems is a classical algorithm to approximate the ground state energy in time singly exponential in $\text{poly}(1/\epsilon)$,

\begin{theorem} 
Fix $\epsilon > 0$. Let $H$ be a 2-Local Hamiltonian defined on $n$ qubits and $m = \Theta(n)$ bounded strength interactions of norm $ < 1$, configured on an $h$-minor free graph $G = (V, E)$ where the minor is constant size $|h| = O(1)$. Then, we can approximate the ground state energy of $H$ up to additive error $\epsilon\cdot n$, in time $\text{poly}(n) + n\cdot 2^{\text{poly}(1/\epsilon)}$.
\end{theorem}

We build on these ideas by combining them with our information-theoretic techniques for the free energy of quantum systems, to construct novel algorithms for the free energy of these classes of sparse graphs at low temperatures as well. 

\begin{theorem} 
Fix $\epsilon > 0$ and an inverse temperature $\beta$. Let $H$ be a 2-Local Hamiltonian on $n$ qubits and $m = \Theta(n)$ bounded strength interactions of norm $ < 1$, configured on an $h$-minor free graph $G = (V, E)$ where the minor is constant size $|h| = O(1)$. Then, we can approximate the the free energy $F(\beta)$ of $H$ up to additive error $\epsilon\cdot n$, in time $\text{poly}(n)+ n\cdot \max(2, \beta^{-1})^{\text{poly}(1/\epsilon)}$, respectively.
\end{theorem}

\subsection{Related Work}
\label{subsection-related}

\textbf{Classical Approximation Schemes for QMA Complete Problems}

While the systematic study of approximation algorithms to QMA-Complete problems is still emerging, there are a number of works we would like to highlight on the topic. \cite{Bansal2009ClassicalAS} developed classical approximation schemes for ground state energies of classical and Quantum 2-Local Hamiltonians configured on planar graphs (of bounded degree, in the quantum case). They leveraged Baker's technique \cite{BakerBrenda1994ApproximationAF} and structural properties of planar graphs to approximately decompose the Hamiltonian into non-interacting partitions, which then could be analyzed by exact diagonalization, or  dynamic programming. \cite{Gharibian2011ApproximationAF} were among the first to construct an approximation algorithm for the $k$-Local Hamiltonian Problem. They argued that product states can provide a $d^{-k+1}$-relative factor approximations to the ground state energy of $k$-Local Hamiltonians defined on qu\textit{d}its, similarly to how Max Cut admits a $1/2$ multiplicative approximation. They then developed an approximation algorithm for the variational problem of finding the minimal energy product state of a given Local Hamiltonian $H$. It constructs a product state that provides an (extensive) $\epsilon \cdot n^k$ additive approximation to the ground state energy, in runtime $n^{O(\epsilon^{-2}\log 1/\epsilon)}$. Their approach was based on an adaptation of a classical technique, the “exhaustive sampling method” by \cite{Arora1999PolynomialTA} to the quantum setting, developed in the context of approximating Max Cut on dense graphs. 

Later, \cite{Brando2013ProductstateAT} developed information-theoretic techniques to argue the existence of product state approximations to the the ground state energy. More precisely, they show that so long as $H$ is \textit{everywhere dense} ($\Omega(n^{k-1})$ minimum degree), has bounded expansion, or is clustered into regions of sub-volume law entanglement entropy, there exist product states that provide \textit{additive error} approximations to the minimum energy. Leveraging their information-theoretic statements, they turned the algorithm of \cite{Gharibian2011ApproximationAF} into a PTAS for the ground state energy, albeit only meaningful when the number of interactions $m = \Omega(n^k)$. Additionally, they devise approximation schemes for Quantum Hamiltonians defined on generic planar graphs (not just those of bounded degree), solving an open problem posed by \cite{Bansal2009ClassicalAS}. Their key insight was what we refer to as a `high-low degree' technique, in which one could consider a product state over all vertices of degree larger than some tunable cutoff $\Delta$, and a generic (entangled) quantum state over the hilbert space of the low-degree particles, while incurring only a small error to the ground state energy. It is worthwhile to raise however, that the runtime of the resulting algorithm is triply-exponential in $1/\epsilon$, where the algorithm returns an $\epsilon \cdot n$ additive approximation. 

More recently, in the context of \textit{relative error} approximation schemes, \cite{Harrow2015ExtremalEO} showed that one can find a product state within a relative error of $\sqrt{l}$ of the ground state of a traceless $k$-Local Hamiltonian of bounded norm, where $l$ is the maximum degree of the underlying hyper-graph. \cite{Bravyi2019ApproximationAF} devised a $O(\log n)$ multiplicative approximation scheme to the ground state energy of 2-Local traceless Hamiltonians by rounding the solutions of SDPs to product states. \\

\textbf{Classical Approximation Schemes for the Free Energy of Quantum Systems}

Our results also contribute to a rich literature of classical
techniques for thermal quantum systems. Perhaps the most well known of these techniques are the Quantum Monte Carlo methods, which approximate the quantum partition function of a quantum system to that of a classical spin system, which in turn is approximated via Markov chain Monte Carlo methods. Despite the enormous practical success of these techniques, rigorous proofs of convergence have only been presented in certain restricted systems  \cite{Bravyi2017PolynomialtimeCS,Bergamaschi2020SimulatedQA, Crosson2021RapidMO}, and they generically are efficient only in the high temperature limit. Another high-temperature technique is the polynomial interpolation method \cite{Barvinok2016CombinatoricsAC, Harrow2020ClassicalAC}, based on a Taylor expansion of the partition function in the high temperature limit. Although both of these approaches are only provably efficient either on restricted classes of systems (such as substochastic Hamiltonians) and/or in the high temperature limit (typically $\beta$ is a constant, or at most $O(\log n)$), they provide quite strong notions of approximation. In fact, they generally provide $(1+\epsilon)$ multiplicative approximations to the partition function (which translates to an $\epsilon$ additive approximation to the free energy), while in this paper we only attempt \textit{extensive}, additive, $\epsilon \cdot m$ error approximations to the free energy.

By approaching the problem via this weaker notion of error, it is possible to devise approximation schemes in a much wider range of temperatures. A recent result by \cite{Bravyi2021OnTC} presented an algorithm that estimates the free energy of dense Local Hamiltonians, also building on the information-theoretic techniques by \cite{Brando2013ProductstateAT}. Their approach is based on a quantum generalization to a classical correlation rounding approach by \cite{Risteski2016HowTC}, and their algorithm finds a $\epsilon \cdot n^2$ additive approximation to the free energy of 2-Local Hamiltonians, in runtime $n^{O(\epsilon^{-2})}$. \\

\textbf{Comparison to Previous Work}

To conclude our introduction we summarize our algorithmic improvements in constrast to previous known constructions for the quantum systems studied. In table \ref{tableone} below we label the Hamiltonians, and runtime and accuracy guarantees of the additive error approximation schemes in previous work for the systems we consider. In table \ref{tabletwo}, we present our results for these same systems. 

For simplicity, unless otherwise stated we concern ourselves with Quantum Local Hamiltonians of bounded interaction strengths $\|H_e\|_\infty \leq 1$ on $n$ qubits and $m$ interactions. In both tables, we refer to a `low threshold rank' Hamiltonian as having constant $\epsilon$-SOS threshold rank of its interaction graph. With the exception of the recent work by \cite{Bravyi2021OnTC}, all the results in table \ref{tableone} concern ground state energy approximation schemes.

\begin{table}[h!]
\centering
\begin{tabular}{||c c c c||} 

 \hline
Result & Context & Accuracy & Runtime \\ 
 \hline \hline
\cite{Gharibian2011ApproximationAF} & $k$-local Hamiltonians & $\epsilon\cdot n^k$ & $n^{\tilde{O}(\epsilon^{-2})}$ \\

\hline
\multirow{2}{*}{\cite{Brando2013ProductstateAT}}
  & Low Threshold Rank & \multirow{2}{*}{$\epsilon\cdot \sum_{e\in E} \|H_e\|_\infty$ }   & \multirow{2}{*}{$n^{O(\epsilon^{-1})}$}\\
 & Hamiltonians & & \\
 \hline
  \multirow{2}{*}{\cite{Bravyi2021OnTC}} & Free Energy of  & \multirow{2}{*}{ $\epsilon \cdot n^2 + \delta \cdot n/\beta$} &
  \multirow{2}{*}{$n^{\tilde{O}(\epsilon^{-2})} \cdot O(\log 1/\delta)$} \\
  & 2-local Hamiltonians & &\\
 \hline
 \multirow{2}{*}{\cite{Bansal2009ClassicalAS}}  & Planar Graphs & $\epsilon\cdot \sum_{e\in E} \|H_e\|_\infty$   & $n^{O(1)}\cdot 2^{2^{\text{poly}(\Delta, \epsilon^{-1})}}$ \\ 
& of bounded degree $\Delta$ & & \\

 \hline
 \cite{Brando2013ProductstateAT}
  & Planar Graphs & $\epsilon\cdot \sum_{e\in E} \|H_e\|_\infty$   & $n^{O(1)}\cdot 2^{2^{2^{\text{poly}(\epsilon^{-1})}}}$ \\
   \hline

\end{tabular}

\caption{A summary of previous algorithms for related Quantum Systems}
\label{tableone}

\end{table}

\begin{table}[h!]
\centering
\begin{tabular}{||c c c c||} 

 \hline
System & Context & Accuracy & Runtime \\ 
 \hline \hline
  
 \multirow{2}{*}{ $k$-local Hamiltonians} & G.S. Energy & $\epsilon \cdot n^k$ & $2^{\text{poly}(\epsilon^{-1})}$ \\
 & Free Energy & $\epsilon \cdot n^k + \delta \cdot n/\beta$ & $2^{\text{poly}(\epsilon^{-1})}\cdot O(\delta^{-2})$  \\

 \hline

Low Threshold Rank  & Maximum   & \multirow{2}{*}{ $\epsilon \cdot m + O(n^{1/3}m^{2/3})$ }  & \multirow{2}{*}{ $(n/\epsilon)^{O(1)}+2^{\tilde{O}(1/\epsilon^2)}$ }  \\
 Quantum Max Cut  & Eigenvalue & &\\

 \hline
 
 $h$-Minor Free Graphs & G.S. Energy & \multirow{2}{*}{$\epsilon \cdot n$} & \multirow{2}{*}{$n^{O(1)} + n\cdot 2^{\text{poly}(\Delta, \epsilon^{-1})}$} \\ 
 of bounded degree $\Delta$ & Free Energy & & \\
 
 \hline
 
 \multirow{2}{*}{ $h$-Minor Free Graphs}& G.S. Energy & \multirow{2}{*}{$\epsilon \cdot n$} & $n^{O(1)} + n\cdot 2^{\text{poly}(\epsilon^{-1})}$ \\ 
  & Free Energy & &$n^{O(1)} + n\cdot \max(2, \beta^{-1})^{\text{poly}(\epsilon^{-1})}$  \\
  
  \hline
 
\end{tabular}

\caption{The main algorithms in this work}
\label{tabletwo}

\end{table}

\subsection{Organization}
We organize the rest of the paper as follows. Before diving into the proofs, in section \ref{section-discussion} we present a discussion on open problems raised in this work and future problems we hope our techniques could be useful to. 

To keep our information-theoretic arguments self-contained, we begin in section \ref{section-existence} by presenting our existence statements, and our extensions to the techniques by \cite{Brando2013ProductstateAT}. We proceed in section \ref{section-regularity} by presenting a simple discussion on our Hamiltonian regularity lemma, key to our additive error approximation schemes. In section \ref{section-gsptas-dense}, we apply our regularity lemma to define additive error approximation schemes for the 2-Local Hamiltonian problem. We defer the generalizations to $k$-local systems to appendix \ref{section-regularityextensions}, and to 2-Local Hamiltonians on low threshold graphs to appendix \ref{section-threshold}. Finally, in section \ref{section-vsc}, we present our sample complexity bounds. 

For conciseness, we've deferred to the appendix all our algorithms for sparse graphs, and our algorithms for the free energy. In appendix \ref{section-feptas-dense}, we extend our additive approximation schemes for the ground state energy, to approximating the free energy of $k$-Local Hamiltonians. In appendix \ref{section-sparsegsPTAS}, we present our algorithms for the ground state energy of Hamiltonians on graphs that exclude a fixed minor. Finally, in appendix \ref{section-sparseFEPTAS}, we lift those techniques to the free energy of graphs that exclude a fixed minor as well.

\section{Discussion}
\label{section-discussion}

We conclude this work by raising some open problems. The first of which is a curious gap between the quality of the mean field approximation to classical and Quantum Local Hamiltonians. To contrast our results to those in the classical setting, \cite{Borgs2012ConvergentSO, Basak2015UniversalityOT, Jain2018TheMA, Jain2019MeanfieldAC} studied the quality of the mean-field approximation to classical spin glass models with generic interaction matrices. The work of \cite{Jain2019MeanfieldAC} culminated in the result that the mean-field approximation is within an additive error of $O(n^{2/3}m^{1/3})$ of the free energy, a strictly better dependence on the number of interactions than our upper bound, $O(n^{1/3}m^{2/3})$. As both these results have roots in the information-theoretic techniques by \cite{Raghavendra2012ApproximatingCW}, it seems intriguing to ask whether there is some deeper structure. A possible direction would be to combine the regularity insights with the correlation rounding techniques, as in \cite{Jain2019MeanfieldAC}. However, there remain certain technical obstacles to approaching the free energy of quantum systems with the regularity lemma, namely analyzing the matrix exponential of the cut decomposition $H_D$.

Another interesting problem is to improve the weak regularity results for `low threshold rank' Hamiltonians, discussed in section \ref{section-threshold}. While we are able to devise approximation schemes based on graph regularity for a range of Hamiltonians whose interaction graphs have low threshold rank, we are unable to provide an actual construction of an approximate Hamiltonian $H'$. It would also be interesting to see whether the coarsest partition technique could be lifted to be applied to more general low threshold rank Hamiltonians, as opposed to relying on the high degree of symmetry of the Quantum Max Cut.

Finally, while the focus of this paper is on product-state approximations, the author considers it to be an outstanding open problem whether one can devise entangled ansatz's for classical approximations schemes to quantum problems. For examples, see \cite{King2022AnIA, Anshu2020BeyondPS}, who devised low-depth quantum circuits which perform slightly better than the best product state on certain Hamiltonians.

% \newpage
\section{On The Existence of Product State Approximations}
\label{section-existence}

In this section, we discuss the existence of separable and product states that approximate the energy and the free energy of Local Hamiltonians that are somewhat `dense'. The approach follows that of \cite{Brando2013ProductstateAT}, in analysing the structure of $n$ qudit quantum states once you condition on, or measure, a small random subset of the qudits. They showed that for any given starting state $\rho$, a certain `entanglement-breaking' mapping could be used to construct a separable state $\sigma$, defined over single-qudit marginals of $\rho$ conditioned on measurement outcomes. They prove that under this mapping, on average $\sigma$ approximates the few-body marginals of $\rho$, which is then used to ensure that $\sigma$ approximately preserves the energy of $\rho$ on sufficiently dense Hamiltonians. 

We note that a similar mapping was used by \cite{Bravyi2019ApproximationAF} to prove a stronger version of a theorem by Lieb \cite{Lieb1973TheCL}, and a recent paper by \cite{Bravyi2021OnTC} uses this idea to ensure that certain pseudo-distributions resulting from convex programming solvers can be accurately rounded to product states which approximate the free energy of dense ($m = \Omega(n^2)$) 2-Local Hamiltonians. Here, we reason that there exist extensive error, product state approximations to the free energy up to a much smaller $m = \omega(n)$ number of edges, and extend the results by \cite{Brando2013ProductstateAT} to this sparse graph regime without assuming that the graph  is everywhere-dense (i.e has a minimum degree). To prove these results, we reason about the entropy of the separable states produced by their `entanglement-breaking' mapping, and its product state components.

We organize the rest of this section as follows: In subsection \ref{subsection-notation}, we introduce notation, discuss the entanglement breaking procedure, and the structure of the separable state $\sigma$. In subsection \ref{subsection-psageneral}, we prove an extension to the results of \cite{Brando2013ProductstateAT} on general graphs, by showing how the separable state $\sigma$ approximately preserve the energy of $\rho$. In subsection \ref{subsection-entropy}, we discuss how the state $\sigma$ is a mixture of states $\eta^{C, b}$, which have higher entropy than $\rho$. Finally, we conclude in subsection \ref{subsection-fe} by showing how to carefully extract a product state from these ensembles which approximates the free energy. 

\subsection{The Entanglement-Breaking Mapping}
\label{subsection-notation}

Before beginning, let us discuss the mapping and setup some notation. Let $\Lambda$ be a single-qubit measurement channel of low \textit{distortion}, meaning that we have an upper bound on $\|\gamma\|_1/\|\Lambda(\gamma)\|_1$ for all traceless operators $\gamma$. While tighter results in terms of the local dimension $d$ and the Hamiltonian locality $k$ are possible by carefully choosing $\Lambda$, for our purposes we follow the choice by \cite{Bravyi2019ApproximationAF,Bravyi2021OnTC} and let $\Lambda$ be a measurement in a random Pauli basis:

\begin{gather}
    \Lambda(\rho) = \frac{1}{3} \sum_{b\in [3], r\in [-1, 1]} |b, r\rangle\langle b, r| \cdot  \langle \psi_{b, r}| \rho |\psi_{b, r}\rangle, \text{ and }\\
    \mathcal{N}(\rho) = \frac{1}{3} \sum_{b\in [3], r\in [-1, 1]} |\psi_{b, r}\rangle\langle \psi_{b, r}| \cdot  \langle \psi_{b, r}| \rho |\psi_{b, r}\rangle, 
\end{gather}

\noindent where $\mathcal{N}$ corresponds to a measurement and averaging over the measurement outcome. In the above, we denote as $|\psi_{b, 1}\rangle, |\psi_{b, -1}\rangle$ the eigenstates of the Pauli matrix $\sigma_b$,  $b\in \{1, 2,3\}$, with eigenvalues $1, -1$. When applied to qudits of local dimension $d = 2^{d'}$, we simply apply the channel $\Lambda^{\otimes \log d}$. 

If $\rho$ is a generic density matrix on a set $V$ of $n$ qubits, 
 \cite{Brando2013ProductstateAT} define its separable state approximation $\sigma$ as follows. Fixed an integer parameter $l$, they first sample a small integer $m \in  [l]$, and then randomly sample a subset $C\subset V$ of $m$ particles to measure. Then, the channel $\mathcal{N}^{\otimes m}_C\otimes \mathbb{I}_{V\setminus C}$ is applied on the qubits in $C$. To shorten notation, we denote the expectation over the size $m$ of $C$ and the choice of $C$ as $\mathbb{E}_{0\leq m\leq l} \mathbb{E}_{C\subset [n], |C| = m}\equiv \mathbb{E}_C$; and, for $x = (b, r)$ with $b\in \{1, 2, 3\}^m$ and $r \in \{-1, 1\}^m$ denoting a $m$ different bases and $m$ measurement outcomes, we write $|\psi_x\rangle = \otimes_{i=1}^m |\psi_{b_i, r_i}\rangle \text{ and }\psi_x \equiv |\psi_x\rangle \langle \psi_x|$. Let $\tau$ denote the density matrix corresponding to the ensemble over measurement outcomes of this process, which can be written as 

 \begin{equation}
     (\mathcal{N}^{\otimes m}_C\otimes \mathbb{I}_{V\setminus C})(\rho) = \tau =  \mathbb{E}_{C, x = (b, r)} \psi_{x}^C \otimes \tau^{(C, x)}, 
 \end{equation}

\noindent with $\tau^{(C, x)} = \langle  \psi_{b, r}^C|\rho |\psi_{b, r}^C\rangle$ the density matrix on the unmeasured qubits $V\setminus C$, consistent with outcomes $x$ on the particles $C$. Finally, to define $\sigma$, each $\tau^{(C, x)}$ is `broken' into the tensor product of its marginals $\tau^{(C, b, r)}_u = \text{Tr}_{(V\setminus C)\setminus\{u\}}[\tau^{(C, b, r)}] $. The resulting separable state $\sigma$ can be expressed as an ensemble over choices of $C$ and the measurement basis and outcomes $x = (b, r)$:

\begin{equation}
    \sigma \equiv \mathbb{E}_{C, x=(b, r)} \psi_{x}^C \bigotimes_{u\notin C}\tau^{(C, b, r)}_u \equiv \mathbb{E}_{C, x=(b, r)} \eta^{(C, b, r)} \equiv \mathbb{E}_{C, b} \eta_{C, b}
\end{equation}

\noindent where in the RHS above we introduce $\eta^{(C, b, r)} = \psi_{(b,r)}^C \bigotimes_{u\notin C}\tau^{(C, b, r)}_u $ and $\eta_{C, b} = \mathbb{E}_r \eta^{(C, b, r)}$. We emphasize that this last step isn't a physical operation - it is just a mathematical mapping to define an un-entangled state. 

We raise two interesting properties about the marginals of these states: since $\tau^{(C, x)}$ is supported on the unmeasured particles, its expectation over measurement outcomes $r$ recovers the original reduced density matrix: $\rho^{V\setminus C} = \mathbb{E}_r \tau^{(C, b, r)}$ for all $C, b$. Next, consider the separable states $\eta_{C, b}$, resulting from fixing the set of measured particles $C$ and the basis $b$, and averaging over the measurement outcomes $r$. If a given particle $u\notin C$ is not measured, then its (marginal) single particle reduced density matrix remains the same as that of the original $\rho$:

\begin{equation}
   \eta^{C, b}_u = \text{Tr}_{V\setminus u}[\eta_{C, b}] = \mathbb{E}_r \tau^{(C, b, r)}_u =  \text{Tr}_{V\setminus (C\cup \{u\})}[\mathbb{E}_{r}\tau^{(C,b, r)}] = \text{Tr}_{V\setminus u}[\rho] =\rho_u 
\end{equation}

\subsection{Product State Approximations on General Graphs}
\label{subsection-psageneral}

In this subsection, we present a simple extension to the results on product states by \cite{Brando2013ProductstateAT} to more general interaction hyper-graphs.

\begin{theorem} \label{theorem-BHgeneral}
Fix $d,  k=O(1)$ and let $H=\sum_{e\in E} H_{e}$ be a $k$-Local Hamiltonian on qu\textit{d}its, where we define $J$ to be the $k$ dimensional array of interaction strengths $J_{u_1\cdots u_k} = \|H_{u_1\cdots u_k}\|_{\infty}$, for all $u_1\cdots u_k\in V$. Let $\rho$ be a generic density matrix on $n$-qubits, and pick a positive integer $l < n$. Then there exists a globally separable state $\sigma$ such that
\begin{equation}
    \big|\text{Tr}[H(\rho-\sigma)]\big| \leq \delta_l \equiv  O\bigg( \frac{l}{n} \cdot |J|_1 +  \frac{n^{k/2}}{\sqrt{l}} \cdot \|J\|_F \bigg)
\end{equation}

\noindent with $\|J\|_F^2 = |J|_2^2 = \sum_e J_{e}^2$ and $|J|_1 = \sum_e |J_{e}|$.

\end{theorem}

We refer to the error factor $\delta_l$ above frequently in the analysis in this section. It corresponds to the error to the energy of the entanglement-breaking state $\sigma$, when the set of particles measured $C$ is of maximum size $l$. For constant $d, k$, and if $H$ has $m$ interactions of strength bounded by 1, then $\min_l \delta_l = O(n^{\frac{k-1}{3}}m^{\frac{2}{3}})$. In particular, this is $O(n^{1/3}m^{2/3})$ for 2-Local Hamiltonians, or $O(nD^{2/3})$ on $D$ regular graphs.  Indeed, this shows there exists a product state approximation to any ground state of $H$ with extensive error $\min_l \delta_l \leq \epsilon \cdot m$ so long as $m =  \Omega(n^{k-1}/\epsilon^{3})$, providing the notion of `dense' we require for hypergraphs.

Our proof strategy follows the techniques by \cite{Brando2013ProductstateAT}. First, we reduce this energy difference to an expression of quantum correlations between $k$ unmeasured quantum particles. By measuring these particles with a quantum channel of low distortion, we upper bound these quantum correlations via the classical correlations between the random variables resulting from the measurement. In turn, we analyze these classical correlations using extensions to the self-decoupling lemmas by \cite{Brando2013ProductstateAT}.

\begin{proof} 

By the triangle inequality, Holder's inequality, and Jensen's inequality in sequence, we upper bound the energy difference in terms of an expectation over the choice of measured particles $C$:
\begin{gather}
     \big|\text{Tr}[H(\rho-\sigma)]\big| \leq \sum_{e} \|H_{e}\|_\infty \cdot \|\rho^{e} - \sigma^{e}\|_1 \leq \mathbb{E}_{C}\bigg(\sum_{e} \|H_{e}\|_\infty \cdot \|\rho^{e} - \eta^{e}_{C}\|_1\bigg)
 \end{gather}  

This enables us to divide into cases on whether the summands $e=(u_1\cdots u_k)$ were measured (i.e. $\in C$). In particular, if any $u_i$ is in $C$, then simply upper bound the distance by $2$. 
\begin{gather}
    \mathbb{E}_{C}\bigg(\sum_{e} \|H_{e}\|_\infty \cdot \|\rho^{e} - \eta^{e}_{C}\|_1\bigg) =\\= \mathbb{E}_{C} \bigg(\sum_{e:e \cap C = \emptyset} \|H_{e}\|_\infty \|\rho^{e} - \eta^{e}_{C}\|_1 + \sum_{e:e \cap C \neq \emptyset} \|H_{e}\|_\infty \|\rho^{e} - \eta^{e}_{C}\|_1\bigg) \leq\\
    \leq 2\cdot k \cdot l/n \cdot |J|_1 + \mathbb{E}_{C}\bigg(\sum_{e:e \cap C = \emptyset} \|H_{e}\|_\infty \|\rho^{e} - \eta^{e}_{C}\|_1\bigg)
\end{gather}

\noindent where in the last line we use linearity of expectation, and that the probability either $e \cap C \neq \emptyset$ is $\leq k\cdot l/n$ (union bound), and finally $|J|_1 \equiv \sum_e \|H_e\|_\infty$. 

Now, let us turn our attention to the remaining term in the RHS above. Since $e = (u_1\cdots u_k)\notin C$ are not measured, we know there is a particular structure to their states as expectations over measurement outcomes: $\eta_C^{e} =\mathbb{E}_{x = (b, r)} \bigotimes_i^k \tau^{(C, b, r)}_{u_i}$, and $\rho^{e} = \mathbb{E}_{x = (b, r)} \tau^{(C, b, r)}_{u_1\cdots u_k}$. In this setting, we can use Jensen's inequality to extract the expectation over $x$, followed by the Cauchy-Schwartz inequality, and then once again Jensen's inequality:

\begin{gather}
    \mathbb{E}_{C}\bigg(\sum_{e \cap C = \emptyset} \|H_{e}\|_\infty \|\rho^{e} - \eta^{e}_{C}\|_1\bigg) \leq \mathbb{E}_{C, x} \bigg(\sum_{e \cap C = \emptyset} \|H_{e}\|_\infty \|\tau^{(C, x)}_{e} - \bigotimes_i^k \tau^{(C, b, r)}_{u_i}\|_1\bigg) \\
    \leq \bigg(\sum_{e} \|H_{e}\|_\infty^2\bigg)^{1/2} \cdot   \mathbb{E}_{C, x}\bigg( \sum_{e \cap C = \emptyset} \|\tau^{(C, x)}_{e} - \bigotimes_i^k \tau^{(C, b, r)}_{u_i}\|_1^2\bigg)^{1/2} \\
    \leq \|J\|_F \cdot  \bigg( \mathbb{E}_{C, x}\sum_{e \cap C = \emptyset} \|\tau^{(C, x)}_{e} - \bigotimes_i^k \tau^{(C, b, r)}_{u_i}\|_1^2\bigg)^{1/2}
\end{gather}

 In this manner, we have reduced the problem to a question about the k-particle quantum correlations of un-measured particles. Following the proof techniques of \cite{Brando2013ProductstateAT}, to analyze them we first consider measuring these states using the Pauli channel $\Lambda$ described previously. In particular, we measure each qudit with the quantum-classical channel $\Lambda' = \Lambda^{\otimes \log d}$. Let us denote as $p_{X_1\cdots X_n} = (\Lambda')^{\otimes n}(\rho)$ the classical output distribution, where each $X_i$ corresponds to a basis $b \in \{X,Y, Z\}^{\log d}$ and the outcomes $r\in \{0, 1\}^{\log d}$. We observe that for any subset $S\subset V\setminus C$ and partial information about $C$, $x_C = (b_C, r_C)$, the distribution of measurement outcomes $x_S$ conditioned on $x_C$ is given by the measurement outcomes of $\tau^{(C, x_C)}_{S}$, indeed:
 
 \begin{equation}
     p_{X_S}\big|_{X_C = x_C} = \Lambda^{\otimes S}(\tau^{(C, x_C)}_{S})
 \end{equation}
 
 Now, we can use a claim by \cite{Bravyi2021OnTC} on the distortion of the Pauli channel $\Lambda$ to relate the corresponding quantum and classical correlations
 
 \begin{claim}
 [Claim 2, \cite{Bravyi2021OnTC}] For any integer $z \geq 1$ and any $z$ qubit Hermitian operator $Q$, we have $\|\Lambda^{\otimes z}(Q)\|_1\geq 6^{-z} \|Q\|_1$.
 \end{claim}
 
 We make the observation that $(\Lambda')^{\otimes k}$ acts on $z=k\log d$ qubits to apply the claim above and obtain
 
\begin{gather}
   \bigg \|\tau^{(C, x)}_{e} - \bigotimes_i^k \tau^{(C, b, r)}_{u_i}\bigg\|_1 \leq d^{3k} \cdot \bigg\| p_{X_e}\big|_{X_C = x} - \prod_{i\in [k]} p_{X_{u_i}}\big|_{X_C = x} \bigg\|_1. 
\end{gather}

\noindent We proceed by applying Pinsker's Inequality for multi-partite classical distributions on the distribution $p_{x_S}$ to express

\begin{gather}
    \mathbb{E}_{C, x} \sum_{e \cap C = \emptyset} \|\tau^{(C, x)}_{e} - \bigotimes_i^k \tau^{(C, b, r)}_{u_i}\|_1^2  
    \leq 2\cdot d^{6k} \cdot \mathbb{E}_{C, x} \sum_{e \notin C} I(X_{u_1}: X_{u_2}\cdots :X_{u_k}|X_C= x) \leq \\
    \leq 2\cdot d^{6k} \cdot n^k\cdot  \mathbb{E}_C \mathbb{E}_{u_1\cdots u_k\notin C}  I(X_{u_1}: X_{u_2}\cdots :X_{u_k}|X_C) 
\end{gather}

To conclude the proof, we apply a self-decoupling lemma, which is a generalization of Lemma 19 in \cite{Brando2013ProductstateAT} to the multi-partite case. For conciseness, we present the proof of which in the appendix.

\begin{lemma}[\ref{lemma-kselfdecoupling}] Let $X_1\cdots X_n$ be classical random variables with some arbitrary joint distribution, and fix integers $k, l < n$. Then

\begin{equation}
    \mathbb{E}_{0\leq m \leq l}\mathbb{E}_{\substack{C\subset [n],  \\ |C| = m}} \mathbb{E}_{ \substack{ u_1\cdots u_k\in V\setminus C \\ u_i\neq u_j}} I(X_{u_1}:\cdots :X_{u_k}| X_C) \leq \frac{k^2}{l}\mathbb{E}_{u}I(X_u: X_{V\setminus \{u\}})
\end{equation}

\end{lemma}

Since each random variable $X_u = (b_u, r_u)$ is supported on a set of size $6^{\log d}$, $I(X_u: X_{V\setminus \{u\}}) \leq \log d \cdot \log 6$, and we conclude 

\begin{equation}
     \big|\text{Tr}[H(\rho-\sigma)]\big|  \leq \frac{2kl}{n} \cdot |J|_1 +  \frac{6kd^{3k}\log d}{\sqrt{l}} \cdot n^{k/2}\|J\|_F
\end{equation}

\end{proof}

\subsection{On the Entropy of the Entanglement-Breaking Mapping}
\label{subsection-entropy}

While originally applied to the context of approximating ground state energies, the results of the previous subsection above work for any state $\rho$ on $n$ particles, not necessarily the ground state. In particular, we can apply Theorem \ref{theorem-BHgeneral} to the Gibbs state, guaranteeing that there exists a separable state $\sigma$ that is close to the Gibbs state in energy. To guarantee that the separable state $\sigma$ is indeed also close in \textit{free energy}, we use the variational characterization of the free energy, and a characterization of the entropy of the entanglement breaking procedure.

\begin{theorem}
\label{theorem-entropy} Let $\rho$ be a generic density matrix on $n$ qudits. Define $\sigma = \mathbb{E}_{C,b} \eta_{C, b}$ to be the separable state approximation to $\rho$, as defined in Theorem \ref{theorem-BHgeneral}. Then, 

\begin{equation}
    S(\rho) \leq S(\eta_{C, b}), \text{ for all } C, b
\end{equation}

\end{theorem}

 We approach the proof of this theorem in three parts. First, we use the chain rule of the entropy, and the fact that conditioning never increases information, to upper bound the entropy of $\rho$ in terms of the entropy of the measured set $C$ and the entropy of the other vertices conditioned on $C$. That is:
\begin{equation}
    S(\rho) = S(C)_\rho + \sum_{u\notin C} S(u| C, \text{ all previous }v< u) \leq S(C)_\rho + \sum_{u\notin C} S(u| C)_\rho
\end{equation}

This effectively decouples the non-measured particles. Next, we argue that in fact both of the terms that arise above are upper bounded by their counter-parts in $\eta_{C, b}$. We do so in two lemmas:

\begin{lemma}
The entropy of the measured set of particles $C$ can only increase in $\eta_{C, b}$, 

\begin{equation}
    S(C)_\rho \leq S(C)_{\eta_{C, b}}
\end{equation}

\end{lemma}

\begin{proof}
Let us consider the result of the Brandao-Harrow mapping, if the original state $\rho$ were the maximally mixed state on $C$, $\mathbb{I}/2^{|C|}$. The result of the measurement in a fixed basis $b$, is naturally again maximally mixed. Let the CPTP channel representing this measurement be $\mathcal{N}_b$. By the data-processing inequality, 
\begin{equation}
    S(\rho^C|| \mathbb{I}/2^{|C|}) \geq  S(\mathcal{N}_b(\rho^C)|| \mathcal{N}_b(\mathbb{I}/2^{|C|})) = S(\eta_{C, b}^C||\mathbb{I}/2^{|C|} )
\end{equation}
Since $S(\gamma || \mathbb{I}/2^{|C|}) = |C| - S(\gamma)\forall \gamma$, then we obtain the inequality $S(C)_\rho \leq S(C)_{\eta_{C, b}}$.
\end{proof}

Let us now consider the conditional entropy of a particle $u$ that wasn't measured. This proof is based on a discussion in \cite{Bravyi2021OnTC} on pseudo-density matrix rounding, and follows from another application of a data-processing inequality.

\begin{lemma}
The conditional entropy of an un-measured particle $u\notin C$ does not decrease in $\eta_{C, b}$, that is
\begin{equation}
    S(u|C)_\rho \leq S(u|C)_{\eta_{C, b}}
\end{equation}

\end{lemma}

\begin{proof}
We have that the reduced density matrix on $C \cup \{u\}$ of $\eta_{C, b}$ is the output of the CPTP map $(\mathbb{I}_u\otimes \mathcal{N}^C_b)$ on $\rho$, and since CPTP maps do not increase the mutual information, 

\begin{gather}
    I(u: C)_\rho \geq  I(u: C)_{\eta_{C, b}}\iff S(u)_\rho - S(u|C)_\rho \geq S(u)_{\eta_{C, b}} - S(u|C)_{\eta_{C, b}}  \\ \Rightarrow S(u|C)_{\eta_{C, b}}\geq S(u|C)_\rho 
\end{gather}

\noindent where we used the fact that the reduced density matrices of the unmeasured particles is the same $\rho^u = \eta_{C, b}^u$, and thereby have the same entropy.
\end{proof}

With these two lemmas, we conclude as well the proof of Theorem \ref{theorem-entropy}.

\subsection{The Existence of Free-Energy Approximations}
\label{subsection-fe}

We can now finally argue the existence of product states which approximate the free energy. We combine the previous statements that the separable state $\sigma$ produced by the entanglement breaking channel \cite{Brando2013ProductstateAT} approximates the energy of the actual Gibbs state (Theorem \ref{theorem-BHgeneral}), and has a higher entropy (Theorem \ref{theorem-entropy}), to argue that it serves as a good approximation to the free energy as well. We then leverage this separable state and a series of averaging arguments to prove our main result of this section, that there exists product states that approximate the free energy at all temperatures.

To begin, it is useful to recall the variational presentation of the free energy:
\begin{definition}
The free energy $F$ is the minimum of the following objective:
\begin{equation}
    F\equiv \min_{\rho \geq 0: \|\rho\|=1} f(\rho)\equiv  \min_{\rho \geq 0: \|\rho\|=1} \text{Tr}[H\rho] - S(\rho)/\beta
\end{equation}
and the minimum is attained when $\rho$ is the Gibbs State, $\rho\propto e^{-\beta H}$.
\end{definition}
For simplicity, for now let us represent the objective above by including the temperature term within the Hamiltonian and the free energy, that is,
\begin{equation}
    F \leftarrow \beta F \text{ and } H\leftarrow \beta H.
\end{equation}

What the variational characterization immediately tells us is that any state $\sigma$ gives us an upper bound to the free energy: $F\leq f(\sigma) \forall \sigma$. Let us now turn to the discussion on the lower bound.

\begin{theorem} \label{theorem-feseparable}
Let $H$ be a $2$-Local Hamiltonian on $n$ particles.  Define an integer parameter $1\leq l\leq n$. Correspondingly, define the error parameter $\delta_l$ as in Theorem \ref{theorem-BHgeneral}. Then there exists a separable quantum-classical state $\sigma$ whose free-energy $f(\sigma)$ satisfies
\begin{equation}
    f(\sigma) \geq F\geq f(\sigma) - \delta_l 
\end{equation}
\end{theorem}

\begin{proof}
If $\rho = \argmin_\gamma f(\gamma) = \argmin_\gamma \text{Tr}[H\gamma] - S(\gamma)$ is the Gibbs state, then let $\sigma =\mathbb{E}_{C, b}\eta_{C, b}$ be the state produced by applying the entanglement breaking map to the state $\rho$. Theorem \ref{theorem-BHgeneral} guarantees that their energies are close, and therefore one can lower bound the free energy via the energy of $\sigma$:

\begin{equation}
   F =  \text{Tr}[H\rho] - S(\rho) = \text{Tr}[H(\rho-\sigma)] + \text{Tr}[H\sigma] - S(\rho) \geq -  \delta_l + \text{Tr}[H\sigma] - S(\rho)
\end{equation}

\noindent and Theorem \ref{theorem-entropy} ensures their entropy doesn't decrease: $S(\rho)\leq \mathbb{E}_{C, b} S(\eta_{C, b})$, thus

\begin{equation}
    F \geq  -\delta_l + \text{Tr}[H\sigma] - \mathbb{E}_{C, b}[S(\eta_{C, b})]  = - \delta_l + \mathbb{E}_{C, b} f(\eta_{C, b})
\end{equation}

This is, up to some additive error, on average the free energy of the states $\eta_{C, b}$ in the ensemble $\sigma$ lower bounds the actual free energy. By an averaging argument, there exists $C^*$, $b^*$ which is better than the expectation, and we conclude that there exists a quantum-classical, separable state $\eta_{C^*, b^*}$ that approximates the free energy.
\end{proof}

In fact, we can actually find a single product of mixed states that approximates the free energy. Let us consider the structure of the separable state $\eta_{C^*, b^*}$ that minimizes the expectation above. 
\begin{equation}
    \eta_{C, b}  = \sum_r p_{C, b}(r)  \psi_{(b, r)}^C \otimes \eta^{(C, b, r)} = \sum_r p_{C, b}(r)   \psi_{(b, r)}^C \bigotimes_{u\notin C}\eta^{(C, b, r)}_u
\end{equation}

\noindent  Indeed, $\eta_{C^*, b^*}$ is a quantum-classical state, and an ensemble of product states, one for each measurement outcome $r \in \{0, 1\}^{|C|}$ of the (now fixed) measurements in the Pauli basis $b \in \{1, 2, 3\}^{|C|}$.

We claim that the free energy of these states can be expressed through an average over the free energy of the product states that compose it:

\begin{claim}\label{claim-festructure}
The free energy of the states $\eta_{C, b}$ can be expressed as an average over the free energy of product states, minus the entropy of the measured particles:

\begin{equation}
    f(\eta_{C, b}) = - S(C)_{\eta_{C, b}} +\mathbb{E}_r f(\psi_{(b, r)}^C \otimes \eta^{(C, b, r)}).
\end{equation}

\noindent Where the expectation $\mathbb{E}_r$ denotes the average over the distribution $p_{C, b}(r)$ of measurement outcomes.

\end{claim}

\begin{proof}
The entropy of a quantum-classical state is well known, and given by
\begin{gather}
    S(\eta_{C, b}) = S(C)_{\eta_{C, b}} + \mathbb{E}_r[S(\otimes_{u\notin C} \eta^{(C, b, r)}_u)]
\end{gather}

That is, it can be directly expressed as the entropy of the measured particles $C$, plus an average over the entropy of the quantum components, which in this case are a product state. Since the entropy of the pure state $\psi_{(b, r)}^C$ is 0, one can expand the free energy as above. 
\end{proof}

We can now use this expression to prove the following theorem on the existence of product state approximations to the free energy

\begin{theorem}\label{theorem-feproduct}

Let $H$ be a $2$-Local Hamiltonian on $n$ particles.  Define an integer parameter $1\leq l\leq n$. Correspondingly, define the error parameter $\delta_l$ as in Theorem \ref{theorem-BHgeneral}. Then there exists a product state $\sigma = \otimes_{u\in V} \sigma_u$ whose free-energy $f(\sigma)$ satisfies
\begin{equation}
    f(\sigma) \geq F\geq f(\sigma) - 2\delta_l 
\end{equation}

\end{theorem}

\begin{proof}
In Theorem \ref{theorem-feseparable} we proved that the entanglement-breaking state $\sigma = \mathbb{E}_{C, b}\eta_{C, b}$ of the Gibbs state of $H$ provides a good approximations to the free energy on average $\mathbb{E}_{C, b}f(\eta_{C, b}) \leq F + \delta_l$. Let us fix our attention to a pair $C, b$. From Claim \ref{claim-festructure}, we know the free energy $f(\eta_{C, b})$ of $\eta_{C, b}$ has a particular form:

\begin{equation}
    f(\eta_{C, b}) = - S(C)_{\eta_{C, b}} +\mathbb{E}_r f(\psi_{(b, r)}^C \otimes \eta^{(C, b, r)}).
\end{equation}

Let us now apply an averaging argument, to pick a measurement outcome $r^*_{C, b}$ for each pair $(C,b)$ s.t. 

\begin{gather}
    f(\psi_{(b, r^*)}^C \otimes \eta^{(C, b, r^*)})\leq \mathbb{E}_r f(\psi_{(b, r)}^C \otimes \eta^{(C, b, r)}) \text{ and thereby }\\
    f(\eta_{C, b}) \geq - S(C)_{\eta_{C, b}} - \sum_{u\notin C }S(\eta^{(C, b, r^*)}_u) + \text{Tr}[H \psi_{(b, r^*)}^C \otimes \eta^{(C, b, r^*)}]
\end{gather}

Let us now consider defining product states $\gamma_{C, b} = \mathbb{I}_C/d^{|C|}\otimes \eta^{(C, b, r^*)}$, that is, replacing the states of the measured particles $C$ by a maximally mixed state, and maintaining the states $\eta^{(C, b, r^*)}$ picked through the averaging argument. Our intention is to argue that the states $\gamma^{C, b}$ only increase the energy of $\psi_{(b, r^*)}^C \otimes \eta^{(C, b, r^*)}$ by a small amount, and don't decrease the entropy, thereby defining an approximate lower bound to $f(\eta_{C, b})$. We do so in two steps. First, since $\gamma^{C, b}_C$ is maximally mixed, the entropy of the particles in $C$ doesn't decrease, and the entropy of particles outside of $C$ doesn't change

\begin{equation}
    S(C)_{\eta_{C, b}} \leq S(C)_{\gamma^{C, b}} \text{ and } S(\gamma^{(C, b)}_u) = S(\eta^{(C, b, r^*)}_u) \text{ for }u\notin C
\end{equation}

Finally, the energy difference between the states $\gamma_{C, b}$ and $\psi_{(b, r^*)}^C \otimes \eta^{(C, b, r^*)}$ can be upper bounded by the interactions with $C$. Let $J_C = \sum_{e\in E: e\cap C} \|H_e\|_\infty$ be the sum of interaction strengths of edges that hit the set of measured particles $C$. We thereby have

\begin{equation}
    \bigg|\text{Tr}[H \psi_{(b, r^*)}^C \otimes \eta^{(C, b, r^*)}] - \text{Tr}[H (\mathbb{I}_C/d^{|C|}) \otimes \eta^{(C, b, r^*)}]\bigg| \leq   J_C
\end{equation}

\noindent and thus $f(\eta_{C, b})\geq f(\gamma_{C, b}) - J_C$. Note that over a uniformly random choice of $C$ of size at most $l$, $\mathbb{E}_C J_C = 2\frac{l}{n} |J|_1$, as previously discussed in the proof of Theorem \ref{theorem-BHgeneral}. Here $J$ corresponds to the matrix of interaction strengths, and $|J|_1$ is simply the sum of all interaction strengths. If we return to the setting of Theorem \ref{theorem-feseparable}, we have proved the existence of product states $\gamma_{C, b}$ such that

\begin{equation}
    \mathbb{E}_{C, b}f(\gamma_{C, b})\leq \mathbb{E}_{C, b}f(\eta_{C, b}) + 2\frac{l}{n} |J|_1\Rightarrow \mathbb{E}_{C, b}f(\gamma_{C, b})\leq F +\delta_l + 2\frac{l}{n} |J|_1 \leq F+2\delta_l
\end{equation}

\noindent by an averaging argument, we prove the theorem.

\end{proof}

% \newpage
\section{The Hamiltonian Regularity Lemma}
\label{section-regularity}

Let us begin by reviewing the cut decomposition of \cite{Frieze1999QuickAT}. The key intuition behind their result is the notion that dense graphs can be roughly viewed as a sum of complete bipartite sub-graphs between subsets of vertices in the graph. Each of these bipartite sub-graphs is essentially a `cut' in the graph, hence the name. 

\begin{definition}
Given two sets $S, T\subset [n]$ and a number $d\in \mathbb{R}$, the $n\times n$ cut matrix $D=$CUT$(S, T, d)$ is defined by $D_{u, v} = d \cdot \delta_{u\in S}\delta_{v\in T}$.
\end{definition}

\begin{definition}
A `cut decomposition' expresses a real matrix $J$ as the sum
\begin{equation}
    J = \sum_{k=0}^s D^{(k)} + W
\end{equation}
where each $D^{(k)}$ is a cut matrix defined on sets $R_k, L_k\subset [n]$, and of weight $d_k$. Such a decomposition is said to have width $s$, coefficient length $(\sum d_k^2)^{1/2}$, and error $\|W\|_{\infty\rightarrow 1}$.
\end{definition}

The main result of \cite{Frieze1999QuickAT} is precisely an algorithm to efficiently find such a decomposition:

\begin{theorem}
[\cite{Frieze1999QuickAT}\label{theorem-frieze}] Let $J$ be an arbitrary real matrix and fix a constant $\epsilon > 0$. Then there exists a cut decomposition of width $O(\epsilon^{-2})$, coefficient length $O(\|J\|_F/n)$, error at most $\epsilon n \|J\|_F$, and such that $\|W\|_F\leq \|J\|_F$. Moreover, with probability $1-\delta$ said decomposition can be found implicitly in time $2^{\tilde{O}(\epsilon^{-2})}/\delta^2$, and explicitly in time $\tilde{O}(n^2 / \epsilon^4) +2^{\tilde{O}(\epsilon^{-2})}/\delta^2 $.
\end{theorem}

\begin{remark}
The key point of the cut decomposition is that the number of cuts only depends on the quality of the approximation, not the size of the graph.
\end{remark}

Perhaps the main tool we introduce in this work is a generalization of this result to the quantum setting. We exploit the fact that quantum density matrices and quantum Hamiltonians can be expressed in a Pauli basis, to reduce the problem of decomposing Hamiltonians into that of a `multi-colored' cut decomposition. For simplicity, here we discuss the case of 2-Local Hamiltonians, on qudits of local dimension $d = 2^{d'}$ which is a power of 2, and defer further generalizations to the appendix. 

Let $H=\sum H_e$ be 2-local Hamiltonian defined on $n$ qudits, and define $\mathbb{P}_{\log d} = \{\mathbb{I}, X, Y, Z\}^{\otimes \log d}$ be the set of Pauli operators acting on a single qudit. Any operator $h_{u, v}$ acting on the Hilbert space of 2 qudits can be decomposed into basis of $\mathbb{P}_{\log d}\otimes \mathbb{P}_{\log d}$:

\begin{equation}
    H_{uv} = \sum_{i, j\in [d^2]} h_{uv}^{ij} \sigma^u_i\otimes \sigma^v_j
\end{equation}

Where the $h_{uv}^{ij}$ are all real coefficients. Group the coefficients of the interactions defined on the same Pauli matrices $i, j$ into an interaction matrix $J^{ij} = \{ h_{uv}^{ij}\}_{u, v}$, i.e., a matrix for each of $d^4$ `colors'. We note that this essentially defines $O(d^4)$ different weighted adjacency matrices. Now, let us apply the regularity lemma of \cite{Frieze1999QuickAT} on each of the colored interaction/adjacency matrices $J^{ij}$ above. By construction, for each pair $(i, j)$ one can express

\begin{equation}
    J^{ij} = \sum_{k=1}^s D^{ijk} + W^{ij} \equiv D^{ij}+W^{ij}
\end{equation}

Where $D^{ijk}$ are the $s$ cut matrices of the interaction $i, j\in [d^2]$, defined on partitions $\{R^{ijk}, L^{ijk}\}$ of the vertex set of the graph for $k\in [s]$. We can thereby define the cut decomposition $H_D$ of the Hamiltonian $H$ to be the edges of the $D^{ijk}$ \textit{crossing} any such cut:

\begin{equation}
   \text{The Hamiltonian Cut Decomposition: } H_D = \frac{1}{2} \sum_{i , j\in [d^2]\atop k\in [s]}\sum_{u \in R^{ijk} \atop v\neq u, v \in L^{ijk}} D^{ijk}_{uv} \sigma^u_i\otimes \sigma^v_j \otimes \mathbb{I}_{V\setminus \{u, v\}}
\end{equation}

\noindent where we appropriately order the tensor product such that $u< v$ and add a factor  of $1/2$ via a handshaking argument. More importantly, we filter out the diagonal entries $D^{ijk}_{uu}$, since the cuts $S, T$ returned by the cut decomposition in Theorem \ref{theorem-frieze} need not be disjoint, and Local Hamiltonians can't have `self-edges' in a basis decomposition. While unfortunately we no longer can interpret the interaction graph of $H_D$ as an exact sum of complete bipartite sub-Hamiltonians, fortunately, we will later recover this interpretation in an approximate sense.

We dedicate the rest of this section to proving two interesting properties of $H_D$. First, we argue that the energy of any product state $\rho = \otimes_{u\in V} \rho_u$ is close, whether in $H$ or $H_D$, arising from the combinatorial structure of the decomposition. Then, we leverage our product state approximation toolkit from section \ref{section-existence}, to argue that $H_D$ is in fact close to $H$ in the spectral norm $\|H-H_D\|_\infty$.

\begin{theorem} \label{theorem-psregularity}
Let $H=\sum_{u, v} H_{u, v}$ be a 2-Local Hamiltonian defined on qudits of local dimension $d = 2^{d'} = O(1)$, let $J_{uv} = \|H_{uv}\|_{\infty}$ be the matrix of interaction strengths, and let $H_D$ be the Hamiltonian cut decomposition of $H$ of width $s = O(\epsilon^{-2})$. Then, for all product states $\rho=\otimes_{u\in V} \rho_u$, 
\begin{equation}
    |\text{Tr}[(H-H_D)\rho]|\leq \epsilon n \|J\|_F
\end{equation}

 Moreover, with probability $1-\delta$ said decomposition can be found implicitly in time $2^{\tilde{O}(\epsilon^{-2})}/\delta^2$, and explicitly in time $\tilde{O}(n^2 / \epsilon^4) +2^{\tilde{O}(\epsilon^{-2})}/\delta^2 $
 
\end{theorem}

\begin{proof}

By restricting our attention to product states, we are able to essentially decouple the `colors' (different Pauli terms) in the Cut Decomposition. 

\begin{gather}
   \big |\text{Tr}[(H-H_D)\rho] \big | = \bigg|\sum_{u < v}\sum_{i, j}(h_{uv}^{ij} - D^{ij}_{uv}) \text{Tr}[\sigma^i_u\otimes \sigma^j_v\rho]\bigg| =\\
   =\bigg|\frac{1}{2}\sum_{i, j}\sum_{u\neq  v}W^{ij}_{uv} \text{Tr}[\sigma^i_u \rho_u]\text{Tr}[ \sigma^j_v\rho_v]\bigg| \leq \sum_{i, j} \bigg| \sum_{u\neq  v}W^{ij}_{uv} \text{Tr}[\sigma^i_u \rho_u]\text{Tr}[ \sigma^j_v\rho_v]\bigg| =\\
   = \sum_{i, j} \bigg| \sum_{v}\bigg(\sum_{u:u\neq  v} W^{ij}_{uv} \text{Tr}[\sigma^i_u \rho_u]\bigg)\text{Tr}[ \sigma^j_v\rho_v]\bigg| \leq \sum_{i, j}\sum_v \bigg| \sum_u W^{ij}_{uv} \text{Tr}[\sigma^i_u \rho_u]\bigg|  + \sum_{i, j} \sum_v |W^{ij}_{vv}| \leq \\
   \leq \sum_{ij}\bigg(\|W^{ij}\|_{\infty\rightarrow 1} + n \cdot \max _v |W^{ij}_{vv}|\bigg)
\end{gather}

where we re-introduced the diagonal terms to obtain the $\infty\rightarrow 1$ norm. From Theorem \ref{theorem-frieze} we can pick a width $s = O(d^8\epsilon^{-2}) = O(\epsilon^{-2})$ s.t. $\|W^{ij}\|_{\infty\rightarrow 1} \leq \epsilon n \|J^{ij}\|_F/d^4$. Finally, the Cauchy-Schwartz inequality tells us the diagonal entries are bounded: $|W^{ij}_{vv}| = |J^{ij}_{vv} - D^{ij}_{vv}| \leq \sum_{k} |d^{ijk}| \leq s^{1/2}\cdot (\sum d_k^{ij})^{1/2}\leq s^{1/2}\cdot \|J^{ij}\|_F/n$. The observation $\|J^{ij}\|_F\leq \|J\|_F$ and assuming $\epsilon^{-2}= o(n)$ concludes the proof.
\end{proof}

By combining the product state cut decomposition above with our results on product state approximations in section \ref{section-existence}, we can extend our results to entangled states as well. 

\begin{lemma} 
[The Hamiltonian Weak Regularity Lemma] \label{lemma-hregularity} In the context of Theorem \ref{theorem-psregularity}, $\|H-H_D\|\leq \epsilon \cdot n\|J\|_F$.
\end{lemma}

\begin{proof}

By Schatten norm duality, there exists a normalized state $\psi^*$ s.t.
\begin{equation}
    \|H-H_D\|_\infty = \max_\psi |\text{Tr}[(H-H_D)\psi]| = |\text{Tr}[(H-H_D)\psi^*]| 
\end{equation}

We now apply the product state approximation Theorem \ref{theorem-BHgeneral} on the state $\psi^*$ and Hamiltonian $H' = H-H_D$, to argue there exists a separable state $\sigma$ s.t. 

\begin{equation}
    |\text{Tr}[(H-H_D)(\psi^*-\sigma)]| \leq  \epsilon n \|J\|_F / 2
\end{equation}

\noindent where we observe that if $J'$ is the matrix of interaction strengths of $H'=H-H_D$, then $\|J'\|_1 \leq n \|J'\|_F$ (Cauchy-Schwartz) and $\|J'\|_F \leq \sum_{i, j\in [d^2]}\|W^{ij}\|_F \leq O(d^4 \|J\|_F)$ by means of a triangle inequality and the guarantees on $W$ in Theorem \ref{theorem-frieze}. Since $\sigma$ is separable, we can appropriately pick the width $s = O(\epsilon^{-2})$ in Theorem \ref{theorem-psregularity} to guarantee 

\begin{equation}
    |\text{Tr}[(H-H_D)\sigma]| \leq \epsilon n \|J\|_F/2
\end{equation}

\noindent and thereby via the triangle inequality:
\begin{equation}
     \|H-H_D\|_\infty\leq  |\text{Tr}[(H-H_D)(\psi^*-\sigma)]| +  |\text{Tr}[(H-H_D)\sigma]| \leq \epsilon n \|J\|_F
\end{equation}

\end{proof}

Using the existing technology of matrix regularity lemmas, in the appendix we present extensions to the result above for Local Hamiltonians defined on hyper-graphs and for graphs of low threshold rank. 

% \newpage

\section{A Ground State Energy PTAS}
\label{section-gseptas}

In this section, we discuss how to use the Hamiltonian cut decomposition (Theorem \ref{theorem-psregularity}) to construct additive error approximation schemes for the ground state energy. Our strategy will be to devise an algorithm to find the minimum energy among product states of the cut decomposition $H_D$ of $H$, which by our weak regularity result and the product state approximations in Theorems \ref{theorem-psregularity}, \ref{theorem-BHgeneral}, guarantee an additive error to the true ground state energy. Naturally, the resulting optimization program over product states is no longer convex, however, the low rank structure to the cut decomposition will enable us to construct a suitable convex relaxation. Our algorithms follow the ideas of \cite{Frieze1999QuickAT}, in establishing linear programming relaxations that give additive approximation schemes for Max Cut. 

Let $H$ be a be a $2$-Local Hamiltonian on $n$ qudits, where we assume the local dimension $d = O(1)$ to be a power of 2 for simplicity. Let $J_{uv} = \|H_{uv}\|_\infty$ be its matrix of interaction strengths. The main result of this section is in Theorem \ref{theorem-gsptasdense}.

\begin{theorem} \label{theorem-gsptasdense}
Fix $\epsilon > 0$. There exists an algorithm that runs in time $2^{\tilde{O}(1/\epsilon^2)}$ in the probe model of computation and estimates the ground state energy of $H$ up to additive error $\epsilon n \|J\|_F$ and is correct with probability $.99$. Alternatively, the estimate can be returned in time $\text{poly}(n, 1/\epsilon, \log 1/\delta) +  2^{\tilde{O}(1/\epsilon^2)}\cdot O(\log 1/\delta)$, together with a product state $|\psi\rangle = \otimes_u |\psi_u\rangle$ s.t 
\begin{equation}
    \langle \psi|H|\psi\rangle \leq \min_{\rho\geq 0, \|\rho\|_1=1}\text{Tr}[H\rho] + \epsilon n \|J\|_F,
\end{equation}

\noindent and the energy estimate is correct with probability $1-\delta$.
\end{theorem}

The key idea in our approximation algorithm for the variational minimum energy is to leverage the regularity lemma to relax the non-convex optimization program over product states, to checking the feasibility of a set of semidefinite programs. At a high level, the cut decomposition allows us to approximately interpret (dense) local Hamiltonians as a sum of a constant number of Hamiltonians, defined on complete bipartite graphs with uniform edge strengths. As we restrict our attention to product states, the uniformity in $H_D$ ensures that its energy becomes a function of the \textit{average magnetization} (the expectations of single qudit pauli operators) on each half of these bipartite graphs. By constraining the average magnetizations of these subsets of vertices to lie within a range, via linear inequality constraints, we consequently have that any product state which is feasible for those constraints must have energy in a fixed range. In this manner, relaxing the non-convex optimization problem to checking the feasibility of a set of SDP's. Given the behavior of these constraints on single qubit observables, we coin them the `subset magnetization constraints', and dedicate section \ref{section-constraints} for an intuitive discussion on their formulation and guarantees. 

Crucially, to actually make these algorithms run in constant time in the probe model of computation, we use a common refinement technique well known in the matrix regularity literature and presented in subsection \ref{subsection-compression}. This enables us to compress the number of variables in our sets of convex constraints, resulting in convex programs whose size only depends on the intended accuracy. In section \ref{section-gsptas-dense} we discuss how to actually solve the resulting feasibility problems, and combine their approximation guarantees with our existence statements to prove Theorem \ref{theorem-gsptasdense}. 

\subsection{The Subset Magnetization Constraints}
\label{section-constraints}

Here, we apply the regularity lemma to construct a collection of convex sets over the description of product states over $n$ qudits. Each such convex set will correspond to the set of product states within a given energy range, or in more detail, partitions of the graph will have magnetizations within a specified range. Naturally, there will be no guarantee that each convex set is feasible, however, by construction at least one such set will be feasible and contain a ball of finite radius. This will later enable a linear relaxation to the the problems of finding the best product state for the ground state energy, or the free energy, that one can solve using Convex Programming algorithms. 

Given a density matrix $\rho = \otimes_{u\in V}\rho_u$ which is a tensor product of $n$ single qudit mixed states, we express each $\rho_u \in (\mathbb{C})^{d\times d}$ in a generalized Pauli basis decomposition. That is, we describe each such $\rho$ via $(d^2-1)\cdot n$ real coeficients $\alpha_{i}^v \in \mathbb{R}, i\in [d^2-1], v\in V$, where $d = O(1)$ is the local dimension and
\begin{equation}
    \rho_v = \frac{1}{d}\mathbb{I}_{ d\times d} + \frac{1}{d} \sum_{i\in [d^2]-1}\alpha_i^v \sigma_i, \alpha_i^v\equiv \text{Tr}[\sigma_i \rho_v], \rho_v\geq 0
\end{equation}

The constraint that these matrices are psd is equivalent to that of checking if the lowest eigenvalue of $\rho_v$ is non-negative. Indeed, following standard techniques in semidefinite programming we note that its lowest eigenvector $|\psi\rangle$ provides a separating hyperplane on the variables $\alpha$, and we can compute it in time $O(d^3)$.

\begin{equation}
    \rho_v \ngeq 0\iff \exists \psi \text{ s.t. } \langle \psi| \rho_v |\psi\rangle = 1/d + \sum_{i\in [d^2-1]}\alpha_i^v \cdot  \langle \psi|\sigma_i|\psi\rangle/d < 0
\end{equation}

The energy of the product state $\rho$ under the Hamiltonian cut decomposition $H_D$ of Theorem \ref{theorem-psregularity} can be expressed as

\begin{equation}
    \text{Tr}[H_D\rho] =  \frac{1}{2d^2} \sum_{i, j\in [d^2]}\sum_{k\in [s]} \sum_{u\neq v} D^{ijk}_{uv} \alpha^u_i \alpha^v_j \cdot \text{Tr}_{u, v}[(\sigma_i^u\otimes \sigma_j^v)^2] = 
    \frac{1}{2}\sum_{i, j\in [d^2]}\sum_{k\in [s]} \sum_{u\neq v} D^{ijk}_{uv} \alpha^u_i \alpha^v_j
\end{equation}

In the above we abuse notation and let $\alpha_0^u=1$. Recall that the $D^{ijk}$, $i,  j\in [d^2], k\in [s]$ are cut matrices, where $s = O(\epsilon^{-2})$ is the width of the decomposition. That is, $D^{ijk}$ has a constant interaction strength of $d^{ijk}$ on the edges that cross from $R^{ijk}$ to $L^{ijk}$. In this setting, if each side of the cut was disjoint, $R^{ijk}\cap L^{ijk} =\emptyset$ we could re-express the energy due to the $k$-th cut of the $i, j$ Pauli interaction as:

\begin{equation}
    \sum_{u, v} D^{ijk}_{uv} \alpha^u_i \alpha^v_j = d^{ijk} \bigg(\sum_{u\in R^{ijk}}\alpha^u_i\bigg)\bigg(\sum_{v\in L^{ijk}}\alpha^v_j\bigg) \equiv d^{ijk} r^{ijk} c^{ijk}
\end{equation}

\noindent where $r^{ijk}$ becomes the total magnetization in the $i$th `direction' of the particles in the partition $R^{ijk}\subset V$, and analogously for $c^{ijk}$. Instead, we have to make due with the approximation (as in the proof of Theorem \ref{theorem-psregularity})

\begin{equation}
  \bigg|  \text{Tr}[H_D\rho] - \sum_{i, j, k} d^{ijk} r^{ijk} c^{ijk} \bigg|\leq \sum_{i, j, k} \sum_{u\in [n]}  |d^{ijk}|  \leq d^4\cdot s^{1/2} \|J\|_F 
\end{equation}

Naturally, if we knew the coefficients $\alpha^*$ of the best product state $\rho^*$, constructing said coefficients $r^*, c^*$ and the corresponding energy would be trivial. The beauty in this method is that one can enumerate (or guess) over the range of $r, c$, with the guarantee that there is at least one assignment that is close to $r^*, c^*$. In particular, let us note that since each coeficient $|\alpha^v_i|\leq 1$, then $|r^{ijk}|\leq |R^{ijk}|\leq n$. Construct the ordered set 
\begin{equation}
    I_\gamma = \{ -\gamma l n, -\gamma (l-1)n, \cdots, 0, \cdots \gamma l n\}
\end{equation}
where $\gamma$ is some accuracy parameter, and $l = \lfloor \gamma^{-1} \rfloor$. We note $|I_\gamma| \leq 2/\gamma + 1$. Let us define the shorthand $r = \{r^{ijk}\}_{ijk}, c = \{c^{ijk}\}_{ijk}$ for the corresponding vectors of length $d^4s$. For each of $|I_\gamma|^{2\cdot d^4s} = 2^{O(1/\epsilon^2)}$ possible assignment to these vectors $r, c\in (I_\gamma)^{d^4s}$, we define the following set of constraints $C_{r, c, \gamma}$:
\begin{gather}
    \rho_{\alpha_u}\geq 0 \text{ for all particles }u\in V  \\
    r^{ijk} - \gamma n \leq \sum_{u\in R^{ijk}} \alpha^u_i \leq r^{ijk} + \gamma n \text{ and } \\
    c^{ijk} - \gamma n \leq \sum_{u\in L^{ijk}} \alpha^u_j \leq c^{ijk} + \gamma n \text{ for all }i, j\in [d^2], k\in [s]
\end{gather}

Let us add a few remarks on the structure of $C_{r, c, \gamma}$. The first constraint ensures that any feasible point must be a density matrix of a product state. The second and third constraints enforce the magnetizations of given subsets of vertices, to be around the 'guess'. In particular, the application of a simple lemma by \cite{Frieze1999QuickAT} guarantees that if $\alpha$ is a feasible point of $C_{r, c, \gamma}$, then its energy is close to its 'guess':

\begin{lemma}\label{lemma-feasibleconstraints}
Let $\alpha = \{\alpha^{u}_i\}_{u\in V, i\in [d^2-1]}$ be a feasible point to the constraints $C_{r, c,\gamma}$, where $r = \{r^{ijk}\}$, $c = \{c^{ijk}\}, i, j\in [d^2], k\in [s]$. Correspondingly, let $\rho = \otimes \rho_u,  \rho_u = \mathbb{I}/d+\alpha^u\cdot \sigma^u/d$ be the product state defined by $\alpha$. Then 
\begin{equation}
    \bigg|\text{Tr}[H_D\rho] - \sum_{ijk} d^{ijk}r^{ijk} c^{ijk} \bigg| \leq O(\gamma/\epsilon \cdot n  \|J\|_F)
\end{equation}
where $s = O(\epsilon^{-2})$ is the width of each cut-decomposition, and $\gamma$ is the accuracy parameter of the constraints.
\end{lemma}

\begin{proof}
We use the following lemma by \cite{Frieze1999QuickAT}

\begin{lemma}[\cite{Frieze1999QuickAT}\label{closenesslemma}] 
Let $J$ be a real matrix, and $D^1\cdots D^s, W$ be its cut decomposition of width $s$. Then given real numbers $r_i, c_i, r'_i, c'_i$ with $|r_i|, |r_i'|, |c_i|, |c'_i|\leq n$ and $|r_i-r'_i|, |c_i-c'_i|\leq \gamma n$ for each $i\in [s]$, then $\sum_i d_i |r_ic_i-r'_ic'_i| \leq 8(\sum_i d_i^2)^{1/2} \cdot n^2\gamma s^{1/2}$
\end{lemma}

\noindent which follows from the Cauchy-Schwartz inequality. Recall the bound of $(\sum_i d_i^2)^{1/2}\leq \|J\|_F/n$ on the coefficient length, and the width $s = O(\epsilon^{-2})$,  from Theorem \ref{theorem-frieze}. If $r', c'$ are the true average magnetizations of the subsets in the cut decomposition, then

\begin{gather}
    \bigg|\text{Tr}[H_D\rho] - \sum_{ijk} d^{ijk}r^{ijk} c^{ijk}\bigg| \leq \bigg|\text{Tr}[H_D\rho] - \sum_{ijk} d^{ijk}r'^{ijk} c'^{ijk}\bigg| +
    \sum_{ijk} |d^{ijk}|\cdot \big|r^{ijk} c^{ijk} - r'^{ijk} c'^{ijk} \bigg|  \\\leq O\bigg(\frac{\gamma d^4}{\epsilon} \cdot n\|J\|_F\bigg)  
\end{gather}

so long as $\gamma n /\epsilon  = \omega(1)$
\end{proof}

We pick, in particular, $\gamma = O(\epsilon/s^{1/2}) = O(\epsilon^2)$, s.t. the error in the above becomes $\epsilon n \|J\|_F$.

\subsubsection{Common Refinements}
\label{subsection-compression}

The high degree of symmetry in the cut decomposed Hamiltonian $H_D$ presents some key advantages. In this section, we discuss a known common refinement technique used throughout the literature on matrix regularity lemmas \cite{Frieze1999QuickAT, Alon2002RandomSA, Jain2018TheVS, Gharan2013ANR} that enables us to compress the set of constraints $C_{r, c, \gamma}$ on $n$ variables into a set of constraints $\tilde{C}_{r, c, \gamma}$ on $2^{O(s)} = 2^{O(\epsilon^{-2})}$ variables. 

Let us consider the partitions $\{R^{ijk}, L^{ijk}\}_{ijk}$ defined in the cut decomposition, and let $\mathcal{A}=\{\mathcal{A}_a\}_{a\in [A]}$ be their common refinement (or, coarsest partition). That is, the subsets $\mathcal{A}_a\subset V$ are a disjoint partition of the vertices $V$, such that each set $R^{ijk}, L^{ijk}$ is the exact union of subsets $\mathcal{A}_a$. Thereby, each $v\in \mathcal{A}_a$ is on the same side of every cut. We note that there are $2$ choices of a side on each of $d^4\cdot s$ partitions, and thereby the number of subsets is $A = 2^{d^4\cdot s}$. The key observation is that each $u\in \mathcal{A}_a$ is indistinguishable from the other $v\in \mathcal{A}_a$ within the constraints. Since the constraints are convex, one can thereby use the same density matrix $\rho_a = \frac{1}{d}\mathbb{I}+\frac{1}{d}\alpha_a\cdot \sigma^a$ for every qudit in $\mathcal{A}_a$. In this manner, we define a set of $(d^2-1)\cdot 2^{d^4\cdot s} = 2^{O(\epsilon^{-2})}$ variables, and correspondingly define the compressed set of constraints $\tilde{C}_{r, c, \gamma}$:

\begin{gather}
    \rho_{\alpha_a}\geq 0 \text{ for all }a\in [A]  \\
    r^{ijk} - \gamma n \leq \sum_{a: \mathcal{A}_a\subset  R^{ijk}} |\mathcal{A}_a|\cdot \alpha^a_i \leq r^{ijk} + \gamma n \text{ and } \\
    c^{ijk} - \gamma n \leq \sum_{a: \mathcal{A}_a\subset L^{ijk}} |\mathcal{A}_a|\cdot \alpha^a_j \leq c^{ijk} + \gamma n \text{ for all }i, j\in [d^2], k\in [s]
\end{gather}

To ensure that this compression is actually useful, we use a simple claim to argue that it preserves feasibility.

\begin{claim}\label{claim-compressionfeasible}
Fix `guess' vectors $r, c$. Then $C_{r, c, \gamma}$ is feasible $\iff $ $\tilde{C}_{r, c, \gamma}$ is feasible.  
\end{claim}

\begin{proof}
If $\tilde{C}_{r, c, \gamma}$ is feasible, then $C_{r, c, \gamma}$ is by definition. The converse is slightly trickier. Let $\alpha^u_i$ be a feasible point to $C_{r, c, \gamma}$. Let us consider grouping the vertices by their common refinements, and average their spins. That is, we define a point $\{\alpha^a_i\}_{i\in [d^2-1], a\in [A]}$ in the variable space of $\tilde{C}_{r, c}$, via

\begin{equation}
    \alpha^a_i = \frac{1}{|\mathcal{A}_a|} \sum_{u\in \mathcal{A}_a} \alpha^u_i.
\end{equation}

This point $\{\alpha^a_i\}_{i\in [d^2-1], a\in [A]}$ clearly still satisfies the magnetization constraints, due to symmetry. Moreover, the resulting matrix $\rho_{\alpha_a}$ is psd since it is a convex combination of psd matrices. In this manner we conclude the claim.
\end{proof}

Finally, we remark that for the purposes of devising sublinear time algorithms, we won't actually have access to the full description of each partition $\mathcal{A}_a$. Instead, 
 we will have to make do with knowledge of their connectivity to other partitions, and estimates $|\hat{\mathcal{A}}_a|$ of their sizes. We denote as $\hat{C}_{r, c, \gamma}$ the convex set of constraints corresponding to modifying $\tilde{C}_{r, c, \gamma}$ above with  $|\hat{\mathcal{A}}_a|$ instead of $|\mathcal{A}_a|$. Fortunately, a simple claim argues these constraints are robust to this sampling noise. 

 \begin{claim} \label{claim-approxfeas}
    Fix $\delta < \gamma A^{-1}/4$, and assume $\big| |\hat{\mathcal{A}}_a| - |\mathcal{A}_a|\big|\leq \delta \cdot n$ for all $a\in [A]$. Then, 
    \begin{enumerate}
        \item For every $n$ qudit product state $\rho  = \otimes_u \rho_u$, there exists a choice of $r,c\in (I_\gamma)^{d^4\cdot s}$ such that $\rho$ is feasible for $C_{r,c, \gamma}$, and $\hat{C}_{r, c, \gamma}$ is also feasible.
        \item $\tilde{C}_{r, c, \gamma}$ feasible  $\Rightarrow \hat{C}_{r, c, 2\gamma}$ feasible, and $\hat{C}_{r, c, 2\gamma}$ feasible $\Rightarrow\tilde{C}_{r, c, 2\gamma}$ feasible, for all $r, c$. 
    \end{enumerate}
 \end{claim}

\begin{proof}
    The  proof of (1) follows from the guarantee that for every product state $\rho$, there always exists a choice $r, c \in $ such that $\rho$ is bounded $\gamma / 2 \cdot n$ away from saturating the magnetization constraints. Then, (1) and (2) follows since the magnetization constraints are affine and $|\alpha|_\infty\leq 1$, such that each magnetization is perturbed by at most $A\cdot \delta \cdot n\leq\gamma \cdot n/4$. We note that this remains bounded away from saturating the magnetization constraints.
\end{proof}

In this manner, given a cut decomposition, it really simply suffices to estimate the sizes of the partitions in the decomposition to instantiate the optimization program. In appendix \ref{appendix-volume}, we present a discussion on the volume of the convex set of feasible solutions to these constraints, which later ensure our algorithms converge efficiently. To conclude this section, let us briefly describe the data-structure used by \cite{Frieze1999QuickAT} to implicitly store this common refinement. Given a single cut decomposition, they explicitly construct a decision tree where the refinements $a\in [A]$ are the leaves, and one can query in $\text{poly}(1/\epsilon)$ time which refinement $\mathcal{A}_a$ contains any vertex $v\in [n]$. For any cut $R^{ijk}, L^{ijk}$, one can use this data-structure to list all the partitions $a$ contained in each of $R^{ijk}, L^{ijk}$ in $2^{O(1/\epsilon^2)}$ time. 

\subsection{The Algorithm}
\label{section-gsptas-dense}

The description of our algorithm is as follows. Given a $2$-Local Hamiltonian, we first implicitly construct the Hamiltonian cut decomposition $H_D$ of Theorem \ref{theorem-psregularity}. Next, using the implicit data-structure described in section \ref{section-regularity}, we use standard sampling guarantees to ensure we estimate the sizes of the coarsest partitions $|\hat{\mathcal{A}}_a|$ up to an additive error of $\gamma/4A\cdot n$. Finally, we construct the convex programs $\hat{C}_{r, c, \gamma}$, for each $r, c \in (I_\gamma)^{d^4\cdot s}$ as defined above, and use the convex program solver of Theorem \ref{theorem-bertsimas} below to check if each of the $2^{O(1/\epsilon^2 \log 1/\epsilon)}$ constraints is feasible. If $\hat{F}_{\gamma}$ are the set of pairs $(r, c) \in (I_\gamma)^{d^4\cdot s}$ such that $\hat{C}_{r,c, \gamma}$ is feasible, then we output

\begin{equation}
   \hat{V}_\gamma \equiv  \min_{r, c\in \hat{F}_{\gamma}} \hat{V}_{r, c} \equiv \min_{r, c\in \hat{F}_{\gamma}} \sum_{i,j,k} d^{ijk} r^{ijk} c^{ijk},
\end{equation}

\noindent where the coefficients $d^{ijk}$ arise from the cut decomposition $H_D$. To prove Theorem \ref{theorem-gsptasdense}, we need to argue the correctness and runtime of our scheme. We begin by arguing correctness.

It will later be relevant to prove properties on the un-perturbed constraint set $\tilde{C}_{r, c, \gamma}$, and so in Lemma \ref{lemma-Vadditive1} we warm up our proof techniques as if we had perfect knowledge of the cut decomposition and the sizes of the coarsest partitions $\mathcal{A}_a$. Correspondingly, let $F_{\gamma}$ be the set of pairs $r, c \in (I_\gamma)^{d^4\cdot s}$ such that $\tilde{C}_{r, c, \gamma}$ is feasible and analogously define the energy estimate:

\begin{equation}
   V_\gamma \equiv  \min_{r, c\in F_{\gamma}} V_{r, c} \equiv \min_{r, c\in F_{\gamma}} \sum_{ijk} d^{ijk} r^{ijk} c^{ijk}.
\end{equation}

First, in Lemma \ref{lemma-Vadditive1} and Corollary \ref{corollary-Vadditive} below we prove that $V_\gamma$ is a close approximation to the minimum energy of the cut decomposition $\text{Tr}[H_D \rho]$ among product states $\rho$. The key intuition behind this claim stems from Lemma \ref{lemma-feasibleconstraints}: If a product state $\rho$ is feasible for a set of constraints $\tilde{C}_{r, c, \gamma}$, its actual energy must be close to the 'guess', which must imply $V_\gamma$ can't be much smaller than the variational minimum energy. Conversely, the minimum energy product state must be feasible  for some $C_{r, c,\gamma}$, implying $V_\gamma$ can't be much larger than the variational energy either. 

\begin{lemma} \label{lemma-Vadditive1}
$V_{\gamma}$ is an $O(\gamma/\epsilon \cdot  n \|J\|_F)$ additive approximation to $\min_{\rho = \otimes \rho_u}\text{Tr}[H_D \rho]$
\end{lemma}

\begin{proof}
We argue the two directions of the approximation inequality separately. To begin, let the product state minimizer of the energy of $H_D$ be $\rho^*$, that is
\begin{equation}
    \min_{\rho = \otimes \rho_u}\text{Tr}[H_D \rho] = \text{Tr}[H_D \rho^*], \big| \text{Tr}[H_D \rho^*] - \sum_{ijk} d^{ijk}r^{ijk}_* c^{ijk}_*\big|\leq O(\|J\|_F/\epsilon)
\end{equation}

\noindent where we let $r^{ijk}_*, c^{ijk}_*$ be the subset magnetizations corresponding of $\rho^*$ . By construction, there is a pair $r, c\in (I_\gamma)^{d^4\cdot s}$ such that $|r^{ijk}_*-r^{ijk}|\leq n\gamma/2$, and therefore $\rho^*$ is feasible for some $C_{r, c, \gamma}$. Therefore, $V_\gamma$ must be upper bounded by the energy of $C_{r, c, \gamma}$, which by Lemma \ref{lemma-feasibleconstraints}, is close to the energy of $\rho^*$:
\begin{equation}
    V_\gamma \leq  \text{Tr}[H_D \rho^*] +  O(\gamma/ \epsilon \cdot  n \|J\|_F  )
\end{equation}

Conversely, let $r, c$ be the pair that defines the minimum of $V_\gamma$, and via Claim \ref{claim-compressionfeasible} we know there exists a product state $\rho'$ feasible for $C_{r, c, \gamma}$. Then $V_\gamma$ is close to the energy of $\rho'$ via Lemma \ref{lemma-feasibleconstraints}, and the energy of $\rho'$ must be lower bounded by that of $\rho^*$

\begin{equation}
     V_\gamma \geq \text{Tr}[H_D\rho' ] - O(\gamma/\epsilon \cdot n \|J\|_F) \geq \text{Tr}[H_D \rho^*] - O(\gamma/\epsilon \cdot  n \cdot \|J\|_F).
\end{equation}
\end{proof}

\begin{corollary} \label{corollary-Vadditive}
$V_{O(\epsilon^2)}$ is an $\epsilon n \|J\|_F$ additive approximation to $\min_{\rho = \otimes \rho_u}\text{Tr}[H \rho]$
\end{corollary}

\begin{proof}
This follows from Lemma \ref{lemma-Vadditive1} above with $\gamma = O(\epsilon^2)$, and the product state regularity in Theorem \ref{theorem-psregularity}.
\end{proof}

To extend these claims to the case in which we perturb the constraints with noisy samples of the sizes of the partitions, we use the discussion in Claim \ref{claim-approxfeas}. 

\begin{lemma} \label{lemma-perturbedconstraints}
   Fix $\delta < \gamma A^{-1}/4$, and assume $\big| |\hat{\mathcal{A}}_a| - |\mathcal{A}_a|\big|\leq \delta \cdot n$ for all $a\in [A]$. Then, $\hat{V}_{O(\epsilon^2)}$ is an $\epsilon n \|J\|_F$ additive approximation to $\min_{\rho = \otimes \rho_u}\text{Tr}[H \rho]$.
\end{lemma}

\begin{proof}
    Once again let us address the two directions. By Claim \ref{claim-approxfeas} (1), the upper bound on $\hat{V}_\gamma$ is straightforward: If $\rho^*$ is the minimum energy product state of $H_D$, then there exists $r,c$ s.t. both $C_{r,c , \gamma}$ and $\hat{C}_{r,c , \gamma}$ are feasible. Then,  $\hat{V}_\gamma\leq  \text{Tr}[H_D \rho^*] +  O(\gamma/ \epsilon \cdot  n \|J\|_F  )$. 
    
    Conversely, if $r, c$ minimize $\hat{V}_\gamma$, then let $\rho_A$ be a product state on $A$ qudits, feasible for $\hat{C}_{r, c, \gamma}$. We note that by copying the product state assignment of a given partition $a\in [A]$ to all the qudits in $\mathcal{A}_a$, one can naturally define a product state $\rho$. Let $r', c'$ be the true average magnetizations of $\rho$ on the subsets $\{R^{ijk}, L^{ijk}\}$. As in the proof of Claim \ref{claim-approxfeas}, 

\begin{equation}
    |r^{ijk} - r'^{ijk}| \leq \gamma \cdot n + \bigg|\sum_{a:\mathcal{A}_a\subset R^{ijk}} \bigg(|\hat{\mathcal{A}}_a|- |\mathcal{A}_a|\bigg)\cdot \alpha_i^a\bigg| \leq (\gamma + A\cdot \delta)\cdot n \leq 2\gamma \cdot n,
\end{equation}

\noindent for all $i, j \in [d^2], k\in [s]$. By Lemma \ref{closenesslemma}, we conclude $\hat{V}_\gamma \geq \text{Tr}[H_D\rho^*] - O(\gamma/\epsilon n \|J\|_F)$. 
\end{proof}

\begin{claim} \label{claim-estimatesizes}
    Given an implicit description of the cut decomposition $H_D$ of $H$, one can find estimates $|\hat{\mathcal{A}}_a|$, $a\in [A]$ for the sizes of all the partitions in the coarsest partition of the cut decomposition in time $2^{O(1/\epsilon^2)}\cdot O(\gamma^{-2}\log 1/\delta)$, such that with probability $1-\delta$
    \begin{equation*}
        \big| |\hat{\mathcal{A}}_a| - |\mathcal{A}_a|\big|\leq \frac{\gamma}{4A}\cdot n \text{ for all } a\in [A]
    \end{equation*}
\end{claim}

\begin{proof}
    Let us recall that the number of partitions $A = 2^{O(1/\epsilon^2)}$, and that we can check which partition a given vertex $v$ is in, in time $\text{poly}(1/\epsilon)$. In this setting, sample $k$ vertices uniformly at random, and let $X^a_i$ be the indicator random variable set to 1 if $v\in \mathcal{A}_a$. The estimator $\frac{n}{k}\sum_i^k X^a_i$ is unbiased and via Hoeffding's inequality, 
    \begin{equation}
        \mathbb{P}\bigg[\big| \frac{n}{k}\sum_i^k X^a_i - |\mathcal{A}_a|\big| \geq  \frac{\gamma}{4A} \cdot n\bigg]\leq \exp\bigg[- \frac{\gamma^2}{16 A^2}k\bigg] \leq \frac{\delta}{A}
    \end{equation}
    If we pick $k = O(\frac{A^2}{\gamma^2}\log \frac{A}{\delta}) = 2^{O(1/\epsilon^2)}\cdot O(\gamma^{-2}\log 1/\delta)$. 
\end{proof}

Now that we have expressed the minimization over product states of the energy of $H$ as checking the feasibility of $|I_\gamma|^{d^4\cdot s}= 2^{O(\epsilon^{-2}\log 1/\epsilon)}$ convex programs $\hat{C}_{r, c, \gamma}$, we discuss how to solve them. We use a theorem by \cite{Bertsimas2002SolvingCP}, in the formulation of \cite{Bravyi2021OnTC}:

\begin{theorem}
[\cite{Bertsimas2002SolvingCP}\label{theorem-bertsimas}] Suppose $K\subset \mathbb{R}^m$ is a convex set, and $R, r \in \mathbb{R}$ and $y \in K$ are such that: $K$ is contained in the ball of radius $R$ centered at the origin, and, if $K$ is non-empty, $K$ contains the ball of radius $r$ centered at $y$. Assume $K$ has a separation oracle which is efficiently computable in time $T$. Then, with probability $1-2^{-\Omega(m)}$ we can compute a feasible point $x\in K$ using $\text{poly}(m, T)\cdot O(\log R/r)$ calls to the separation oracle. 
\end{theorem}

In our case, we are dealing with $m=2^{O(\epsilon^{-2})}$ real variables and $2^{O(\epsilon^{-2})}$ constraints. We can thereby compute whether a given product state parametrized by a vector $\alpha$ is feasible for $\hat{C}_{r, c, O(\epsilon^{-2})}$ in time $T = 2^{O(\epsilon^{-2})}$. If $\alpha$ is not feasible for $\hat{C}_{r, c, O(\epsilon^{-2})}$, then it violates either a PSD constraint or a magnetization constraint, both of which have associated hyperplane witnesses which serve as the separation oracle. In appendix \ref{appendix-volume} we prove we can choose $R  = 2^{O(\epsilon^{-2})}$ and $r =O(\gamma)$. Applying Theorem \ref{theorem-bertsimas}, we can find the feasible set of pairs $F_{\gamma}$ in time $2^{O(\epsilon^{-2})}$ for each pair, for a total time of $2^{O(\epsilon^{-2}\log 1/\epsilon)}$. Once we include the runtime of constructing the decomposition $H_D$ as per Theorem \ref{theorem-psregularity}, and of defining the constraints $\hat{C}_{r, c, O(\epsilon^{-2})}$ as in Claim \ref{claim-estimatesizes}, we arrive at our main result in Theorem \ref{theorem-gsptasdense}. Finally, to actually output a pure product state, we expand our implicit solution $\rho=\otimes \rho_u$ into an explicit product of $n$ qudit density matrices and use the method of conditional expectations \cite{Vazirani2003ApproximationA} to extract a pure state from each qudit.

\section{The Vertex Sample Complexity of the Ground State Energy}
\label{section-vsc}

In this section, we study the vertex sample complexity of the ground state energy of local Hamiltonians of bounded interaction strengths. We prove that sample of all the interactions between a constant number of particles suffices to estimate the ground state energy up to a constant factor times $n^2$. Formally, the main result of this section is 

\begin{theorem} \label{theorem-vertexsample}
    Let $H = \sum_e H_e$ be a 2-Local Hamiltonian on $n$ qudits, of local dimension $d=O(1)$ and of bounded interaction strengths $\|H_e\|\leq 1$. Pick $Q\subset [n]$ to be a uniformly random sample of $q = \Omega(\epsilon^{-6} \log 1/\epsilon)$ of those qudits. Let $H_Q$ be the sum of interactions with support contained entirely in $Q$. Then, with probability $0.99$,
    
    \begin{equation}
        \bigg|\min_{\rho}\text{Tr}[H\rho]-\frac{n^2}{q^2}\min_{\rho_Q}\text{Tr}[H_Q\rho_Q]\bigg|\leq \epsilon \cdot n^2
    \end{equation}
    
\end{theorem}

Let us overview the proof of this theorem, which we detail in subsection \ref{subsection-prooftheoremsample}. In order to prove that the sampled estimate for the ground state energy is accurate, we follow a sequence of reductions. We begin by arguing that the ground state energy is close, up to an additive error $\epsilon\cdot n^2$ to the un-entangled, minimum product state energy of the Hamiltonian using the product state approximations of Theorem \ref{theorem-BHgeneral}. This holds on both the original Hamiltonian, and for the ground state and variational energies of the sampled Hamiltonian, so long as the sample size is sufficiently large. Next, we reduce these variational (product state) problems to their corresponding versions on cut decomposed Hamiltonians. To do so, we rely crutially on a lemma by \cite{Alon2002RandomSA} on the cut norm of random sub-matrices of a matrix of bounded cut, to argue that the error of the restriction of the cut decomposition to the sample $Q$ is still an accurate cut decomposition. That is, with constant probability the energy of any product state on $q$ qubits is close, whether on the restriction of $H$, $H_Q$ or the restriction of $H_D$, $H_{D_Q}$. 

At this point, we have reduced the problem to that of the vertex sample complexity of the variational minimum energy of the Hamiltonian cut decomposition $H_D$. Here, we extensively leverage the machinery of the subset magnetization constraints $C_{r, c, \gamma}$ developed in the previous section \ref{section-constraints}. To proceed, we follow the proof techniques of \cite{Alon2002RandomSA} and \cite{Jain2019MeanfieldAC}, in studying the properties of randomized restrictions of convex programs. We begin by arguing the easy direction: that if the set of constraints $C_{r, c, \gamma}$ is feasible, then with constant probability, a sample of a feasible point of $C_{r, c, \gamma}$ is approximately feasible for the sampled constraints. Therefore, the variational minimum energy of the sampled constraints is less than that of the original cut decomposed Hamiltonian, with constant probability. The converse is more tricky: we use convex programming duality to argue that if $C_{r, c, \gamma}$ is in-feasible, then a set of lagrange multipliers acts as a witness to its in-feasibility, and we sample from this witness to argue the in-feasibility of the sampled constraints. This infeasibility relation implies that the variational energy of the sample is, with constant probability, not much less than the true variational energy of $H_D$.

\subsection{Proof of Theorem \ref{theorem-vertexsample}}
\label{subsection-prooftheoremsample}

Let us begin by relating the ground state energy with  variational (product state) minimum energy of the Hamiltonian, on both the original system and the sample. To do so, we use the asymetric extension to the product state approximations of \cite{Brando2013ProductstateAT}, presented in Theorem \ref{theorem-BHgeneral} and simplified here:

\begin{corollary} \label{corollary-bhdense}
Given a 2-Local Hamiltonian $H=\sum_e H_e$ on $n$ qudit particles and $m$ interactions of strength bounded by $\|H_e\|_\infty\leq 1$, there exists a product state $\sigma$ such that
\begin{equation}
    \text{Tr}[H\sigma]\leq \lambda_{min}(H) + O(n^{1/3}m^{2/3})
\end{equation}
\end{corollary}

Given a uniformly random sample of vertices $Q\subset [n]$ of size $q$, let $H_Q$ be the sub-Hamiltonian of the interactions of the induced sub-graph $G[Q]$. From the corollary and $m\leq n^2$ we see
\begin{gather}
   \min_{\text{product states } \sigma} \text{Tr}[H\sigma] - \min_{\sigma} \text{Tr}[H\sigma] \leq  n^{5/3} \text{ and } \\
   \frac{n^2}{q^2}\cdot \bigg(\min_{\text{product states } \sigma_Q} \text{Tr}_Q[H_Q\sigma_Q] - \min_{\sigma} \text{Tr}[H_Q\sigma_Q]\bigg) \leq  n^2/ q^{1/3}
\end{gather}

Now that we have turned our attention to product states, we can use the Hamiltonian cut decomposition of Theorem \ref{theorem-psregularity} to further simplify the variational problem.  However, we require a guarantee on how accurate the cut decomposition is when we sample a random $Q\times Q$ sub-matrix of it. \cite{Alon2002RandomSA} proved the following theorem on the cut norm of random sub matrices of the cut decomposition

\begin{theorem}[\cite{Alon2002RandomSA}\label{theorem-submatrixcutnorm}]
Let $W$ be an $n\times n$ matrix, with bounded vector norms $\|W\|_{\infty} = O(\epsilon^{-1}), \|W\|_{\infty\rightarrow 1}\leq \epsilon n^2$, and $\|W\|_{F}\leq O(n^2)$. Suppose $Q$ is a random subset of $[n]$ of size $q= \Omega(\delta^{-5}\epsilon^{-4}\log 1/\epsilon)$, and let $W_Q$ be the sub-matrix defined by the restriction of $W$ to $Q$. Then, with probability $1-\delta$, $\|W_Q\|_{\infty\rightarrow 1} \leq O(\epsilon/\sqrt{\delta}\cdot q^2)$, $\|W_Q\|_{F}\leq O(q^2/\sqrt{\delta})$.
\end{theorem}

We use the theorem above to prove a bound on the accuracy of the cut decomposition on the sampled sub-graph. If $H_D$ is the cut decomposition of Theorem \ref{theorem-psregularity}, let $H_{D_Q}$ be the interactions of $H_D$ of support contained entirely in $Q\subset V$.

\begin{lemma} \label{lemma-sampleregularity}
Let $H_{D_Q}$ be the sub-Hamiltonian of the decomposition $H_D$ of support only in the random set $Q\subset V$. Then with probability $1-\delta$ over the choice of $Q$, for every product state $\rho_Q$ on $Q$, 
\begin{equation}
   \bigg| \text{Tr}[(H_Q-H_{D_Q})\rho]\bigg|\leq O(\epsilon/\sqrt{\delta}\cdot q^2),
\end{equation}
so long as $q= \Omega(\delta^{-5}\epsilon^{-4}\log 1/\epsilon)$.
\end{lemma}

\begin{proof}
Consider the $O(d^4) = O(1)$ `error' matrices $W^{ij}$, for each (generalized) Pauli matrix pair $i, j\in [d^2]$ resulting from the cut decomposition of the interactions of $H_D$. Applying the result of Theorem \ref{theorem-submatrixcutnorm} \cite{Alon2002RandomSA} above, and a union bound over all $d^4 = O(1)$ matrices, we are guaranteed that with probability $1-\delta$, $\sum_{ij}\|W^{ij}_Q\|_{\infty\rightarrow 1}\leq O(\epsilon/\sqrt{\delta}\cdot q^2)$, so long as $q= \Omega(\delta^{-5}\epsilon^{-4}\log 1/\epsilon)$. Moreover, we still have $W^{ij}_{u, u}\leq O(1/\epsilon)$ as in Theorem \ref{theorem-psregularity} and the lemma follows.
\end{proof}

We emphasize that these first two steps allowed us to reduce the problem to an analysis of sub-sampling Classical CSPs: by the triangle inequality,

\begin{gather}
   \bigg| \min_\rho \text{Tr}[H\rho] - \frac{n^2}{q^2} \cdot \min_{\rho_Q} \text{Tr}[H_Q\rho_Q]\bigg|\leq \\ \leq \bigg| \min_{\text{product state }\rho } \text{Tr}[H_D\rho] - \frac{n^2}{q^2} \cdot \min_{\text{product state }\rho_Q } \text{Tr}[H_{D_Q}\rho_Q]\bigg| + O\big(\epsilon/\sqrt{\delta}\cdot n^2\big)
\end{gather}

That is, now we can reason directly on the cut-decomposed Hamiltonian $H_D$ and its sub sample $H_{D_Q}$. Here we can draw from the machinery of the subset magnetization constraints, that we developed in sections \ref{section-constraints} and \ref{section-gsptas-dense}. Recall the definition of the convex set of constraints $C_{r, c, \gamma}$, which corresponded to the set of density matrices of product states with `subset magnetizations' (a linear function) within a $\gamma n$ additive range around the constraint coefficients $r, c$. We proved that the minimum over the choices of $r, c$ that are feasible $C_{r, c, \gamma}$ provides a good estimate to the variational minimum energy:

\begin{equation}
    V_\gamma =  \min_{\substack{r, c \in (I_{n, \gamma})^{d^4\cdot s} \\ C_{r, c, \gamma} \text{ feasible}}} \sum_{ijk} d^{ijk} r^{ijk}c^{ijk} \text{ and } |V_\gamma - \min_{\text{product state }\rho } \text{Tr}[H_D\rho]| \leq O(\gamma/\epsilon \cdot n^2) \text{ (Lemma \ref{lemma-Vadditive1} })
\end{equation}

To provide an analogous bound for the sample, let us formally define its sub-program. We let $C_{r, c,  \gamma}(Q)$ be the subset magnetization constraints of the sample $Q$, given guesses $r, c \in (I_{q, \gamma})^{d^4\cdot s}\equiv (\{-q, -q+\gamma q, -q+2\gamma q , \cdots ,0, \cdots, q\} )^{d^4 \cdot s}$ and accuracy $\gamma$:

\begin{gather}
C_{r, c,  \gamma}(Q):
    \rho_{\alpha_u} = \frac{\mathbb{I}}{d} + \frac{\alpha_u\cdot \sigma}{d}\geq 0 \text{ for all particles }u\in Q  \\
    r^{ijk} - \gamma q \leq \sum_{u\in R^{ijk}} \alpha^u_i \leq r^{ijk} + \gamma q \text{ and } \\
    c^{ijk} - \gamma q \leq \sum_{u\in L^{ijk}} \alpha^u_j \leq c^{ijk} + \gamma q \text{ for all }i, j\in [d^2], k\in [s]
\end{gather}

We emphasize that this is simply the corresponding set of constraints to $H_{D_Q}$, the sub-Hamiltonian of the cut-decomposition. More importantly, $C_{r, c,  \gamma}(Q)$ should be viewed as the restriction of $C_{r n/q, cn/q,  \gamma}$ to $Q$, with scaled coefficients $rn/q, cn/q$. We can similarly define an estimate for the variational minimum energy of the sample:

\begin{equation}
     V_{ \gamma}(Q) = \min_{\substack{r, c \in (I_{q, \gamma})^{16s} \\ C_{r, c, \gamma}(Q) \text{ feasible}}}\sum_{ijk} d^{ijk} r^{ijk}c^{ijk}
\end{equation}

The following claim guaranteess that $V_\gamma(Q)$ is an accurate estimate to the variational minimum energy of the sample. 

\begin{claim} \label{claim-sampleVadditive}
$V_\gamma(Q)$ is an $O(\gamma/\epsilon \cdot q^2 )$ additive approximation to $\min_{\text{product state }\rho } \text{Tr}[H_{D_Q}\rho_Q]$.
\end{claim}

\begin{proof}
The proof follows analogously to that of Lemmas \ref{lemma-feasibleconstraints} and \ref{lemma-Vadditive1}, where we note that the `coefficient length' of $D_Q$ is still the same as $D$. We emphasize that this is independent of the sample $Q$, it just relies on the properties of the original cut decomposition. 
\end{proof}

In order to relate $V_\gamma(Q)$ and $V_\gamma$, we need to argue about the structure of the convex constraints $C_{r, c,  \gamma}$ and how they relate to their sub-samples, $C_{r', c',  \gamma'}(Q)$. We note that proving one of the directions is much easier: if the global constraints $C_{r, c,  \gamma}$ are feasible, then a sample from the coordinates of a feasible point is likely feasible for the constraints of the sample. To approach the converse, we follow the techniques of \cite{Alon2002RandomSA} and \cite{Jain2018TheVS} for affine constraints. We essentially need to argue that the absence of a good solution to the original problem, implies the absence of a good solution to the sampled problem. The approach exploits the existence of certain witnesses to the dual convex program of the constraint set $C_{r, c, \gamma}$, whenever it is infeasible. Let us begin by proving the easier direction: 

\begin{claim} \label{claim-sampleupperbound}
Assume $q = \Omega(\gamma^{-2}\log 1/(\delta \epsilon))$. Then with probability $1-\delta$, if $C_{r,c, \gamma}$ is feasible, then  $C_{r q/n, cq/n,  2\gamma}(Q)$ is feasible as well. 
\end{claim}

\begin{proof}
Consider the $r, c$ that maximizes $V_\gamma$, and let $\alpha = (\alpha^u_i)_{u\in V, i\in [d^2-1]}$ be a feasible point of $C_{r,c, \gamma}$. For a fixed sample $Q$, consider the point $\alpha_Q = (\alpha^u_i)_{u\in Q, i\in [d^2-1]}$ defined by the restriction of $\alpha$ to the sampled particles. We can define the magnetizations $r_Q^{ijk}, c_Q^{ijk}$  of the sample by $r_Q^{ijk} = \sum_{u\in R^{ijk} \cap Q}\alpha_u^i, c_Q^{ijk} = \sum_{u\in L^{ijk}\cap Q}\alpha_u^j$, where in expectation $\mathbb{E}_Q [ r_Q^{ijk}] = \frac{q}{n}r^{ijk}$. By the psd constraint we have $|\alpha_u^i|\leq 1$, so a Chernoff bound tells us that the sampled magnetizations aren't far from their expectation:

\begin{equation}
    \mathbb{P}[|r_Q^{ijk} - \frac{q}{n}\cdot r^{ijk} | \geq \gamma q] \leq \cdot e^{-\Omega(\gamma^2 q)}
\end{equation}

\noindent and by a union bound over the $d^4\cdot s = O(\epsilon^{-2})$ subsets $(i, j, k)$, we have with probability $1-\delta$ that all the magnetization estimates are accurate up to error $\gamma \cdot q$ , so long as $q = \Omega(\gamma^{-2} \log 1/(\epsilon\delta))$. We adequate for this sampling error by increasing the slack in the constraints: It follows that if $C_{r,c, \gamma}$ is feasible, then with probability $1-\delta$, $C_{r q/n, cq/n,  2\gamma}(Q)$ is feasible as well. 

\end{proof}

To argue the converse of the Claim \ref{claim-sampleupperbound} is very tricky. To do so, let us begin by expressing the constraints $C_{r, c, \gamma}$ in a more standard SDP formulation, which will be more convenient to write out the dual. First, WLOG reintroduce the constraints over the trace $\text{Tr}[\rho_{u}] = 1$ of the density matrices, and generically define matrix variables $\rho_{u}\geq 0, \rho_{u} \in \mathbb{C}^{d\times d}$ parameterized by their real-valued Pauli basis description. Next, consider each of the $2\cdot d^4\cdot s$ `regularity' constraints in $C_{r, c, \gamma}$, and observe that $p$th constraint (indexed by the pauli matrices $i, j$ and the cut $k$) can be cast into the matrix form $\sum \text{Tr}[A_{p, u} \rho_{\alpha_u}] - b_p \leq 0$. In this context, $A_{p, u} \in \{0, \pm \sigma^u_i, \pm \sigma^u_j\}$ are Pauli matrices (including identity) or 0, and $b_p\in \{\pm r^p +\gamma n, \pm c^p+\gamma n\}$. We express the program:

\begin{gather}
    \text{ minimize } 0, \text{ subject to } C_{r, c, \gamma}: \\
    \text{Tr}[\rho_{_u}] = 1, \rho_{u}\geq 0 \text{ for } u\in [n] \text{ and } \\ \sum_{p, u} \text{Tr}[A_{p, u} \rho_{u}] - b_p \leq 0 \text{ for each constraint }p \in [2d^4s]
\end{gather}

We denote the optima of this program as $O_{r, c, \gamma}$, which is set to $+\infty$ whenever $C_{r, c, \gamma}$ is infeasible. Given this set of constraints, we can define the Lagrangian $\mathcal{L}_{r, c, \gamma}(\rho, \lambda, \mu, Y)$, where we define the multipliers as follows: $Y\geq 0$ are a collection of $n$ psd matrices in $\mathbb{C}^{d\times d}$ associated to the PSD constraint on each $\rho_u$, $\lambda \in \mathbb{R}_+^{2d^4 s}$ are the positive valued multipliers on the regularity constraints, and $\mu\in\mathbb{R}^n$ are unconstrained variables associated to the trace constraints. We have:

\begin{gather}
   \mathcal{L}_{r, c, \gamma}(\rho, \lambda, \mu, Y)=  - \sum_{u\in [n]} \text{Tr}[Y_u \rho_u] + \sum_{u\in [n]} \mu_u \bigg(\text{Tr}[ \rho_u] -  1\bigg) + \sum_p \lambda_p \bigg(\sum_{u} \text{Tr}[A_{p, u} \rho_{u}] - b_p \bigg) \\ =
     -\sum_j \lambda_j b_j -\sum_u \mu_u  + \sum_{u \in [n]} \text{Tr}\bigg[\rho_u \bigg( \sum_p \lambda_p A_{p, u} - Y_u +\mu_u \mathbb{I}\bigg)\bigg]  
\end{gather}

Indeed, note that $\max_{Y, \lambda > 0, \mu} \mathcal{L}_{r, c, \gamma}(\rho, \lambda, \mu, Y) = 0$ if $\rho$ is feasible for $C_{r, c, \gamma}$ by picking $Y, \lambda = 0$, and can be made $+\infty$ otherwise. One can thereby construct the Lagrangian dual objective $\mathcal{D}_{r, c, \gamma}(\lambda, \mu, Y) = \min_{\rho} \mathcal{L}_{r, c, \gamma}(\rho, \lambda, \mu, Y)$, and its corresponding dual program, which we note to be equivalent to

\begin{gather}
   \max_{\lambda \geq  0, \mu} \mathcal{D}'_{r, c, \gamma}(\lambda, \mu) =  \max_{\lambda\geq 0, \mu} -\sum_j \lambda_j b_j -\sum_u \mu_u  \text{ subject to } \\ \mu_u \mathbb{I} + \sum_p \lambda_p A_{p, u} \geq 0 \forall u\in [n]
\end{gather}

We note in that this dual is trivially feasible, with $\mu, \lambda = 0$ and objective value 0. Moreover, it is actually \textit{strictly feasible}, as in particular one can pick small $\lambda > 0$ and sufficiently large $\mu > 0$ such that every PSD constraint above is strictly satisfied $>0$. In this manner, duality of SDPs tells us that if, by assumption, the primal $C_{r, c, \gamma}$ is infeasible, then the dual must be unbounded. In particular, there must exist a choice of $\mu, \lambda$ better than the trivial solution $(0,0)$, i.e. bounded away from 0, 

\begin{equation}
    \exists \mu, \lambda \text{ such that }  -\sum_j \lambda_j b_j -\sum_u \mu_u > 0
\end{equation}

In this manner, the tuple $(\mu, \lambda)$ acts as a witness to the infeasibility of $C_{r, c, \gamma}$. We sample from said witness to argue the infeasibility of the randomly restricted SDP:

\begin{claim}  \label{claim-samplelowerbound}
Assume $q =  \Omega(\gamma^{-2}\log 1/ \delta)$. Then with probability $1-\delta$ over the choice of $Q\subset [n]$, if $C_{r, c, \gamma}$ is infeasible, then $C_{rq/n, cq/n, \gamma/2}(Q)$ is infeasible as well
\end{claim}

\begin{proof}
To begin, we note that the PSD constraints in the dual enable a bound on the multipliers $\mu_u$. For fixed $\lambda$, the feasible choice of $\mu$ which maximizes the dual, is the minimum $\mu$ satisfying the PSD constraints (the dual objective is negative in $\mu$). Since the minimum $\mu_u$ satisfying the PSD constraints is the maximum eigenvalue of $- \sum_p \lambda_p A_{p, u}$, $\mu_u$ can always be chosen to be

\begin{gather}
    \mu_u \leq \| \sum_p \lambda_p A_{p, u}\|_\infty \leq \sum_p \lambda_p \|A_{p, u}\|_\infty \leq \|\lambda\|_1 \text{ and } \\ 
    \mu_u \mathbb{I} \geq  - \sum_p \lambda_p A_{p, u}\Rightarrow \mu_u \geq - \|\lambda\|_1
\end{gather}

Let us consider using $\lambda$, and the restriction of $\mu$ to $Q$ as a witness to the dual of the sampled program $C_{rq/n, cq/n, \gamma'}(Q)$. We note that the constraints of the dual of $C_{rq/n, cq/n, \gamma'}(Q)$ are the restriction to $Q\subset [n]$ of the dual constraints $\mathcal{D}'_{r, c, \gamma}(\lambda, \mu)$, and thereby $(\lambda, \mu_Q)$ is always feasible for the sampled dual. Moreover, Hoeffding's inequality tells us that 

\begin{equation}
    \mathbb{P}\bigg[ \bigg|\sum_{u\in Q}  \mu_u  -  \frac{q}{n}\sum_{u\in V}  \mu_u\bigg|  \geq  \frac{\gamma q}{2} \cdot \|\lambda\|_1 \bigg] \leq e^{-\Omega(\gamma^2 q)}
\end{equation}

and therefore, with probability $1-e^{-\Omega(\gamma^2 q)}$, 
\begin{gather}
0 < - \frac{q}{n}\cdot \sum_{u\in V}  \mu_u(\lambda) - \frac{q}{n} \sum_j  \lambda_j  b_j \leq  \\ \leq -  \sum_{u\in Q}  \mu_u(\lambda) - \sum_j  \lambda_j q\cdot \bigg( \frac{1}{n}\cdot b_j  - \frac{\gamma}{2} \bigg) \leq  \cdot O_{rq/n, cq/n, \gamma/2}(Q)
\end{gather}

where here we use the fact that the constraint coefficients $b_j$ are all additive in the error, e.g. $b_j = +\gamma n \pm r_{ijk} $, and therefore subtracting from $b_j$ can be viewed lowering the slack in the constraint. The last inequality is simply weak duality of the sampled program $C_{rq/n, cq/n, \gamma/2}(Q)$. Precisely, this implies that if $C_{r, c, \gamma}$ is infeasible, with probability $1-e^{-\Omega(\gamma^2 q)}$, $C_{rq/n, cq/n, \gamma/2}(Q)$ is infeasible as well. We pick $q = \Omega(\gamma^{-2}\log 1/ \delta)$ to conclude the claim.

\end{proof}

We are now in position to relate the two objectives.

\begin{claim} \label{claim-samplecutenergy}
The variational minimum energy of the cut decomposition is close to that of its sample:

\begin{equation}
    \bigg| \min_{\text{product state }\rho } \text{Tr}[H_{D}\rho]  - \frac{n^2}{q^2}\cdot \min_{\text{product state }\rho_Q } \text{Tr}[H_{D_Q}\rho_Q]\bigg| \leq O(\epsilon\cdot n^2)
\end{equation}

with probability $1-\delta$ over the choice of $Q$, so long as $q = \Omega(\epsilon^{-6}\log 1/\epsilon\delta)$

\end{claim}

\begin{proof}

By a union bound over all $2^{O(\epsilon^{-2}\log 1/\gamma)}$ possible choices of the guesses $r, c \in (I_\gamma)^{16s}$ (with width $s = O(\epsilon^{-2})$), we have that the events of Claims \ref{claim-samplelowerbound}, \ref{claim-sampleupperbound} hold with probability $1-\delta$, so long as $q = \Omega(\epsilon^{-2} \gamma^{-2}\log 1/(\gamma \delta))$. In this setting, let us first combine Claim \ref{claim-sampleupperbound} with Claim \ref{claim-sampleVadditive}, to argue that the optima of the sample provides a lower bound to the energy: 

\begin{gather}
   V_\gamma =  \min_{\substack{r, c \in (I_{n, \gamma})^{16s} \\ C_{r, c, \gamma} \text{ feasible}}}\sum_{ijk} d^{ijk} r^{ijk}c^{ijk} \geq \frac{n^2}{q^2} \cdot  \min_{\substack{r, c \in (I_{q, \gamma})^{16s} \\ C_{r, c, 2\gamma}(Q) \text{ feasible}}}\sum_{ijk} d^{ijk} r^{ijk}c^{ijk} \text{ with prob. }1-\delta/2 \\ \text{ and }
   \min_{\substack{r, c \in (I_{q, \gamma})^{16s} \\ C_{r, c, 2\gamma}(Q) \text{ feasible}}}\sum_{ijk} d^{ijk} r^{ijk}c^{ijk} \geq  \min_{\text{product state }\rho_Q } \text{Tr}[H_{D_Q}\rho_Q]  -  O(\gamma/\epsilon \cdot q^2 ) 
\end{gather}

where the second line follows from the observation that if $\rho_Q$ is a density matrix feasible for $C_{r, c, 2\gamma}(Q)$, then its energy must be at most $O(\gamma/\epsilon \cdot q^2 )$ away by Claim \ref{claim-sampleVadditive}. Conversely, 

\begin{gather}
    V_\gamma(Q) = \min_{\substack{r, c \in (I_{q, \gamma})^{16s} \\ C_{r, c, \gamma}(Q) \text{ feasible}}} \sum_{ijk} d^{ijk} r^{ijk}c^{ijk} \geq \frac{q^2}{n^2}\cdot   \min_{\substack{r, c \in (I_{n, \gamma})^{16s} \\ C_{r, c, 2\gamma} \text{ feasible}}}\sum_{ijk} d^{ijk} r^{ijk}c^{ijk} \text{ with prob. }1-\delta/2 \\ \text{ and }
     \min_{\substack{r, c \in (I_{n, \gamma})^{16s} \\ C_{r, c, 2\gamma} \text{ feasible}}}\sum_{ijk} d^{ijk} r^{ijk}c^{ijk} \geq \min_{\text{product state }\rho } \text{Tr}[H_{D}\rho]  -  O(\gamma/\epsilon \cdot n^2 ) 
\end{gather}

Using the expressions for $V_\gamma$ and $V_{\gamma}(Q)$ in Lemma \ref{lemma-Vadditive1} and Claim \ref{claim-sampleVadditive}, with $\gamma = O(\epsilon^2)$, we conclude the proof of the claim.
\end{proof}

We can now finally conclude the proof of the theorem:

\begin{proof} 

[of Theorem \ref{theorem-vertexsample}]
Via the triangle inequality, Claim \ref{claim-samplecutenergy}, Lemma \ref{lemma-sampleregularity}, and corollary \ref{corollary-bhdense} tell us that with probability 0.99, 

\begin{equation}
        \bigg|\min_{\rho}\text{Tr}[H\rho]-\frac{n^2}{q^2}\min_{\rho_Q}\text{Tr}[H_Q\rho_Q]\bigg|\leq \epsilon n^2
\end{equation}
    
\noindent by the union bound and the appropriate choice of $\delta = O(1)$ and $\epsilon$, so long as $q $ is larger than the largest size constraint, $\Omega(\epsilon^{-6}\log 1/\epsilon)$.
\end{proof}

\section{Acknowledgements}

The author would like to thank Anurag Anshu, Yunchao Liu, Umesh Vazirani and Sevag Gharibian for many fruitful discussions on product state approximations, Aram Harrow and Daniel Ranard for a discussion leading to the Hamiltonian Regularity Lemma, and Anirban Chowdhury for questions and conversations on the Free Energy of dense Local Hamiltonians. 
%%

% \newpage
\bibliographystyle{alphaurl}
\bibliography{ref}

% \newpage

\appendix

\section{The Multi-Partite Self Decoupling Lemma}

\begin{lemma}\label{lemma-kselfdecoupling} Let $X_1\cdots X_n$ be classical random variables with some arbitrary joint distribution, and fix integers $k, l < n$. Then

\begin{equation}
    \mathbb{E}_{0\leq m \leq l}\mathbb{E}_{\substack{C\subset [n],  \\ |C| = m}} \mathbb{E}_{ \substack{ u_1\cdots u_k\in V\setminus C \\ u_i\neq u_j}} I(X_{u_1}:\cdots :X_{u_k}| X_C) \leq \frac{k^2}{l}\mathbb{E}_{u}I(X_u: X_{V\setminus \{u\}})
\end{equation}

\end{lemma}

\begin{proof}
We begin by expressing the multi-partite mutual information as bipartite, followed by the chain rule of the mutual information:
\begin{gather}
    I(X_{u_1}:\cdots :X_{u_k}| X_C) = \sum_{j=1}^{k-1} I(X_{u_1}\cdots X_{u_j}:X_{u_{j+1}}|X_C) = \\ = \sum_{j=1}^{k-1}\sum_{i=1}^{j}I(X_{u_i}:X_{u_{j+1}}|X_C, X_{u_1}\cdots X_{u_{i-1}}) 
\end{gather}

Consider a single summand above, in expectation over uniformly random (and without replacement) choices of $c_1\cdots c_m, u_1\cdots u_k\in [n]$.

\begin{gather}
   \mathbb{E}_{\substack{C\subset [n],  \\ |C| = m}} \mathbb{E}_{ \substack{ u_1\cdots u_k\in V\setminus C \\ u_a\neq u_b}}  I(X_{u_i}:X_{u_{j+1}}|X_C, X_{u_1}\cdots X_{u_{i-1}}) = \\
    \mathbb{E}_{\substack{u, c_1\cdots c_{m+1}, u_1\cdots u_{i-1} \in [n] \\ \text{all distinct}}} I(X_{u}:X_{c_{m+1}}|X_{c_1}, \cdots X_{c_m}, X_{u_1}\cdots X_{u_{i-1}})
\end{gather}

\noindent in expectation over the choice of size of $C$ as well,
\begin{gather}
    \mathbb{E}_{0\leq m \leq l} \mathbb{E}_{\substack{C\subset [n],  \\ |C| = m}} \mathbb{E}_{ \substack{ u_1\cdots u_k\in V\setminus C \\ \text{all distinct}}} I(X_{u_1}:\cdots :X_{u_k}| X_C)\leq   \\
    k\cdot \sum_{i=1}^{k-1} \mathbb{E}_{0\leq m \leq l} \mathbb{E}_{\substack{u, c_1\cdots c_{m+1}, u_1\cdots u_{i-1} \in [n] \\ \text{all distinct}}} I(X_{u}:X_{c_{m+1}}|X_{c_1}, \cdots X_{c_m}, X_{u_1}\cdots X_{u_{i-1}}) = \\
    k\cdot \sum_{i=1}^{k-1} \mathbb{E}_{\substack{u, c_1\cdots c_{l+1}, u_1\cdots u_{i-1} \in [n] \\ \text{all distinct}}} \mathbb{E}_{0\leq m \leq l} I(X_{u}:X_{c_{m+1}}|X_{c_1}, \cdots X_{c_m}, X_{u_1}\cdots X_{u_{i-1}}) = \\
    = \frac{k}{l}\cdot \sum_{i=1}^{k-1} \mathbb{E}_{\substack{u, c_1\cdots c_{l+1}, u_1\cdots u_{i-1} \in [n] \\ \text{all distinct}}} I(X_u: X_{C_{\leq l+1}}| X_{u_1}\cdots X_{u_i})
\end{gather}    

    \noindent where first we used linearity of expectation, then we reordered the expectation via the permutation invariance of the sampling distribution, and finally applied the chain rule. To conclude, we can simply apply the monotonicity relations $I(A:B|C), I(A:B)\leq I(A:BC)$
    
\begin{gather}    
      \leq \frac{k}{l}\cdot \sum_{i=1}^{k-1} \mathbb{E}_{\substack{u, c_1\cdots c_{l+1}, u_1\cdots u_{i-1} \in [n] \\ \text{all distinct}}} I(X_u: X_{C_{\leq l+1}}, X_{u_1}\cdots X_{u_i}) \leq \frac{k^2}{l} \cdot \mathbb{E}_{u\in [n]}I(X_u: X_{[n]\setminus u}) 
\end{gather}

\end{proof}

\section{The Volume of Feasible Set $K_{r, c, \gamma}$}
\label{appendix-volume}

In this appendix, we construct upper and lower bounds on the volumes of feasible regions $K_{r, c, \gamma}\subset \mathbb{R}^{(d^2-1)n}$, the convex set of feasible points to constraints $C_{r, c, \gamma}$ described in section \ref{section-constraints}. In particular, we argue that any product state is contained in at least some feasible region $K_{r, c, \gamma}$  (for some choice of $r, c$), and moreover said $K_{r, c, \gamma}$ is contained in a ball of radius $R_u$ centered at the origin, and contains a ball of radius $R_l$. In this manner, by running a convex program solver, say, the Ellipsoid algorithm up to a cutoff volume $R_l^{(d^2-1)n}$ for every set of constraints $C_{r, c, \gamma}$, we are guaranteed that every product state is in at least one of the feasible $K_{r, c, \gamma}$ found by the algorithm. 

Before beginning, let us define the corresponding set and radii $\tilde{K}_{r, c, \gamma}, \tilde{R}_l, \tilde{R}_u$ for the compressed set $\tilde{C}_{r, c, \gamma}$ given by the common refinements, which we return to at the end of this section. For simplicity, let us begin with the upper bound. 

\begin{claim}
\label{volumeUB}
Every product state $\rho$ is contained in some $K_{r, c, \gamma}$, where $K_{r, c, \gamma}$ is contained in a ball of radius $R_u\leq d\cdot \sqrt{n}$ centered at the origin.
\end{claim}

\begin{proof}
Recall that each density matrix $\rho_v$ is specified by the vector $(\alpha^v_i)_{i\in [d^2]}$, where $|\alpha^v_i| = |\text{Tr}[\sigma_i \rho^v]|\leq 1$. Thus the variable space $\{\alpha_v:v\in V\}$ of valid product states is therefore contained in a ball of radius
\begin{equation}
    R^2 = \max_{\alpha \text{ valid}} \|\alpha\|_2^2 = \sum_{v\in V} \|\alpha^v\|^2_2\leq n\cdot d^2\Rightarrow R\leq d\cdot \sqrt{n}
\end{equation}
\end{proof}

The lower bound on $K_{r, c, \gamma}$ is more subtle. Conceptually, the lower bound hinges on the idea of perturbing a carefully chosen feasible point away from saturating any constraints, such that the perturbed point is a center for a ball within $K_{r, c, \gamma}$. To find a suitable starting point, we observe that the 'guesses' $r, c$ are overlapping, and thereby for every product state $\rho$ there is a choice of $r, c$ such that $\rho$ is already bounded away from the regularity constraints. 

\begin{claim}
\label{volumeLB}
Every product state $\rho$ is contained in some $K_{r, c, \gamma}$, where $K_{r, c, \gamma}$ contains a ball of radius $R_l\geq O(\gamma)$.
\end{claim}

\begin{proof}
Consider an arbitrary product state $\rho$, let $\{\alpha^v_i\}_{v\in V, i\in [d^2-1]}$ be its pauli basis description, and let $r^*, c^* \in [-n, n]^{(d^4\cdot s)}$ be its average subset magnetizations, as defined in section \ref{section-constraints}. By definition of $I_\gamma$, there is a choice of $r, c\in (I_\gamma)^{(d^4\cdot s)}$ such that for each constraint $i,j,k$, their `subset magnetizations' are close $|r^{ijk}-r^{*, ijk}|\leq \gamma n /2$. Note this is bounded $\gamma n /2$ away from saturating any magnetization constraint. Our intention is now to perturb the description $\alpha$ by an arbitrary vector $\delta$ of bounded $\|\delta\|_\infty$, and argue that the resulting point is still feasible, to reason that $K_{r', c', \gamma}$ contains a ball of radius $\|\delta\|_\infty$. Since the magnetization constraints are linear and we are already bounded away from their saturation, their analysis is quite straightforward and we defer it to later. The main technical issue is that this perturbation may violate a psd constraint, which we address by 'bumping up' all the eigenvalues of each $\rho_u$ which are smaller than some cutoff. We formalize this notion in the following claim:

\begin{claim} \label{claim-psdperturbation}
Let $Q =\mathbb{I}/d + \alpha\cdot \sigma/d\geq 0$ be a trace 1, Hermitian, PSD matrix, and let $Q = \sum_i \lambda_i |\psi_i\rangle \langle \psi_i|$ be its eigenbasis decomposition. Fix $d^{-1}>\delta>0$ and define the $\delta$-truncation 

\begin{equation}
    Q_\delta =  \sum_{i: \lambda_i < \delta} \delta |\psi_i\rangle \langle \psi_i| + \sum_{i: \lambda_i \geq \delta} \lambda_i |\psi_i\rangle \langle \psi_i|,  \tilde{Q}_\delta = \frac{Q_\delta}{\text{Tr}[Q_\delta]}
\end{equation}

\noindent Then $\tilde{Q}_\delta \geq \frac{\delta}{2} \mathbb{I}, \text{Tr}[\tilde{Q}_\delta] = 1$, and if $\alpha_\delta \in \mathbb{R}^{d^2-1}$ is the Pauli basis description of $Q_\delta$, $\|\alpha-\alpha_\delta\|_\infty \leq d^2 \delta$.
\end{claim}

We defer the proof of this claim, and use it to note that if we $\delta$-truncate the description $\alpha$ of each single qudit density matrix in the product state $\rho$ to obtain $\alpha_\delta$, by linearity and the triangle inequality we have that the average subset magnetizations $r^{ijk}_\delta, c^{ijk}_\delta$ have changed by at most $n\cdot  \|\alpha_\delta - \alpha\|_\infty \leq nd^2 \delta$. For an appropriate choice of constant $\delta = \Theta(\gamma/d^2) = \Theta(\gamma)$, we have constructed a feasible point $\alpha_\delta$ that is $\delta/2 = \Theta(\gamma)$ bounded away from the psd constraints, and $\Theta(\gamma n)$ bounded away from the magnetization constraints. To conclude, let us consider any perturbation $z\in \mathbb{R}^{(d^2-1)n}$ to the description $\alpha_\delta$ to obtain $\alpha_z = z+\alpha_\delta$, with a sufficiently small $\|z\|_\infty = \Theta(\gamma)$. Via the argument above, $\alpha_z$ remains $\Theta(\gamma \cdot n)$ away from saturating the magnetization constraints. Moreover, if $\rho_\alpha \equiv \mathbb{I}/d + \alpha\cdot \sigma/d\geq 0$, then 

\begin{equation}
    \|\rho_\alpha- \rho_{\alpha_z}\|_\infty = \frac{1}{d}\|\sum_i z_i \sigma_i\|_\infty \leq \frac{1}{d}\sum_i^{d^2}|z_i|\|\sigma_i\|_\infty\leq d \|z\|_\infty
\end{equation}

and thus $\rho_{\alpha+z} \geq (\frac{\delta}{2} - d \|z\|_\infty)\mathbb{I}\geq \frac{\delta}{4}\mathbb{I}$ if $\|z\|_\infty\leq \frac{\delta}{2d} = \Theta(\gamma)$. Thus, we have reasoned that every product state $\rho$ is contained in some $K_{r,  c, \gamma}$, which contains a feasible point $\alpha_\delta$ and a ball of radius $\Theta(\gamma)$ around it.
\end{proof}

\begin{proof} [Proof of Claim \ref{claim-psdperturbation}] Note via the normalization that $\text{Tr}[\tilde{Q}_\delta]=1$. Let $m < d$ be the number of eigenvalues of $Q$ with $\lambda_i< \delta$, and recall $\delta < 1/d$. Then
\begin{gather}
    1\leq \text{Tr}[Q_\delta] = \sum_{i:\lambda_i < \delta}\delta +  \sum_{i:\lambda_i \geq \delta}\lambda_i = \sum_{i:\lambda_i < \delta}(\delta-\lambda_i) +  1 \leq 1 + d \delta \leq 2 \\
    \Rightarrow \tilde{Q}_\delta \geq \sum_{i: \lambda_i < \delta} \frac{\delta}{1 + d \delta} |\psi_i\rangle \langle \psi_i| + \sum_{i: \lambda_i \geq \delta} \frac{\lambda_i}{1 + d \delta} |\psi_i\rangle \langle \psi_i|,  \tilde{Q}_\delta = \frac{Q_\delta}{\text{Tr}[Q_\delta]} \geq \frac{\delta}{2}\mathbb{I}
\end{gather}

Moreover, 
\begin{gather}
\tilde{Q}_\delta - Q = \sum_{i:\lambda_i < \delta}\bigg( \frac{\delta}{Tr[Q_\delta] }- \lambda_i\bigg) +  \sum_{i:\lambda_i \geq \delta}\lambda_i \bigg(\frac{1}{Tr[Q_\delta]}-1\bigg) \\
\Rightarrow \|\tilde{Q}_\delta - Q\|_\infty \leq \max \bigg(\delta, \frac{\text{Tr}[Q_\delta] - 1}{Tr[Q_\delta]}\bigg) \leq d\delta
\end{gather}

Finally, $\|\alpha-\alpha_\delta\|_\infty = \max_i|\text{Tr}[\sigma_i(Q-\tilde{Q}_\delta)]| \leq \|\sigma_i\|_1\cdot  \|\tilde{Q}_\delta - Q\|_\infty\leq d^2\delta$ as claimed. 

\end{proof}

Let us now return to the compressed constraints $\tilde{C}_{r, c,\gamma}$. Recall how $\tilde{C}_{r, c, \gamma}$ is defined on $2^{O(s)}$ variables, such that $\tilde{K}_{r, c, \gamma}$ is contained in a ball of radius $\tilde{R}_u\leq 2^{O(s)}$ via Claim \ref{volumeUB}. To prove the lower bound, we note that WLOG we could have picked the starting product state $\rho$ in the proof of Claim \ref{volumeLB} to be one in which the density matrices of the particles in each refinement set $\mathcal{A}_a$ are all the same. The argument follows immediately and we conclude $\tilde{R}_l\geq \Theta(\gamma)$.

\section{Regularity and Applications on $k$-Local Hamiltonians}
\label{section-regularityextensions}

In this section, we discuss extensions to our Hamiltonian regularity lemma to `dense' $k$-Local Hamiltonians defined on qu\textit{d}its. Our strategy generically follows the proof techniques of section \ref{section-regularity}, in first restricting our attention to product states and developing a weak regularity result for these systems, and then lifting the techniques to general entangled states via the product state approximation toolkit from \cite{Brando2013ProductstateAT} (section \ref{section-existence}). The cut decompositions we construct can be understood simply as multi-colored versions of existing regularity statements for hypergraphs by \cite{Frieze1999QuickAT}, and \cite{Alon2002RandomSA}. 

We organize this appendix as follows. In section \ref{subsection-kreg} we present and define the Hamiltonian weak regularity lemma for these systems. Then, in section \ref{subsection-algsklocal}, we show how they can be used to develop algorithms. Once again, the general idea is to exploit the low rank structure to the cut decomposition, and relax the optimization problem to checking a series of Complex SDPs (or Conic programs). Finally, in section \ref{subsection-vscklocal}, we prove a generalization of our vertex sample complexity results for $k$-local systems. Given how these results parallel that of the 2-local case, we emphasize this section to proofs of the modifications.

\subsection{A Regularity Lemma}
\label{subsection-kreg}

Let us assume the local dimension of each particle is $d = 2^{d'} = O(1)$, and observe that the set of Pauli matrices $\mathcal{P}_{\log d} = \{I, X, Y, Z\}^{\otimes \log d}$ on $\log d$ qubits defines a complete basis for any single qu\textit{d}it density matrix. We begin by expressing any $k$-local Hamiltonian over Qu\textit{d}its $H = \sum_{e\in E}h_e$, where $h_e$ acts non-trivially on the $k$-hyperedge $e\in E$, through a Pauli basis decomposition,

\begin{equation}
    H = d^{-k}\sum_{\sigma\in (\mathcal{P}_{\log d})^{\otimes k}} \sum_{e\in E}\text{Tr}[h_e \sigma_e\otimes \mathbb{I}_{V\setminus
    e}] \sigma_e \otimes \mathbb{I}_{V\setminus
    e} 
\end{equation}

\noindent where each $\sigma_e = \bigotimes_j^k \sigma_{e_j}$ is a $k$-qudit Pauli. 

Note that we have re-ordered the summation to explicitly group the interactions with the same basis. Indeed, each of $d^{2k}$ terms in the outer summation will correspond to a different colour in the cut decomposition. We can now explicitly define the tensors $M^c$ for each color $c\in [d^{2k}]$, correspondent to the hyperedge adjacency matrix for edges which support the Pauli term $\sigma^c$:

\begin{equation}
    M_{(u_1\cdots u_k)}^c \equiv \text{Tr}[h_{(u_1\cdots u_k)} \sigma^c_{(u_1\cdots u_k)}\otimes \mathbb{I}_{V\setminus
    (u_1\cdots u_k)}] \text{ where } (u_1\cdots u_k) \in ([n])^{\times k}
\end{equation}

We note we can define every entry $(u_1\cdots u_k)$ of $M^c$ WLOG, so long as we compute the trace above in an appropriate order $u_i<u_j, i<j\in [k]$ in the tensor product. In this setting, the notion of a cut decomposition for an array or `$k-$matrix' as referred to by \cite{Alon2002RandomSA}, is simply a sum over a set of `cut arrays':

\begin{definition}
For any real value $d\in \mathbb{R}$ and $k$ sets $S= S_1\cdots S_k\subset [n]$, we define $CUT(d, S)$ to be a $k$-cut array as follows:

\begin{equation}
    CUT(d, S)_{i_1\cdots i_k} = \begin{cases} d & i_1\in S_1 \cdots i_k\in S_k \\
        0 & \text{ otherwise}
    \end{cases}
\end{equation}

\end{definition}

We now have a choice on which tensor cut decomposition to use. The works by \cite{Frieze1999QuickAT} and \cite{Alon2002RandomSA} developed different high-dimensional generalizations to cut decompositions, which differ in their accuracy and rank guarantees, as well as their runtimes. The construction by \cite{Frieze1999QuickAT} can be constructed implicitly via sampling (in the probe model of computation) in time roughly $2^{\tilde{O}(1/\epsilon^2)}$, and provides a width $O(\epsilon^{2-2k})$. Thus, this is the construction we discuss for our additive approximation schemes which run in constant time. Correspondingly, the result by \cite{Alon2002RandomSA} provides a better width, $O(1/\epsilon^2)$, at the cost of polynomial runtime in $n$. That's the construction we use for the vertex sample complexity, and for the explicit algorithms. 

\begin{theorem}
    [\cite{Frieze1999QuickAT}\label{theorem-kmatrixregfrieze}] Let $J$ be an arbitrary $k$-dimensional matrix on $X_1\times \cdots \times X_k$ where we assume that $k\geq 2$ is a fixed constant $k=O(1)$. Let $N = |X_1|\cdots |X_k|$ and $\epsilon, \delta > 0$. Then, with probability $1-\delta$ and in time $2^{\tilde{O}(1/\epsilon^2)}/\delta^2$ we can find a cut decomposition of width $O(\epsilon^{2-2k})$, error $\epsilon\cdot N^{1/2}\cdot \|J\|_F$, and coefficient length $O(\|J\|_F/\sqrt{N})$.
\end{theorem}

\begin{theorem}
[\cite{Alon2002RandomSA}\label{theorem-kmatrixreg}]  In the setting of Theorem \ref{theorem-kmatrixregfrieze}, in time $2^{O(1/\epsilon^2)}O(N\log 1/\delta)$ and with probability at least $1-\delta$, we can find a cut decomposition of width $O(\epsilon^{-2})$, error at most $\epsilon \sqrt{N}\|J\|_F$, and the following bound on the coefficient length: $\sum_i|d_i|\leq 2\|J\|_F/\epsilon\sqrt{N}$, where $(d_i)^s_{i=1}$ are the coefficients of the cut arrays.
\end{theorem}

\begin{remark}
The error referred to in the theorems above is in the cut norm $\|\cdot \|_C$. As discussed in \cite{Alon2002RandomSA}, it is equivalent to the statement that
\begin{equation}
     \|W\|_{\infty\rightarrow 1} = \max_{x_1\cdots x_k\in \mathbb{R}^n, \|x_i\|_\infty=1} \bigg| \sum_{i_1\cdots i_k \in [n]}W_{i_1\cdots i_k} \prod_j x_{j, i_j}\bigg| \leq 2^k \|W\|_C 
\end{equation}
\end{remark}

In this setting, we apply either of the theorems above to each of $d^{2k}$ coloured tensors $M^c$ to obtain cut decompositions $M^c = D^c+W^c$, with cuts $D^c = \sum^s_{i=0} D^{c, i}$ of width $s$ and error tensor $W^c$. Finally, we recompose them to define our Hamiltonian cut decomposition:

\begin{equation}
    H_D = \frac{1}{k! d^k}\sum_{c\in [d^{2k}]} \sum_{i=0}^s \sum_{e\in [n]^k \atop e_a\neq e_b} D^{c,  i}_e \sigma^c_e \otimes \mathbb{I}_{V\setminus e}
\end{equation}

\noindent where $\sigma^c_e$ is the Pauli matrix on $k$ qudits acting on a hyperedge $e$, indexed the color $c\in [d^{2k}]$, and we implicitly order the tensor product in $\sigma^c_e$ such that $e_a<e_b, a<b\in[k]$ with a factor of $k!$ from a hand-shaking argument. For simplicity, in the claim below we restrict our attention to local Hamiltonians $H$ of $m$ hyper-edges of bounded strengths $\|h_e\|_\infty\leq 1$, although the proofs are readily generalizable. Following our proof techniques we immediately readout a claim over product states:

\begin{claim}\label{claim-kdpsreg} 
Fix $d, k=O(1)$. For any product state over $n$ qu\textit{d}its $\rho = \otimes_{u\in V} \rho_u$, and $H_D$ defined as above, then $|\text{Tr}[(H-H_D)\rho]|\leq \epsilon n^{k/2}\sqrt{m}$. .
\end{claim}

\begin{proof}
Let us make the observation that any Pauli matrix $\sigma^c$ defined on $k$ qudits can be expressed as the tensor product of $k$ single-qudit Paulis $\sigma^c = \bigotimes_j^k \sigma^{c, j}$, and reduce the energy difference to an expression of the $\infty\rightarrow 1$ norm by removing the diagonal terms:

\begin{gather}
   \bigg| \text{Tr}[(H-H_D)\rho]\bigg| \leq \frac{1}{d^{k}k!} \sum_{c\in [d^{2k}]} \bigg| \sum_{e\in [n]^k \atop e_a\neq e_b} W^{c}_e \prod_{j}^k \text{Tr}[\sigma^{c, j}\rho_{e_j}]\bigg| \leq \\ \leq \frac{1}{d^{k}k!} \sum_{c\in [d^{2k}]}\bigg( \|W^c\|_{\infty \rightarrow 1} + k^2 n^{k-2}\sum_v \max_{v, e'} |W^c_{v, v, e'}| \bigg)
\end{gather}

If we instantiate $H_D$ using Theorem \ref{theorem-kmatrixregfrieze}, the maximum diagonal term is $\max_{v, e'} |W^c_{v, v, e'}| = s^{1/2} \|J\|_F n^{-k/2} = \epsilon^{1-k}\|J\|_F n^{-k/2}$ by the Cauchy-Schwartz inequality on the coefficient length, and thus $k^2 n^{k-2}\sum_v \max_{v, e'} |W^c_{v, v, e'}|\leq \|W^c\|_{\infty \rightarrow 1}$ if $n =\omega(\epsilon^{-k}).$  In Theorem \ref{theorem-kmatrixreg} the maximum diagonal term is $\sum_i|d_i|\leq 2\|J\|_Fn^{-k/2}/\epsilon$, and thus $k^2 n^{k-2}\sum_v \max_{v, e'} |W^c_{v, v, e'}|\leq \|W^c\|_{\infty \rightarrow 1}$ if $n =\omega(\epsilon^{-2}).$ In either case we prove the claim.

\end{proof}

\begin{lemma} \label{lemma-kcutdecomp}
In the context of Claim \ref{claim-kdpsreg}, $\|H-H_D\| \leq \epsilon n^{k/2} m^{1/2}$
\end{lemma}

\begin{proof}
Follows from the proof of Lemma \ref{lemma-hregularity}, Claim \ref{claim-kdpsreg} and the product state approximation Theorem \ref{theorem-BHgeneral}, and the observation that $m^{2/3}n^{(k-1)/3}=o(n^{k/2}m^{1/2})$.
\end{proof}

\subsection{Algorithms for $k$-local Hamiltonians}
\label{subsection-algsklocal}

We prove a generalization to the results in section \ref{section-gsptas-dense}, on  approximation schemes for dense $k$ local systems. The main conclusion of this section is a sampling algorithm which approximates the ground state energy of a given $k$ local Hamiltonian $H$ up to an additive error of $\epsilon\cdot n^{k/2} m^{1/2}$:

\begin{theorem}
    Fix $d,k=O(1)$, and let $H$ be a $k$ local Hamiltonian on $n$ qudits of local dimension $d$ and $m$ bounded strength interactions. Then there exists an algorithm to estimate the ground state energy of $H$ up to an additive error of $\epsilon\cdot n^{k/2} m^{1/2}$ in time $2^{\tilde{O}(\epsilon^{2-2k})}$ and is correct with probability $.99$.
\end{theorem}

This algorithm is analogous to the 2-local case in section \ref{section-gsptas-dense}, and is based on implicitly generating a cut decomposition $H_D$ of $H$ using Lemma \ref{lemma-kcutdecomp}, instantiated with the higher-dimensional array cut decomposition by \cite{Frieze1999QuickAT}. Thus, for conciseness we focus this section on the proofs of the necessary modifications. At the end of this section, we discuss how an explicit approach using the array cut decomposition by \cite{Alon2002RandomSA} enables us to devise sub-exponential time algorithms whenever $m = \tilde{\Omega}(n^{k-1})$. We begin by briefly discussing the algorithm description to introduce notation and the modifications to the relaxation scheme, and to present the runtime of the approximation scheme. Afterwards, we detail the modifications to the correctness proofs. 

Recall how $H_D$ can be understood as a multi-colored cut decomposition with $d^{2k}$ different `colors', each color incurring $s$ (the width) different cuts, and each cut incurring a partition of the vertices into $k$ different partitions. The common refinement/coarsest partition of this set of $k\cdot d^{2k}\cdot s$ different subsets is a set of $A= 2^{k\cdot d^{2k}\cdot s}$ subsets. We estimate the sizes $|\hat{\mathcal{A}}_a|$ of each of these subsets $a\in [A]$, up to some additive accuracy $O(\gamma /A \cdot n)$, by random sampling using Claim \ref{claim-estimatesizes}. Now, for each cut $i$ in the decomposition and for each subset $j\in [k]$ in cut $i$, we define a `guess' $r^i_{j}\in I_\gamma\subset [-n, n]$ for its average magnetization in the direction indicated by the decomposition. For each guess vector $\Vec{r}\in (I_\gamma)^{k\cdot d^{2k}\cdot s}$, we define the convex set of constraints $C_{\Vec{r},\gamma}$ by analogously defining a linear inequality constraint over the average magnetization of each of $k$ sides of each cut. We correspondingly define the compressed set of constraints $\tilde{C}_{\Vec{r},\gamma}$, and the set $\hat{C}_{\Vec{r},\gamma}$ defined on the size estimates $|\hat{\mathcal{A}}_a|$. The algorithm enumerates over choices of $\Vec{r}$, checking the feasibility of $\hat{C}_{\Vec{r},\gamma}$ and outputting

\begin{equation}
    \hat{V}_\gamma = \min_{\Vec{r}: \hat{C}_{\Vec{r}, \gamma}\text{ feasible }} \sum_i d_i \prod_j^k r^i_j 
\end{equation}

Where $d_i$ is the coefficient of the $i$th cut in the entire multi-colored decomposition, $i\in [d^{2k}\cdot s]$. We readoff that the number of choices of $\Vec{r}$, and thereby the SDP's we need to check the feasibility of is $(2/\gamma+1)^{k\cdot d^{2k}\cdot s} = 2^{O(\epsilon^{2-2k}\log 1/\gamma)}$ if we use the implicit cut decomposition by \cite{Frieze1999QuickAT} in Theorem \ref{theorem-kmatrixregfrieze}, and assume $d, k = O(1)$. The number of partitions in the common refinements/coarsest partition of these subsets is $A = 2^{k\cdot d^{2k}\cdot s} = 2^{O(\epsilon^{2-2k})}$. We conclude in this fashion that the runtime of implicitly constructing the Hamiltonian cut decomposition, estimating the sizes of the $A$ different coarsest partitions, and solving the $2^{O(\epsilon^{2-2k}\log 1/\gamma)}$ different convex programs on $2^{O(\epsilon^{2-2k})}$ variables using Theorem \ref{theorem-bertsimas}, is $2^{O(\epsilon^{2-2k}\log 1/\gamma)}/\delta^2$ with probability $1-\delta$.

To prove correctness and in particular the accuracy guarantees, the main claim missing is a generalization of Lemma \ref{lemma-feasibleconstraints}, which we present in Claim \ref{claim-kfeasibleconstraints}. That is, we need to guarantee that any product state which is feasible for the constraints, must have energy close to its guess. This key claim ensures immediate generalizations of Lemmas \ref{lemma-Vadditive1} and corollary \ref{corollary-Vadditive}, which proves that the energy `estimate' output by the algorithm is close to the minimum energy of $H_D$ among product states. 

\begin{claim}\label{claim-kfeasibleconstraints}
    If a product state $\rho$ is feasible for the constraints $C_{\Vec{r}, \gamma}$, then
    \begin{equation}
        \bigg|\text{Tr}[H_D\rho] - \sum d^{i} \prod_{j=1}^k r^i_j\bigg|\leq \begin{cases}
            O(\gamma \cdot s^{1/2}\cdot n^{k/2}\|J\|_F) \text{ under thm }\ref{theorem-kmatrixregfrieze}\\
            O(\frac{\gamma}{\epsilon} \cdot n^{k/2}\|J\|_F) \text{ under thm }\ref{theorem-kmatrixreg}
        \end{cases}
    \end{equation}
    where we distinguish whether we generate the Hamiltonian cut decomposition $H_D$ using Theorem  \ref{theorem-kmatrixregfrieze} or \ref{theorem-kmatrixreg}.
\end{claim}

In the above, recall $J$ is the array of interaction strengths of $H$, $J_e =\|H_e\|_\infty$ where $e$ is a $k$-tuple over $[n]$. The proof of the above follows from a simple technical claim, which we defer to the end of this section and present in Claim \ref{claim-kfeas}. Beforehand, let us briefly argue that \ref{claim-kfeasibleconstraints} concludes the proof of correctness. It will later prove relevant to first present a statement for the noiseless constraints, and their objective $V_\gamma$:

\begin{corollary} [A generalization of Lemma \ref{lemma-Vadditive1}] \label{corollary-kvadditive}
    $V_\gamma$ is an $\epsilon \cdot n^{k/2}\cdot \|J\|_F$ additive error approximation to $\min_{\rho = \otimes \rho_u}\text{Tr}[H_D\rho]$, so long as $H_D$ is generated using Theorem \ref{theorem-kmatrixreg} and $\gamma = O(\epsilon^2)$, or, $H_D$ is generated using Theorem \ref{theorem-kmatrixregfrieze} and $\gamma = O(\epsilon^{k})$.
\end{corollary}

By a brief modification to the feasibility statements in \ref{claim-kfeasibleconstraints}, we extend the argument in the corollary above to the constraints $\hat{C}_{\Vec{r},\gamma}$. Recall $\hat{C}_{\Vec{r},\gamma}$ is defined over estimates of the actual size of the common refinements. We argue the objective value $ \hat{V}_\gamma$ is robust to these noisy estimates, so long as the tolerance $\gamma$ is small enough, and the noise to the size estimates is bounded.

\begin{corollary} 
    Let $H_D$ be the Hamiltonian cut decomposition of a $k$-Local Hamiltonian $H$ using Theorem \ref{theorem-kmatrixregfrieze}, and assume we have estimates $|\hat{\mathcal{A}}_a|$ for the sizes of each subset $a\in [A]$ in the coarsest partition of the cuts in $H_D$ which are accurate up to additive error $\gamma/(4A) \cdot n$. Then, $\hat{V}_{O(\epsilon^{k})}$ is an $\epsilon \cdot  n^{k/2}\cdot \|J\|_F$ additive error estimate to $\min_{\rho = \otimes \rho_u}\text{Tr}[H_D\rho]$.
\end{corollary}

The proof of the corollary above is the same as Claim \ref{lemma-perturbedconstraints}. By direct application of Claim \ref{claim-estimatesizes}, we can compute these estimates efficiently to within the intended accuracy guarantees. By combining the corollary above with the guarantees in Claim \ref{claim-kdpsreg} and Lemma \ref{lemma-kcutdecomp}, we ensure $\hat{V}_{O(\epsilon^{k})}$ is an $\epsilon\cdot n^{k/2}m^{1/2}$ additive error approximation to the ground state energy $\min_\rho \text{Tr}[H\rho]$, as intended. 

To conclude the proof of correctness, we present the missing claim.

\begin{claim} \label{claim-kfeas}
    Let $d_i$, $x_i^j$, $y_i^j$ for $j\in [k]$ and $i\in [s]$ be numbers such that $|x_i^j|, |y_i^j|\leq n, |x_i^j - y_i^j|\leq \gamma \cdot n$ and either (1) $\sum_i d_i^2  = O(\|J\|_F/n^{k/2})$ or (2) $\sum_i |d_i|  = O(\|J\|_F/\epsilon n^{k/2})$. Then, 
    \begin{gather}
       (1) \bigg| \sum d_i \prod_j x_i^j - \sum d_i \prod_j y_i^j\bigg| = O(\gamma \cdot s^{1/2}\cdot n^{k/2}\|J\|_F) \text{ and }\\
       (2)\bigg| \sum d_i \prod_j x_i^j - \sum d_i \prod_j y_i^j\bigg| = O(\gamma \cdot s\cdot n^{k/2}\|J\|_F) \leq O(\frac{\gamma}{\epsilon} \cdot  n^{k/2}\|J\|_F)
    \end{gather}
    
\end{claim}

\begin{proof} 
    \begin{equation}
        \bigg| \sum d_i \prod_j x_i^j - \sum d_i \prod_j y_i^j\bigg|\leq \sum |d_i| \cdot \bigg|\prod_j x_i^j-y_i^j\bigg|
    \end{equation}
    In case (1), we apply the Cauchy-Schwartz inequality as before, with the observation $\big|\prod_j x_i^j-\prod_j y_i^j\big|\leq k\cdot \gamma \cdot n^{k}$, obtaining the error intended. In case (2), we simply upper bound $\big|\prod_j x_i^j-\prod_j y_i^j\big|$.
\end{proof}

\begin{proof}

[of Claim \ref{claim-kfeasibleconstraints}] Let $\rho$ be a product state with vector of true average subset magnetizations $\Vec{r}$. We know that the energy of $\rho$ is roughly a function of $\Vec{r}$, up to the `self-edges' which appear in the cut decomposition:

\begin{equation}
    \bigg|\text{Tr}[H_D\rho] - \sum_i d_i \prod_j^k r^i_j \bigg| \leq k^2n^{k-1}\max_{v,e'}|W^C_{v,v, e'}|\leq \epsilon n^{k/2}\|J\|_F \text{(see Claim \ref{claim-kdpsreg})}
\end{equation}

\noindent and thereby with Claim \ref{claim-kfeas} and the appropriate choice of $\gamma$ for each cut decomposition we conclude the claim.
\end{proof}

To conclude this section, we remark that we could have performed all of the steps above explicitly, using the explicit higher-dimensional cut decomposition by \cite{Alon2002RandomSA} presented in Theorem \ref{theorem-kmatrixreg}. We show this gives us sub-exponential time algorithms in almost the entire regime where product states give extensive additive error approximations to the ground state energy:

\begin{theorem} 
Fix $d, k=O(1)$ and $\epsilon > 0$. Let $H = \sum_e h_e$ be a $k$-Local Hamiltonian on $n$ qudits of local dimension $d$, and $m = \Omega(n^k\log n)$ interactions of bounded strength $\|h_e\|_\infty\leq 1$. There exists a randomized algorithm which runs in time $\tilde{O}(n^k)\cdot 2^{\tilde{O}(n^k/\epsilon^2 m)}$, and with high probability computes an estimate for the ground state energy of $H$ accurate up to an additive error of $\epsilon \cdot  m$. 
\end{theorem}

\begin{proof}
As discussed in corollary \ref{corollary-kvadditive}, the estimate $V_{O(\epsilon'^2)}$ for the variational minimum energy using the cut decomposition $H_D$ instantiated with Theorem \ref{theorem-kmatrixreg} is an $\epsilon' \cdot n^{k/2}\cdot m^{1/2}$ approximation to the variational minimum energy of $H_D$. If we pick $\epsilon' = \epsilon\cdot m^{1/2}n^{-k/2}/2$, this is a $\epsilon\cdot m/2$ additive error. By Claim \ref{claim-kdpsreg} and Theorem \ref{theorem-BHgeneral}, this ensures a $\epsilon\cdot m + O(n^{(k-1)/3}m^{2/3})\leq 2\epsilon \cdot m$ for any constant $\epsilon =O(1)$ and $m = \Omega(n^{k-1}\log n)$. The resulting runtime of this scheme is dominated by that of computing the cut decomposition and solving the $2^{\tilde{O}(\epsilon'^{-2})} = 2^{O(n^k \log n/\epsilon^2 m)}$ convex programs. To ensure the correctness guarantees with high probability, we repeat the algorithm $\log n$ times and output the product state of minimum energy. 
\end{proof}

\subsection{Vertex Sample Complexity for $k$-Local Hamiltonians}
\label{subsection-vscklocal}

In this section we prove a generalization of the vertex sample complexity result for 2-Local Hamiltonians to $k$-Local Hamiltonians, following the techniques in section \ref{section-vsc}. The main result of this section was presented in Theorem \ref{results-vsc}. Recall we assume the locality $k$ and the local dimension of the quantum particles $d$ to be both $O(1)$.

First of all, let us reason on the cut decomposition $H_D$ of $H$, using the result in Theorem \ref{theorem-kmatrixreg} by \cite{Alon2002RandomSA}. We note that under this decomposition, we have introduced $O(d^{2k}/\epsilon^2)$ $k-$cuts in the graph, and thereby the coarsest partition of the cuts is a family of $A=2^{O(1/\epsilon^2)}$ subsets of $[n]$, asymptotically the same as in the 2-local case. Crucially, the structure of the compressed constraints $\tilde{C}_{\Vec{r},\gamma}$ is the same as in section \ref{section-vsc}: there is a $d\times d$ density matrix defined for each subset $a\in [A]$, a trace and PSD constraints, and the relaxed average magnetization constraints. Once placed in standard form, we observe the duality arguments are entirely unmodified. Moreover, we have already proved a generalization to Lemma \ref{lemma-feasibleconstraints} in Claim \ref{claim-kfeasibleconstraints}, guaranteeing the properties of $V_\gamma$ and how it relates to the minimum energy over product states of $H_D$. In this manner, up to irrelevant constant factors of $d, k$ the only modifications we require to prove Theorem \ref{results-vsc} are Corollary \ref{corollary-bhdense}, on product state approximations to the ground state energy and Lemma \ref{lemma-sampleregularity}, on the cut norm of random restrictions of arrays with small cut norm.

Fortunately, our product state approximations on $k$-local systems in Theorem \ref{theorem-BHgeneral} present a generalization to Corollary \ref{corollary-bhdense}:

\begin{corollary}
Fix $d, k =  O(1)$. Given a $k$-Local Hamiltonian $H=\sum_e H_e$ on $n$ qudits and $m$ interactions of strength bounded by $\|H_e\|_\infty\leq 1$, there exists a product state $\sigma$ such that
\begin{equation}
    \text{Tr}[H\sigma]\leq \lambda_{min}(H) + O(n^{\frac{k-1}{3}}m^{2/3})
\end{equation}
\end{corollary}

To conclude, we use a theorem by \cite{Alon2002RandomSA} on the cut norm of random restrictions of arrays with small cut norm, which they proved as a higher dimensional generalization to the previously discussed Theorem \ref{theorem-submatrixcutnorm}. In particular, 

\begin{theorem}[\cite{Alon2002RandomSA}, Theorem 6] \label{theorem-ksubmatrixcutnorm} Fix $\epsilon, \delta>0$. Let $W$ be a $k$ dimensional array, with bounded norms $\|W\|_{\infty} = O(\epsilon^{-1}), \|W\|_{\infty\rightarrow 1}\leq \epsilon n^k$, and $\|W\|_{F}\leq O(n^k)$. Suppose $Q$ is a random subset of $[n]$ of size $q= \Omega(\delta^{-5}\epsilon^{-4}\log 1/\epsilon)$, and let $W_Q$ be the sub-matrix defined by the restriction of $W$ to $Q$. Then, with probability $1-\delta$, $\|W_Q\|_{\infty\rightarrow 1} \leq O(\epsilon/\sqrt{\delta}\cdot q^k)$, $\|W_Q\|_{F}\leq O(q^k/\sqrt{\delta})$.
\end{theorem}

\cite{Alon2002RandomSA} use the above to prove that a random restriction of the cut decomposition, is still a valid and accurate cut decomposition for the induced subgraph on the sample. We apply their statement to derive an analogous statement to Lemma \ref{lemma-sampleregularity} for $k$-local systems.

\begin{lemma} \label{lemma-ksampleregularity}
Let $H_{D_Q}$ be the sub-hamiltonian of the decomposition $H_D$ of support only in the random set $Q\subset V$. Then with probability $1-\delta$ over the choice of $Q$, for every product state $\rho_Q$ on $Q$, 
\begin{equation}
   \bigg| \text{Tr}[(H_Q-H_{D_Q})\rho]\bigg|\leq O(\epsilon/\sqrt{\delta}\cdot q^k),
\end{equation}
so long as $q= \Omega(\delta^{-5}\epsilon^{-4}\log 1/\epsilon)$.
\end{lemma}

\begin{proof}
Consider the $d^{2k} = O(1)$ `error' arrays $W^{c}$, for each `color' $c\in [d^{2k}]$ in the generalized Pauli basis decomposition. We note the arrays $W^{c}$ fit the norm guarantees of Theorem \ref{theorem-ksubmatrixcutnorm} \cite{Alon2002RandomSA} by construction, see Theorem \ref{theorem-kmatrixreg}. By a union bound over all $d^{2k} = O(1)$ arrays, we are guaranteed that with probability $1-\delta$, $|\text{Tr}[(H_Q-H_{D_Q})\rho]| \leq \sum_{c}\|W^{c}_Q\|_{\infty\rightarrow 1}\leq O(\epsilon/\sqrt{\delta}\cdot q^k)$ for all product states $\rho$, so long as $q= \Omega(\delta^{-5}\epsilon^{-4}\log 1/\epsilon)$. 
\end{proof}

\section{Regularity and Applications on Low-Threshold Rank Hamiltonians}
 \label{section-threshold}

Let us now turn to extending the Hamiltonian regularity lemma to 2-Local Hamiltonians on graphs of low threshold rank. Following the proof techniques of the previous sections, we begin by devising a cut decomposition for these systems by considering multi-coloured versions of the cut decomposition for low threshold rank graphs by \cite{Gharan2013ANR}. Unfortunately, the structure of these systems will prohibit us from exhibiting a clean spectral characterization of the decomposition as in Lemma \ref{lemma-hregularity}, however, combined with our assymetric product state approximations, we will use them to devise extensive error approximation algorithms for these 2-local quantum systems. 

We present two results in this section. The first of which is an algorithm for 2-Local Hamiltonians on `generic' low threshold rank interaction graphs. It highlights the main challenges in establishing a regularity lemma for quantum systems on these graphs, and develops the main algorithmic techniques we use. Our second and main result of this section is a classical algorithm for Quantum Max Cut on low threshold rank graphs, discussed in Theorem \ref{results-qmc}. We show how the symmetry in the Quantum Max Cut Hamiltonian in different basis enables us to apply the `common refinement' technique discussed in section \ref{section-gseptas}, and compress the size of the resulting optimization program, which we were not able to do on general low threshold rank Hamiltonians.

\subsection{A Regularity Lemma}
\label{subsection-regthreshold}

The $\delta$-SOS threshold rank of a graph $G$ is the number of eigenvalues in the normalized adjacency matrix of $G$ outside of the range $[-\delta, \delta]$. In the more general case of a real symmetric matrix $J$, we define the normalized adjacency matrix $J_D$ of $J$ in terms of `effective degrees'. For each row $u$, let the effective degree $d_u = \sum_v |J_{uv}|$, $D = \text{diag}(d_u, u\in [n])$ be the diagonal matrix of degrees, and $J_D = D^{-1/2}JD^{-1/2}$.

\begin{definition}
The $\delta$-SOS threshold rank of a symmetric real matrix $J$ is defined to be $t_\delta(J_D) = \sum_{i:|\lambda_i|\geq \delta}\lambda_i^2$, where $\lambda_1\cdots \lambda_n$ are the eigenvalues of $J_D$.
\end{definition}

The weak regularity lemma by \cite{Gharan2013ANR} for low threshold rank matrices is based on designing a cut decomposition for a low rank approximation of $J_D$. In their decomposition, the notion of a `cut matrix' is now a weighted complete bipartite sub-graph, CUT$(S_1, S_2, \alpha) = \alpha \cdot d_{S_1}d^T_{S_2}$, where here we indicate $d_S$ as the vector of degrees $(d_S)_i = d_i$ if $i\in S$, $0$ otherwise. They prove:

\begin{theorem} \label{thm-regthreshold}
[\cite{Gharan2013ANR}] Let $J$ be a real symmetric matrix, $\epsilon > 0$, and $t\equiv t_{\epsilon/2}(J_D)$. Then there exists a cut decomposition of $J = \sum D^{(i)}+W$ of width $O(t\epsilon^{-2})$, error $\|W\|_{\infty\rightarrow 1}\leq \epsilon \cdot|J|_1$, and where each coefficient of the decomposition satisfies $|\alpha_i|\leq O(\sqrt{t}/|J|_1)$. Furthermore, this decomposition can be found in \text{poly}$(n,  \epsilon^{-1}, t)$ time.
\end{theorem}

We use this theorem to define decompositions for 2-Local Hamiltonians $H= \sum_{(u,v)\in E}H_{u,v}\otimes  \mathbb{I}_{V\setminus (u, v)}$ on qudits whose Pauli interaction graphs are undirected and each have low threshold rank. That is, let $H = \sum_{a, b\in [d^2]}\sum_{(u,v)\in E } h_{uv}^{ab }\sigma_{u}^a\otimes \sigma_{v}^b\otimes \mathbb{I}_{V\setminus (u, v)}$ be a Pauli basis decomposition of $H$, and let $J^{ab} = \{h^{ab}_{uv}\}_{u, v\in [n]}$ be the $(a, b)\in [d^2]^2$ interaction matrix of $H$. By undirected, we simply mean that $h_{uv}^{ab} = h_{vu}^{ab}$, and thus $J^{ab}$ can be considered the adjacency matrix of a weighted undirected graph. We refer to the effective $(a, b)$-degree of a vertex $u$ as $d_u^{a, b} = \sum_v |J^{a, b}_{uv}|$, and note that we assume $d^{a, b}=d^{b, a}$.

Let $J^{a, b}_D$ be the normalized adjacency matrix for the interaction $(a, b)$. In this setting, the `threshold rank' we consider is the maximum $\delta$-SOS threshold rank $t\equiv \max_{a, b} t_{\epsilon/2}(J^{a, b}_D)$ among the interactions $a, b$. If this $t$ is sufficiently small, then we can apply the framework of \cite{Gharan2013ANR} to our setting efficiently. In particular, we use Theorem \ref{thm-regthreshold} on each $J^{(a, b)}$ for $(a, b) \in [d^2]^2$, obtaining cuts $(R^{a, b, i}, L^{a, b, i})$ and cut matrices $D^{(a, b, i)}, i\in [s]$ in $\text{poly}(n, t, 1/\epsilon)$ time. Now, we would like to follow our ideas in section \ref{section-regularity} and express a Hamiltonian $H_D$ 

\begin{equation}
    H_D  = \frac{1}{2}\sum_{a, b\in [d], i\in [s]}\sum_{u\neq v}D^{(a, b, i)}_{u , v}\sigma_{u}^a\otimes \sigma_{v}^b.
\end{equation}

However, this unfortunately does \textit{not} necessarily preserve the energy of product states. If the cut decomposition generated only disjoint cuts ($R^{a, b, i}\cap L^{a, b, i} = \emptyset$) such that the diagonal entries of $D^{(a, b, i)}$ are all 0, then indeed we would recover the product state statement $|\text{Tr}[(H-H_D)\rho]|\leq \|W\|_{\infty\rightarrow 1}$ and the ensuing spectral bound. Instead, we can only prove the weaker statement with the self-edges:

\begin{claim}\label{claim-regularitythreshold}
Let $H=\sum_{e} H_{e}$ be a 2-Local Hamiltonian defined on qudits of local dimension $d=O(1)$, and let $J_{u, v} = \|H_{(u, v)}\|_{\infty}$ be the matrix of interaction strengths. Let $t\equiv \max_{a, b} t_{O(\epsilon)}(J^{a, b}_D)$ be the maximum threshold rank among the Pauli interactions $a, b\in [d^2]$, and let $D^{(a, b, i)} =$ CUT$(R^{a, b, i},$ $L^{a, b, i}, \alpha^{a, b, i}),  i\in [O(t\epsilon^{-2})]$ be a cut decomposition of each $J^{a, b}$ using Theorem \ref{thm-regthreshold}. Then, for all product states $\rho = \otimes_u\rho_u$ 

\begin{equation}
        \bigg| \text{Tr}[H\rho] - \frac{1}{2}
        \sum_{a, b\in [d^2]\atop i\in [s]} \alpha^{a, b, i} \bigg(\sum_{u\in R^{a, b, i}} d^{a, b}_u\cdot \text{Tr}[\sigma^a \rho_u]\bigg) \bigg(\sum_{v\in L^{a, b, i}} d^{a, b}_v\cdot\text{Tr}[\sigma^b \rho_v]\bigg) \bigg|\leq \epsilon\cdot |J|_1
\end{equation}

\end{claim}

\begin{proof}
Let $W^{a, b} = J^{a, b}-\sum_i D^{a, b, i}$ be the $n\times n$ error matrices in the cut decomposition, and note that the error above is exactly 

\begin{equation}
    = \bigg|\sum_{a, b} \sum_{u, v}W^{a, b}_{u, v} \text{Tr}[\sigma^a \rho_u]\text{Tr}[\sigma^b \rho_v]\bigg|\leq \sum_{a, b}\|W^{a, b}\|_{\infty \rightarrow 1} \leq \epsilon\cdot d^4 \cdot |J|_1
\end{equation}

by appropriately rescaling $\epsilon$ and the width of the decomposition we prove the claim.
\end{proof}

\subsection{Algorithms for Low-Threshold Rank Hamiltonians}
\label{subsection-algthreshold}

In this section, we discuss two algorithms for local Hamiltonians configured on low threshold rank graphs. The first of which is the natural generalization of the techniques in section \ref{section-gseptas}, in relaxing the optimization program to checking the feasibility of a small number of convex programs. However, in the general case that $H$ is composed of multiple \textit{distinct} low threshold rank interaction graphs, we will not be able to perform the common refinement technique analogous to section \ref{section-gseptas} to compress these convex constraints into a constant number of variables. Instead, for our second algorithm we restrict our attention to the Quantum Max Cut on on low threshold rank graphs, and show how for these systems with further symmetry in their interaction graph the compression technique holds and leads to faster algorithms.

To begin, let us consider a 2-Local Hamiltonian $H$ on $n$ qubits whose Pauli graphs are undirected and of maximum threshold rank $t = \max_{a, b}t_{O(\epsilon)} (J^{a, b}_D)$. In the setting of Claim \ref{claim-regularitythreshold}, the energy of $H$ on any product state $\rho=\otimes \rho_u$ is
\begin{equation}
    \text{Tr}[H\rho] \approx \frac{1}{2} \sum_{a, b\in [d^2]}\sum_{k\in [s]} \alpha^{a, b, k} \bigg(\sum_{u\in R^{a, b, k}} d_u^{a, b}\text{Tr}[\sigma_{u}^a\rho_u]\bigg)\cdot \bigg(\sum_{u\in L^{a, b, k}} d_v^{a, b}\text{Tr}[\sigma_{v}^b\rho_v]\bigg)
\end{equation}

We establish a series of inequality constraints on the weighted magnetizations above, to relax this non-convex optimization problem over product states to a convex feasibility problem. In particular, since $\sum_{u\in R^{a, b, k}} d_u^{a, b} \leq  \|J^{a, b}\|_1$, we pick precision parameters $\Delta^{a, b}$ and ranges for each color $a, b$:

\begin{equation*}
    I^{a, b}= \bigg\{- \Delta^{a,  b} \big(\lfloor \frac{ \|J^{a, b}\|_1}{\Delta^{a,  b}}\rfloor +1\big) , - \Delta^{a,  b} \lfloor \frac{ \|J^{a, b}\|_1}{\Delta^{a,  b}}\rfloor, \cdots \Delta^{a,  b} \big(\lfloor \frac{ \|J^{a, b}\|_1}{\Delta^{a,  b}}\rfloor +1\big)\bigg\}
\end{equation*}

That is, rounding to the nearest multiple of $\Delta^{a,  b}$. For each of $d^4 \cdot s$ cuts in the cut decomposition, we pick guesses $\Vec{r}, \Vec{c} = (r^{a, b,  k})_{a, b\in[d^2],k\in [s]}, (c^{a, b,  k})_{a, b\in[4],k\in [s]}, r^{a, b,  k}, c^{a, b,  k} \in I^{a, b}$ for the weighted magnetizations of both sides of each cut. For each $\Vec{r}, \Vec{c}$, we define the relaxed constraints $C_{\Vec{r}, \Vec{c}}$ over the description of a product state over $n$ qudits to check whether the product state has these weighted magnetizations within a $\Delta$ range of the guess:
\begin{gather}
  C_{\Vec{r}, \Vec{c}}: \text{Tr}[\rho_{u}]=1,   \rho_{u}\geq 0 \text{ for all particles }u\in V  \\
    r^{a, b, k} - \Delta^{a, b} \leq \sum_{u\in R^{a, b, k}} d_u^{a, b} \cdot \text{Tr}[\sigma_{u}^a\rho_u] \leq r^{a, b, k} + \Delta^{a, b} \text{ and } \\
    c^{a, b, k} - \Delta^{a, b} \leq \sum_{u\in L^{a, b, k}} d_u^{a, b} \cdot \text{Tr}[\sigma_{u}^b\rho_u] \leq c^{a, b, k} + \Delta^{a, b} \text{ for all }a, b\in [d^2], k\in [s]
\end{gather}

Our algorithm is to enumerate over all the choices of $\Vec{r}, \Vec{c}$, and for each one, check the feasibility of $C_{\Vec{r}, \Vec{c}}$ via the algorithm in Theorem \ref{theorem-bertsimas} up to a cutoff radius $r'$. Among the pairs $\Vec{r}, \Vec{c}$ which are feasible, we output as estimate the pair $\Vec{r}^*, \Vec{c}^*$ which minimizes the energy estimate $\sum_{a, b, k} \alpha_{a, b, c} r_{a, b, k}^* c_{a, b, k}^*$, and we output any feasible $\rho$ for the choice of $C_{\Vec{r}^*, \Vec{c}^*}$. Let $\hat{V}$ be said energy estimate. 

To prove correctness of this scheme, we need to show that for an appropriate choice of the precision parameters $\Delta^{a,b}$, the feasible solutions to $C_{\Vec{r}, \Vec{c}}$ have energy close to the guess $\sum_{a, b, k} \alpha_{a, b, c} r_{a, b, k} c_{a, b, k}$, and thereby the output energy is close to the ground state energy.

\begin{claim} \label{claim-feasconstraintsthreshold}
    Pick $\Delta^{a, b} = \delta \|J^{a, b}\|_1$, with $\delta = O(\epsilon^3/t^{3/2})$. If a product state $\rho$ is feasible for $C_{\Vec{r}, \Vec{c}}$, then 
    \begin{equation}
        \bigg|\text{Tr}[H\rho] - \sum_{a, b, k} \alpha_{a, b, c} r_{a, b, k} c_{a, b, k}\bigg|\leq \epsilon \cdot |J|_1 
    \end{equation}
\end{claim}

\begin{proof}  
Let $r', c'$ be the length $d^4\cdot s = O(t\epsilon^{-2})$ vectors of the true weighted average magnetizations of $\rho$, then the error to the energy estimate above satisfies

    \begin{gather}
        \leq \bigg|\text{Tr}[H\rho] - \sum_{a, b, k} \alpha_{a, b, c} r'_{a, b, k} c'_{a, b, k}\bigg| + \sum_{a, b, k} |\alpha_{a, b, c}|\cdot   \bigg|r'_{a, b, k} c'_{a, b, k} - r_{a, b, k} c_{a, b, k}  \bigg| \leq \\
        \leq \epsilon\cdot |J|_1 + s\cdot \sum_{a, b} \max_{k}|\alpha_{a, b, k}| \cdot (2\Delta^{a,  b} \|J^{a, b}\|_1+(\Delta^{a, b})^2) \leq  \\\leq \epsilon\cdot |J|_1 +  3\sqrt{t}\cdot s \sum_{a, b} \Delta^{a,  b} = \epsilon\cdot |J|_1 + O( t^{3/2}\epsilon^{-2} \cdot \delta |J|_1 ) 
    \end{gather}

\noindent where we used Claim \ref{claim-regularitythreshold} and the properties of the decomposition in Theorem \ref{thm-regthreshold}. We conclude this claim by picking $\delta = O(\epsilon^3 t^{-3/2})$ and appropriately rescaling $\epsilon$.
\end{proof}

To show we can find these feasible points efficiently, if they exist, we need to argue that the convex set of feasible points to $C_{\Vec{r}, \Vec{c}}$ contains a small ball of radius $r'$. This once again follows from the fact that the ranges $I^{a, b}$ are overlapping: that is, for every product state $\rho$, there exists a choice of $\Vec{r}, \Vec{c}$ s.t. $\rho$ is feasible and bounded $\min_{a,b}\Delta^{a,b}/2$ away from saturating any of the magnetization constraints. It is an easy corollary of the proof techniques in section \ref{appendix-volume} that $C_{\Vec{r}, \Vec{c}}$ thereby contains a ball of radius $r' = O(\min_{a,b} \frac{\Delta^{a,b}}{|J^{a, b}|_1}) = O(\epsilon^3/t^{3/2})$, under the conditions on $\Delta^{a,b}$ in the claim above. 

\begin{corollary}
    $\hat{V}$ is an $\epsilon \cdot |J|_1  + O(n^{1/3} |J|_1^{1/3} \|J\|_F^{2/3})$ additive error estimate to the ground state energy $\min_{\rho }\text{Tr}[H\rho]$.
\end{corollary}

\begin{proof}
Since every product state feasible for some $C_{r, c}$ has energy close to its estimate, and since every product state is feasible for some $C_{r, c}$ for some choice of $r, c$, we have  $|\hat{V} - \min_{\text{product state }\rho}\text{Tr}[H\rho]|\leq \epsilon |J|_1$.   To conclude, we use the assymetric product state approximations in Theorem \ref{theorem-BHgeneral} to relate $\min_{\rho =\otimes \rho_u}\text{Tr}[H\rho]$ and $\min_{\rho }\text{Tr}[H\rho]$.
    
\end{proof}

We summarize this result in the following theorem:

\begin{theorem}
    Fix $d = O(1), \epsilon >0$. Let $H=\sum_{u,  v} H_{u,v}$ be a 2-Local Hamiltonian defined on $n$ qudits of local dimension $d$, with Pauli interaction graphs $J^{a, b}$ of maximum threshold rank $t\equiv \max_{a, b\in[d^2]} t_{O(\epsilon)}(J^{a, b}_D)$, and matrix of interaction strengths $J = \{\|H_{uv}\|_\infty\}_{u, v\in [n]} $. Then, there exists an algorithm which finds an $\epsilon |J|_1 +O(n^{1/3} |J|_1^{1/3} \|J\|_F^{2/3}) $ approximation to the ground state energy of $H$ in time \text{poly}$(n)\cdot 2^{\tilde{O}(t/\epsilon^2)}$. 
\end{theorem}

\begin{proof}
    It remains to reason on the runtime of the algorithm. We compute the decomposition  using Theorem \ref{thm-regthreshold} in time poly$(n, 1/\epsilon, t)$. For fixed $\Vec{r}, \Vec{c}$, we can check the feasibility of $C_{\Vec{r}, \Vec{c}}$ in time poly$(n, 1/\epsilon, t)$ using Theorem \ref{theorem-bertsimas}, given that the number of variables is $O(n)$, the number of constraints is  poly$(1/\epsilon, t)$, and the guarantees on the volume of the feasible regions is $\log R/r = O(\log n + \log 1/\epsilon + \log t)$. To conclude, the number of such programs is 

    \begin{equation}
       \leq  \prod_{a, b\in [d^2], k\in [s]} \bigg(2\frac{\|J^{a, b}\|_1}{\Delta^{a, b}}  + 3\bigg) \leq 2^{O(t/\epsilon^2 \log 1/\delta)} = 2^{\tilde{O}(t/\epsilon^2)}
    \end{equation}

    \noindent where $\tilde{O}$ hides factors of poly$\log 1/\epsilon$, poly$\log t$ and we use the definition of $\Delta^{a, b}$ in claim \ref{claim-feasconstraintsthreshold}.
\end{proof}

For certain restrited classes of low threshold rank Hamiltonians, we are able to reduce this runtime to \text{poly}$(n, 1/\epsilon, t) + 2^{\tilde{O}(t/\epsilon^2)}$. To do so, we would essentially like to exploit the same common refinement technique discussed in the dense case in section \ref{section-gseptas}, where we replace the density matrices $\rho_u$ within each partition by a single `averaged' density matrix $\rho_P$ for every $u\in P\subset [n]$. That is, in the dense case, given any product state $\rho$ which is feasible for some $C_{\Vec{r}, \Vec{c}}$, we argued that averaging over the partitions $\rho_P = \sum_{u\in P}\rho_u / |P|$ defines an $n$ qubit product state which is still feasible for $C_{\Vec{r}, \Vec{c}}$. Unfortunately, this no longer holds in the low threshold rank case, since different colors $a, b$ may attribute different degrees $d^{a, b}_u$ to $u$, and thereby there is no guarantee that any average over the components in each partition $P$ preserves the feasibility of the magnetization constraints. 

However, when the degrees $d^{a, b}_u$ are color-independent, i.e. either $d^{a, b}_u  = d_u$ for all $u\in V$ or $d^{a, b}_u = 0$ for all $u\in V$, then a certain convex combination of the density matrices in each partition achieves our goal. Quantitatively, an example of this condition is the Quantum Max-Cut Hamiltonian, where each $X_u\otimes X_v$ interaction is accompanied by a $Y_u\otimes Y_v$ of the same weight, but the generic Quantum Heisenberg model is not. By weighting the convex combination by the degrees, 

\begin{equation}
    \rho_P = \sum_{u\in P} \frac{d_u}{\sum_{u\in P} d_u} \rho_u \Rightarrow  \sum_{u\in P} d_u\text{Tr}[ \sigma^a\rho_P] = \sum_{u\in P} d_u \text{Tr}[ \sigma^a\rho_u]
\end{equation}

\noindent $\rho_P$ has trace 1 and is PSD, and preserves the weighted magnetizations of the subset $P$ for every choice of basis $a\in [4]$. In this setting, we can define a `compressed' set of constraints $\tilde{C}_{\Vec{r}, \Vec{c}}$, where the variables are $2^{O(s)} = 2^{O(t/\epsilon^2)}$ PSD matrices $\rho_P$, one for each subset $P$ in the common refinement. For notional convenience, we denote as $d^{a,b}(P) =  \sum_{u\in P} d^{a,b}_{u}$.

\begin{gather}
  \tilde{C}_{\Vec{r}, \Vec{c}}: \text{Tr}[\rho_{P}]=1,   \rho_{P}\geq 0 \text{ for all particles }u\in V  \\
    r^{a, b, k} - \Delta^{a, b} \leq \sum_{P: P\subset R^{a, b, k}} d^{a, b}(P) \cdot \text{Tr}[\sigma^a\rho_P] \leq r^{a, b, k} + \Delta^{a, b} \text{ and } \\
    c^{a, b, k} - \Delta^{a, b} \leq \sum_{P:P\subset L^{a, b, k}} d^{a, b}(P) \cdot \text{Tr}[\sigma_{u}^b\rho_P] \leq c^{a, b, k} + \Delta^{a, b} \text{ for all }a, b\in [4], k\in [s]
\end{gather}

\begin{theorem}
    Let $H = \sum_{(u, v)} H_{u,v}$ be a instance of the Quantum Max Cut Hamiltonian, with $n\times n$ edge weight matrix $J$ of threshold rank $t = t_{O(\epsilon)}(J_D)$. Then there exists an algorithm which finds a $\epsilon |J|_1 +O(n^{1/3} |J|_1^{1/3} \|J\|_F^{2/3}) $ additive error approximation to the Quantum Max Cut $\max_\rho \text{Tr}[H\rho]$ in time $\text{poly}(n, t,1/\epsilon) + 2^{\tilde{O}(t/\epsilon^2)}$.
\end{theorem}

\begin{proof}
    Crucially, as previously discussed $C_{\Vec{r}, \Vec{c}}$ is feasible $\iff$  $\tilde{C}_{\Vec{r}, \Vec{c}}$ is feasible. Moreover, every set of $2^{O(t/\epsilon^2)}$ single qubit density matrices on the partitions $\rho_P$ is feasible for at least one choice of $\Vec{r}, \Vec{c}$, and bounded away from saturating the magnetization constraints, and thereby via section \ref{appendix-volume} we know that the feasible region of $\tilde{C}_{\Vec{r}, \Vec{c}}$ contains a ball of radius $r'= O(\delta) = O(\epsilon^3/t^{3/2})$. Thus, we can check the feasibility of all the $\tilde{C}_{\Vec{r}, \Vec{c}}$ in time $2^{\tilde{O}(t/\epsilon^2)}$.  
\end{proof}

\section{A Free Energy PTAS on Dense Graphs}
\label{section-feptas-dense}

To extend the regularity lemma and its applications to the context of the free energy, we revisit the ground state energy approximation scheme devised in section \ref{section-gseptas}, with a point of view based on the results by \cite{Jain2018TheMA}. Recall how we estimated the minimum energy of $H$ among product states by reducing the computation to checking the feasibility of a small number of convex constraints. \cite{Jain2018TheMA} showed that for Ising models, it is the \textit{maximum entropy program} subject to these regularity-based constraints that enables an estimate for the true free energy. By combining our product state approximations for the free energy with our Hamiltonian regularity statements, we are able to draw quantum generalizations of their results on Local Hamiltonians. 

As we later discuss, since the free energy itself is a maximum entropy program \textit{regularized by the temperature}, our algorithms often incur a tradeoff between combinatorial, regularity-based errors, and thermal errors incurred from noise in our sampling algorithms. In this fashion we devise two main algorithms, the first of which is a sublinear time, additive error approximation algorithm which provides accurate approximations in a low temperature regime:

\begin{theorem} \label{theorem-variationalfe}
Fix $k, d = O(1)$, and $\epsilon, \delta >\omega(n^{-1/(2k-2)})$ and an inverse temperature $\beta > 0$, and let $H$ be a $k$-Local Hamiltonian on $n$ qudits of local dimension $d$ and $m$ bounded strength interactions. Then, there exists an algorithm that runs in time $2^{\tilde{O}(\epsilon^{2-2k})}\cdot O(\delta^{-2})$, that returns an estimate to the free energy accurate up to an additive error of $\epsilon n^{k/2}m^{1/2} + \delta n /\beta$ and is correct with probability $.99$. 
\end{theorem}

In effect, the thermal error above arises since we only have imperfect knowledge of the cut decomposition, and the sizes of the partitions within each cut. In the low temperature regime, whenever $\beta =\Omega( n^{1-k/2}m^{-1/2})$, this first algorithm ensures a $\epsilon n^{k/2}m^{1/2}$ approximation in $2^{\tilde{O}(\epsilon^{2-2k})}$ time, much like the sublinear time approximation algorithm for the ground state energy. Our second approach explicitly computes the cut decomposition, significantly improving the thermal error dependence at higher temperatures, at the cost of a higher runtime:

\begin{theorem} \label{theorem-variationalfeexplicit}
Fix $k, d = O(1)$, and $\epsilon, \delta >\omega(n^{-1/2)})$ and an inverse temperature $\beta > 0$, and let $H$ be a $k$-Local Hamiltonian on $n$ qudits of local dimension $d$ and $m$ bounded strength interactions. Then, there exists an algorithm that runs in time $2^{\tilde{O}(\epsilon^{-2})}\cdot \tilde{O}(n^k \log 1/\delta)$, that returns an estimate to the free energy accurate up to an additive error of $\epsilon n^{k/2}m^{1/2} + \delta n /\beta$ and is correct with probability $.99$. 
\end{theorem}

\subsection{Finding the best Product State Approximation}

In the setting of the section \ref{section-gseptas} and appendix \ref{section-regularityextensions}, given a $k$-Local Hamiltonian $H$ on $n$ qudits of local dimension $d$ and $m$ bounded strength interactions, let $J  = \{\|H_e\|_\infty\}_{e \in [n]^k}$ be its $k$ dimensional array of interaction strengths, and let $H_D$ be its cut decomposition following theorems \ref{lemma-kcutdecomp} and \ref{theorem-kmatrixregfrieze} with width $s = O(\epsilon^{2-2k})$. While we phrase most of the discussion in this section with the cut decomposition in Theorem \ref{theorem-kmatrixregfrieze} by \cite{Frieze1999QuickAT}, the analysis under the cut decomposition in Theorem \ref{theorem-kmatrixreg} by \cite{Alon2002RandomSA} follows analogously. To proceed, fix a precision parameter $\gamma$ and for every `guess' vector $r$ of size $k\cdot d^{2k}\cdot s$ with $r_i\in I_\gamma \subset [-n, n]$ for the average magnetization of each of $k$ sides of the $d^{2k}\cdot s$ cuts in the cut decomposition of $H$, we can formulate the convex program $O_{r, \gamma}$ to be the maximum entropy program over the subset magnetization constraints $C_{r, \gamma}$

\begin{equation}
    O_{r,  \gamma}= \max_{\alpha\in\mathbb{R}^{(d^{2-1})n}} \sum_{u\in V} S(\rho^{\alpha_u}) \text{ subject to } C_{r,  \gamma}
\end{equation}

In this manner, we express the problem of finding the product state of approximately minimum free energy as the minimum over the optima of $2^{O(k\cdot d^{2k}\cdot s\log 1/\gamma)}$ convex programs.

\begin{equation}
    F_\gamma = \min_{r: C_{r, \gamma} \text{feasible}} F_{r, \gamma} = \min_{r: C_{r,  \gamma} \text{feasible}} \bigg(\sum_{i\in [d^{2k}\cdot s]}d^{i} \prod_j^k r^i_j -  O_{r, \gamma}/\beta\bigg)
\end{equation}

To speedup the optimization, we replace the convex constraints by the induced constraints 
 $\tilde{C}_{r, \gamma}$ on the common refinement of the partitions in the cut decomposition. If $\mathcal{A}_a\subset [n]$, $a\in [A]$ is the coarsest partition of $H_D$, then $A \leq  2^{k\cdot d^{2k}\cdot s}$. We appropriately re-scale the objective to define the program $\tilde{O}_{r, \gamma}$:
\begin{gather}
    \tilde{O}_{r, \gamma} = \max_{\alpha\in\mathbb{R}^{(d^2-1)\cdot A}} \sum_{a\in [A]} |\mathcal{A}_a| S(\rho^{\alpha_a}) \text{ subject to } \tilde{C}_{r, \gamma}, \\ \text{ and return }\min_{r: \tilde{C}_{r, \gamma} \text{feasible}} \bigg(\sum_{i\in [d^{2k}\cdot s]}d^{i} \prod_j^k r^i_j -  \tilde{O}_{r, \gamma}/\beta\bigg)
\end{gather}

Recall that Claim \ref{claim-compressionfeasible} tells us that $C_{r, \gamma}$ is feasible $\iff \tilde{C}_{r, \gamma}$ is feasible, and thereby the convexity of the entropy ensures that $O_{r, \gamma} = \tilde{O}_{r, \gamma}$. Once again, since our goal is a sublinear time algorithm, we only have imperfect knowledge of the sizes $|\mathcal{A}_a|$ of the common refinements. However, we can construct estimates $|\hat{\mathcal{A}}_a|$ accurate up to some additive error efficiently, following Claim \ref{claim-estimatesizes}, and instantiate programs $\hat{O}_{r, \gamma}$ with `noisy' constraints $\hat{C}_{r, \gamma}$. We solve each of the convex programs using a result of \cite{Bertsimas2002SolvingCP}, an analog of Theorem \ref{theorem-bertsimas} for optimization which we present at the end of this section. 

Let us begin by reasoning on the correctness of this algorithm. First, we argue that the noise-less, compressed programs $\tilde{O}_{r, \gamma}$ provides a good approximation to the optimum product state assignment to the free energy of $H_D$. So long as the noise on $|\hat{\mathcal{A}}_a|$ isn't too large, next we prove that the optima of the noisy programs $\hat{O}_{r, \gamma}$ is close to the noise-less case. Finally, our product state approximations together with the Hamiltonian regularity lemma will ensure that the true free energy has a good product state approximation, tying the output of our algorithm to the free energy.

\begin{claim}
$\tilde{F}_{O(\epsilon^k)}$ is an $\epsilon \cdot n^{k/2} \|J\|_F$ additive approximation to $\min_{\rho=\otimes \rho_u}f_D(\rho) =\min_{\rho=\otimes \rho_u} \text{Tr}[H_D\rho]-S(\rho)/\beta$.
\end{claim}

\begin{proof}
Let $\rho^*$ be the product state minimizer of $f_D(\rho) = \text{Tr}[H_D\rho]-S(\rho)/\beta$. Correspondingly, let $r^{*} \in [-n, n]^{k\cdot d^{2k}\cdot s}$ be the subset magnetizations corresponding to $\rho^*$. By construction, there is a guess vector $r \in (I_\gamma)^{k\cdot d^{2k}\cdot s}$ s.t. $|r_j^{*}-r_j|\leq n \gamma /2 $, and thus $\rho^*$ is feasible for $C_{r, \gamma}$. Picking $\gamma = O(\epsilon^k)$ in Claim \ref{claim-kfeasibleconstraints} ensures that the energy of $\rho^*$ on $H_D$ is close to its estimate. We further observe that by definition, $O_{r, \gamma}\geq S(\rho^*)$, since $\rho^*$ is feasible for $C_{r, \gamma}$, and thus:

\begin{gather}
    f_D(\rho^*) = \text{Tr}[H_D\rho^*] - S(\rho^*)/\beta \geq - \epsilon n^{k/2} \|J\|_F +\sum_{i}d^{i}\prod_j^k r^{i}_{j} - S(\rho^*)/\beta \geq \\
    \geq -\epsilon n^{k/2} \|J\|_F +\sum_{i}d^{i}\prod_j^k r^{i}_{j} - O_{r, \gamma}/\beta
    \geq -\epsilon n^{k/2} \|J\|_F +\min_{r: C_{r, \gamma} \text{ feas }} \sum_{i}d^{i}\prod_j^k r^{i}_{j} - O_{r, \gamma}/\beta = \\ = -\epsilon n^{k/2} \|J\|_F + \hat{F}_{\gamma}
\end{gather}

The lower bound, in turn, requires a definition. Let $\rho^{r, \gamma}$ be the product state that maximizes the program $O_{r, \gamma}$ if feasible. Then, Claim \ref{claim-kfeasibleconstraints} ensures we have $f_D(\rho^{r,  \gamma})\leq F_{r, \gamma} +\epsilon n^{k/2} \|J\|_F$, if $\gamma = O(\epsilon^k)$. In turn, let $\rho_{r', \gamma}$ be the product state where $r'$ minimizes $\min_{r}\hat{F}_{r, \gamma}$. That is, $\hat{F}_\gamma = \hat{F}_{r', \gamma}$. By the variational description of the free energy product state optima of $f_D$, we obtain the lower bound:
\begin{equation}
    f_D(\rho^*) \leq f_D(\rho_{r', \gamma}) \leq \epsilon n^{k/2} \|J\|_F + \hat{F}_{r',  \gamma} =\epsilon n^{k/2} \|J\|_F+\hat{F}_{\gamma}
\end{equation}

\end{proof}

Now that we've ensured that the optima $\tilde{F}_{O(\epsilon^k)}$ of the noise-less programs is a good approximation to the variational free energy of $H_D$, let us prove that approximating this quantity with the noisy constraints is still a good approximation.

\begin{claim}
    Assume we have estimates for the sizes $|\hat{\mathcal{A}}_a|,a\in [A]$ of each partition in the common refinement of the cuts in $H_D$, accurate up to an additive error $\delta/A \cdot n$ with $\delta < \gamma/4$ and $\gamma = O(\epsilon^k)$. Then, the estimate $\hat{F}_{O(\epsilon^k)}$ corresponding to the optima of the noisy programs is an $\epsilon\cdot n^{k/2}\|J\|_F + \frac{\delta \log d}{\beta} \cdot n$ additive error approximation to the variational free energy of $H_D$, $\min_{\rho=\otimes \rho_u}f_D(\rho)$.
\end{claim}

\begin{proof}
    By Claim \ref{claim-approxfeas} (1), there is a choice of $r$ s.t. the optimum product state $\rho^*$ of $f_D$ is feasible for $C_{r, \gamma}$, and thus there is a choice of an $A$ qudit product state $\sigma^*$, one qudit for each subset in the common refinement, which is feasible for both $\tilde{C}_{r, \gamma}$ and $\hat{C}_{r, \gamma}$. We remark we can choose in particular $\sigma_a = |\mathcal{A}_a|^{-1} \sum_{u\in \mathcal{A}_a} \rho^*_u$. We note that $S(\rho^*)\leq \sum_{a\in A} |\mathcal{A}_a|\cdot S(\sigma^*_a)$ by subadditivity and convexity of the entropy, and thus $S(\rho^*)\leq \sum_{a\in A} |\hat{\mathcal{A}}_a|\cdot S(\sigma^*_a) + \delta \cdot n\cdot \log d$ by assumption. However, $\sigma^*$ is feasible for $\hat{C}_{r, \gamma}$, and thus $\sum_{a\in A} |\hat{\mathcal{A}}_a|\cdot S(\sigma^*_a)\leq \hat{O}_{r, \gamma}$. In this manner, once again picking $\gamma = O(\epsilon^k)$ gives
    
    \begin{equation}
        f_D(\rho^*) + \epsilon\cdot n^{k/2}\|J\|_F + \delta/\beta  \cdot n\cdot \log d  \geq  \min_{r: \hat{C}_{r, \gamma} \text{ feas }} \sum_{i}d^{i}\prod_j^k r^{i}_{j} - \hat{O}_{r, \gamma}/\beta = \hat{F}_{O(\epsilon^k)}
    \end{equation}
    
    Conversely, if $r$ is the vector that extremizes $\hat{F}_{\gamma}$, and if $\sigma^{r, \gamma}$ is an $A$-qudit product of density matrices which optimizes $\hat{O}_{r, \gamma}$, let us consider $\rho$ to be the $n$ qubit density matrix given by copying the assignment of each subset in the common refinement to every vertex within it: $\rho_u = \sigma^{r, \gamma}_a$ for $a\in \mathcal{A}_a$, $a\in [A]$. We note $\hat{O}_{r, \gamma} \leq S(\rho) + \delta \log d\cdot n$, and by Claim \ref{claim-kfeasibleconstraints}, 
    \begin{gather}
        f_D(\rho^*) \leq f_D(\rho) \leq  \epsilon\cdot n^{k/2}\|J\|_F  +\sum_{i}d^{i}\prod_j^k r^{i}_{j} - S(\rho)/\beta \leq \\ \leq \sum_{i}d^{i}\prod_j^k r^{i}_{j} - \hat{O}_{r, \gamma}/\beta +  \epsilon\cdot n^{k/2}\|J\|_F + \frac{\delta \log d}{\beta}\cdot n 
    \end{gather}
\end{proof}

We emphasize that this noisy estimate $\hat{F}_{O(\epsilon^k)}$ incurs a thermal error to the estimate of the free energy, since we use the noisy estimates for the sizes of the partitions to compute the entropy of the $n$ qubit system. In the low temperature regime $\beta  = \Omega(1/n^{k/2 - 1}\cdot \|J\|_F)$, we can essentially ignore this thermal error, however in the  high temperature regime this regularization starts taking effect. We return to this discussion shortly. By combining these results with our product state approximations for the free energy, we prove that these estimates in fact approximate the free energy of $H$.

\begin{corollary}
$\tilde{F}_{O(\epsilon^k)}$ is an $2\epsilon\cdot n^{k/2}\|J\|_F$ additive error estimate to  $F= \min_{\rho\geq 0} f(\rho)$, and $\hat{F}_{O(\epsilon^k)}$ is an $2\epsilon\cdot n^{k/2}\|J\|_F +  \frac{\delta \log d}{\beta}\cdot n $ additive error estimate to $F$.
\end{corollary}

\begin{proof}
    We note $\big|\min_{\rho: \rho=\otimes \rho_u} f_D(\rho) - \min_{\rho: \rho=\otimes \rho_u} f(\rho)\big|\leq \max_{\rho: \rho=\otimes \rho_u} \big|\text{Tr}[(H-H_D)\rho]\big| \leq \epsilon n^{k/2}\|J\|_F$, and thus the variational free energies of $H, H_D$ are close. Moreover, by Theorem \ref{theorem-feproduct} the free energy has a product state approximation $\big|\min_{\rho: \rho=\otimes \rho_u} f(\rho) - \min_{\rho\geq 0} f(\rho)\big|\leq O(n^{\frac{k-1}{3}}|J|_1^{1/3}\|J\|_F^{2/3}) = O(n^{k/2-1/3}\|J\|_F)$.
\end{proof}

 It suffices now only to argue the runtime of the algorithm to conclude the proof of our approximation algorithm to the free energy. We use the following result by \cite{Bertsimas2002SolvingCP}:

\begin{theorem}
[\cite{Bertsimas2002SolvingCP}]\label{theorem-bertsimasoptimization} Suppose $K\subset \mathbb{R}^m$ is a convex set, and $R, r \in \mathbb{R}$ and $y \in K$ are such that: $K$ is contained in the ball of radius $R$ centered at the origin, and, if $K$ is non-empty, $K$ contains the ball of radius $r$ centered at $y$. Assume $K$ has a separation oracle which is efficiently computable in time $T$. Further suppose $g:\mathbb{R}^m\rightarrow [-1, 1]$ is a convex function, which is efficiently computable and differentiable at any point $\gamma \in K$ in time $T'$. Then, with probability $1-2^{-\Omega(m)}$ we can compute a feasible point $x'\in K$ which is approximately minimal, 

\begin{equation}
    |g(x') - \min_{x\in K} g(x)|\leq \epsilon
\end{equation}

in time $\text{poly}(m, T, T')\cdot O( \log R/r \cdot  \log 1/\epsilon)$. 

\end{theorem}

We are now in a position to prove theorems \ref{theorem-variationalfe} and \ref{theorem-variationalfeexplicit}:
 
\begin{proof}

[of Theorem \ref{theorem-variationalfe}]

We note as in sections \ref{section-gseptas} and \ref{section-regularityextensions} that the convex sets $\tilde{C}_{r, O(\epsilon^k)}$ and $\hat{C}_{r, O(\epsilon^k)}$ have $m = A = 2^{O(\epsilon^{2-2k})}$, $R = \text{poly}(d, A) = 2^{O(\epsilon^{2-2k})}, r = O(\epsilon^k)$, as well as $T, T' =  \text{poly}(d, A) = 2^{O(\epsilon^{2-2k})}$. 

We implicitly compute the cut decomposition $H_D$ of Theorem \ref{theorem-kmatrixregfrieze}, and estimate the sizes of the coarsest partitions of the cuts using Claim \ref{claim-estimatesizes}, for an appropriate choice of $\delta < \frac{\gamma}{4 \log d}$. Then, we instantiate all the $2^{\tilde{O}(\epsilon^{2-2k})}$ maximum entropy programs subject to the noisy constraints $\hat{C}_{r, O(\epsilon^k)}$, and solve them all up to an additive error $\delta \cdot n$ in time $2^{\tilde{O}(\epsilon^{2-2k})}\cdot O(\log 1/\delta)$. Overall, with probability $.99$, this achieves an estimate for the true free energy accurate up to additive error $2\epsilon n^{k/2} \|J\|_F + 2\cdot \delta \cdot  n/\beta$, and runs in time $2^{\tilde{O}(\epsilon^{2-2k})}\cdot O(\delta^{-2})$ in the probe model of computation. In the low temperature regime, whenever $\beta =\Omega( n^{1-k/2}\|J\|_F^{-1})$, this provides an $\epsilon n^{k/2} \|J\|_F$ approximation in time  $2^{\tilde{O}(\epsilon^{2-2k})}$. 

\end{proof}

\begin{proof}

[of Theorem \ref{theorem-variationalfeexplicit}]

Alternatively, in the explicit approach, we explicitly compute the cut decomposition $H_D$ using Theorem \ref{theorem-kmatrixreg}, and and explicitly compute the coarsest partition of the cuts in time $O(n^k)\cdot 2^{\tilde{O}(1/\epsilon^2)}$. Recall that the array cut decomposition by \cite{Alon2002RandomSA} has width $O(1/\epsilon^2)$, and that we can use Claim \ref{claim-kfeasibleconstraints} to ensure the analogous accuracy guarantees to the free energy if we pick $\gamma = O(\epsilon^2)$. 

Under Theorem \ref{theorem-kmatrixreg}, we note the convex set $\tilde{C}_{r, O(\epsilon^k)}$ has $m = A = 2^{O(\epsilon^{-2})}$, $R = \text{poly}(d, A) = 2^{O(\epsilon^{-2})}, r = O(\epsilon^2)$, as well as $T, T' =  \text{poly}(d, A) = 2^{O(\epsilon^{-2})}$. We can then instantiate all the $2^{\tilde{O}(\epsilon^{2})}$ maximum entropy programs subject to the \textit{noise-less} constraints $\tilde{C}_{r, O(\epsilon^2)}$, and solve them up to an additive error $\delta\cdot n$ in time $2^{\tilde{O}(\epsilon^{-2})}\cdot O(\log 1/\delta)$. If we pick $\delta = \epsilon \min(\beta n^{k/2-1} \|J\|_F, 1)$, then overall this achieves an estimate for the true free energy up to additive error $\epsilon n^{k/2} \|J\|_F + \delta n/\beta\leq \epsilon n^{k/2} \|J\|_F\leq 2 \epsilon n^{k/2} \|J\|_F$ in time $O(n^k)\cdot 2^{\tilde{O}(\epsilon^{-2})} + 2^{\tilde{O}(\epsilon^{-2})} \cdot O(\log \frac{1}{\epsilon \beta n^{k/2-1} \|J\|_F})$.
\end{proof}

% \newpage
\section{A Ground State Energy PTAS on Sparse Graph Classes}
\label{section-sparsegsPTAS}

In this section, let us turn to the converse limit of Local Hamiltonians studied until now: those on sparse graphs. Our intention is to construct approximation schemes for the ground state energy and the free energy on certain restricted classes of sparse graphs, namely, graphs excluding a fixed minor, where their structure enables us to construct efficient divide-and-conquer and dynamic programming algorithms. We follow the ideas of \cite{Bansal2009ClassicalAS} and \cite{Brando2013ProductstateAT} on planar graphs, improving their results by using ideas from the algorithmic graph minor theory of \cite{Demaine2005AlgorithmicGM} and an improved quantum-to-classical mapping over the high degree vertices in the graph. In the next section, we extend these techniques to approximating the free energy as well. In particular, the main result of this section is the following theorem:

\begin{theorem} \label{theorem-sparsegsapprox}
Let $H = \sum_{e\in E}H_e$ be a 2-Local Hamiltonian defined on $n$ qubits, configured on an $h$-minor free graph $G = (V, E)$ where $|h| = O(1)$. Let the maximum interaction strength be $\max_e \|H_e\|_\infty = 1$, let the number of interactions be $m = |E|$, and let $\epsilon > 0$. Then there exists a clustered product state $\sigma$ that approximates the ground state energy of $H$ up to error

\begin{equation}
    \text{Tr}[H\sigma] \leq \min_\rho \text{Tr}[H\rho] + \epsilon m .
\end{equation}

Moreover, $\sigma$ can be found in $\text{poly}(n) + n\cdot 2^{O(\epsilon^{-9}\log 1/\epsilon)}$ time.
\end{theorem}

We organize the rest of this section as follows. In subsection \ref{summary-dynamicp}, we overview our approach and highlight our improvements to the quantum-to-classical mappings in previous work. In subsection \ref{prelimbidim}, we summarize and prove the combinatorial properties of $h$-minor free graphs that we need in our algorithms. In subsection \ref{subsection-boundeddeggraphs}, we restrict our attention to the easier bounded degree $h$-minor free graphs, and develop an algorithm which efficiently approximates their ground state energy based almost exclusively on our combinatorial decomposition theorems. Finally, in subsection \ref{subsection-highdegreedp}, we combine our bounded degree approach with the high-low degree techniques by \cite{Brando2013ProductstateAT} to develop a novel dynamic programming algorithm culminating in Theorem \ref{theorem-sparsegsapprox}.

\subsection{Overview}
\label{summary-dynamicp}

In this subsection, we briefly present a high level description of the approach and our main contributions to approximation schemes on Quantum 2-Local Hamiltonians defined on graphs that exclude a fixed minor. There is a rich literature of approximation algorithms to NP Hard problems on $h$ minor free graphs, which we extensively draw from in our paper. This is since these graphs quite generically have many interesting structural properties, which can be exploited to construct algorithms. One such property is the concept of a vertex separator. In a seminal work, \cite{Lipton1977AST} introduced the notion of a planar separator, a small set of vertices in a planar graph that once removed, divides the graph into roughly balanced disconnected components. This idea essentially enabled a divide and conquer approach to optimization problems on planar graphs, in which the (disconnected) components are optimized independently. A number of further studies significantly generalized the concept beyond the planar graph setting \cite{Gilbert1984AST,Alon1990AST, Klein1993ExcludedMN, Kawarabayashi2010AST, Nesetril2012CharacterisationsAE, Dvok2016StronglySS}. 

Another interesting graph-theoretic concept we draw from is that of the tree-width and the tree decomposition. While we defer most technical details to later in this work, the tree-width is an integer parameter that quantifies how far a graph is from a tree. The concept was originally developed by \cite{Robertson1986GraphMI}, and has been widely influencial in developing parametrized algorithms (often via dynamic programming) for problems that become simpler on graphs of low tree-width. For a review of equivalent definitions, algorithms, and applications, refer to \cite{Bodlaender1993ATG}. Most relevant to our paper is `Baker's technique' \cite{BakerBrenda1994ApproximationAF}, an algorithmic technique to approximate many NP Hard problems on planar graphs. The technique is essentially based on another decomposition theorem, that partitions planar graphs into disconnected components of small tree-width, which can be optimized independently. Later, a beautiful line of work \cite{Demaine2004EquivalenceOL, Demaine2004DiameterAT, Demaine2005GraphsEA,  Demaine2005AlgorithmicGM, Demaine2009ApproximationAV, Demaine2011ContractionDI} developed an algorithmic graph minor theory, which generalized Baker's technique and defined many other decompositions for $h$ minor free graphs. 

\cite{Bansal2009ClassicalAS} were the first to bring these ideas into the context of approximating QMA-hard problems. At a high-level, their algorithm for the minimum energy of quantum hamiltonians on planar graphs of bounded degree was to use a edge separator construction by \cite{Klein1993ExcludedMN}, which removed a small number of edges (or interactions) from $H$ such that the sub-Hamiltonian $H'$ that remains is disconnected, and organized into connected components of small size. These clusters could then be optimized independently by exact diagonalization. They used Weyl's inequality to formalize the error to the ground state energy that incurred by computing that of $H'$:

\begin{equation}
    \big|\lambda_{\min}(H)-\lambda_{\min}(H')\big| \leq \|H-H'\|_\infty \leq O(\text{Edge  Weight Removed})
\end{equation}

Naturally, neither the \textit{existence} of such a choice of a small number edges, nor algorithms to find them, are guaranteed for generic graphs. However, graphs that exclude a fixed minor are among the most general classes of graphs that these techniques are applicable to. Without yet entering into too many details of the decomposition, the key intuition we need is that graphs that exclude a fixed minor \textit{of bounded degree}, have efficient edge separators. The bounded degree constraint here is crucial, as the star graph (which is a planar graph) would immediately arise as a worst case to the decomposition techniques. 

\cite{Brando2013ProductstateAT} bypassed this bounded degree constraint, by directly arguing about the structure of the quantum ground states of graphs that have high degree vertices. They used information-theoretic ideas to argue that the two-particle reduced density matrices of high degree vertices with their neighbors are, on average, close to a product state. Quantitatively, they `condition' on all the particles of degree larger than poly$(1/\epsilon)$ in the graph being a product state, and prove such a state exists with energy only $\epsilon \cdot m$ above the ground state. By considering the original graph, with the high degree vertices removed, they then applied the decomposition techniques of \cite{Bansal2009ClassicalAS} to the bounded degree vertices in the graph. This information-theoretic `high-low degree' technique allowed them to discretize the Hilbert spaces of high degree vertices, and clusters of $2^{\tilde{O}(\epsilon^{-1})}$ low degree vertices, into $\epsilon$-nets, and therefore treat them as classical particles of spin dimension $2^{2^{\tilde{O}(\epsilon^{-1})}}$ on a 2-Local classical Hamiltonian.

In our work, we build on the graph decomposition techniques of \cite{Bansal2009ClassicalAS}, and the high-low degree techniques of \cite{Brando2013ProductstateAT}. In subsection \ref{prelimbidim}, we begin by refining this `cluster decomposition' using a generalized version of Baker's technique \cite{Demaine2005AlgorithmicGM}, and a recursive vertex separator theorem based on the tree decomposition for graphs of bounded edge weights. These graph-theoretic ideas allow us to improve on the size of the low degree clusters in the decomposition to poly$(1/\epsilon)$ vertices, which already enables us in subsection \ref{subsection-boundeddeggraphs} to present simpler and faster divide-and-conquer algorithms for approximating ground state and free energies of $h$ minor free graphs of bounded degree. The key idea to further compress the classical representation is the observation that if the clusters of low degree vertices have bounded size poly$(1/\epsilon)$, and if they all have bounded degree poly$(1/\epsilon)$, then the neighborhood of each cluster has at most poly$(1/\epsilon)$ high degree vertices. In this manner, the low degree clusters act as effective, poly$(1/\epsilon)$-local interactions between the high degree vertices. This observation hinges on the fact that once you condition on the all product states of the high degree neighbors of a given cluster (by choosing them from a $\epsilon$-net), the minimum energy (possibly entangled) state of that cluster is well defined by the hamiltonian terms acting on it. This results in the construction of an effective classical hamiltonian defined on the $O(n)$ high degree vertices, of locality poly$(1/\epsilon)$ and local dimension poly$(1/\epsilon)$. That is, essentially a classical $k$-CSP with $k=$ poly$(1/\epsilon)$. To actually solve this optimization problem, we heavily exploit properties of the tree decomposition that the graph is defined over, and formalize a dynamic programming algorithm in the known tree decomposition framework which we refer to as `high degree dynamic programming'.

\subsection{Preliminaries: Graph Minor theory}
\label{prelimbidim}

In this section, we discuss some ideas from Graph Minor theory, later used to construct approximations to the ground state energy. Let us begin by presenting the definition of the tree-width of a graph by \cite{Robertson1986GraphMI} and some of its properties. For a review on the subject, refer to \cite{Bodlaender1993ATG}.

\begin{definition} [\cite{Robertson1986GraphMI}\label{defdecomp}]
A tree decomposition of $G = (V, E)$ is a family of subsets $(X_i:i\in I)$ of $V$, arranged on a tree $T$ with vertices $I$, with the following properties:
\begin{enumerate}
    \item Every vertex is covered by a `bag', i.e. $\cup_i X_i = V$.
    \item Every edge is contained in a `bag', i.e. $u, v\in X_i$ for some $i$, for every $e = (u, v)\in E$. 
    \item Given three bags $i, j, k\in I$, if $j$ lies on the path in $T$ from $i$ to $k$, then $X_i\cap X_k \subset X_j$.
\end{enumerate}
\end{definition}

By inspection of the properties above, one easily deduces

\begin{lemma}[\cite{Bodlaender1993ATG}\label{treeprops}]
The tree-width is non-increasing under vertex and edge deletions, and edge contractions.
\end{lemma}

We additionally use the non-trivial property

\begin{lemma} [\cite{Bodlaender1993ATG}\label{treebin}]
Given a tree decomposition $T$ of a graph $G$ of tree-width $t$, one can construct a tree decomposition $T'$ which is a binary tree, of same tree-width, where the size of $T' = O(n)$.
\end{lemma}

While finding the exact tree-width of a graph is NP-Complete, for fixed small tree-width an algorithm by \cite{Bodlaender2016ACN} finds a constant multiplicative approximation in polynomial time:

\begin{lemma} [\cite{Bodlaender2016ACN}\label{treefind}]
Given a graph of tree-width $t$, one can find a tree-decomposition of width $O(t)$ in time $n\cdot 2^{O(t)}$.
\end{lemma}

Integral in our constructions will be to use the tree decomposition to cut up the graph into disconnected regions by deleting edges. A formalization of this notion is the concept of a graph separator, originally proposed by \cite{Lipton1977AST} for planar graphs. We use a result of \cite{Robertson1986GraphMV} to extract separators from the tree-decomposition:

\begin{lemma}
[\cite{Robertson1986GraphMV}, Vertex Separators]
Any graph $G=(V, E)$ of tree-width $t$ has a vertex separator of size $O(t)$. That is, there is a set $X$ of vertices s.t. removing them defines a partition into at least 2 components that do not have an edge between them, and are of size $\leq |V|/2 + 1$. Moreover, given a tree decomposition, one can find said separator in time $O(n^2)$.
\end{lemma}

To leverage this result to construct approximation algorithms, we use a simple recursive application of the result above. 

\begin{lemma} 
[Recursive Vertex Separators]\label{recvertexsep}
Let $G=(V, E)$ be a graph of tree-width $t$. Then it can be separated into mutually disconnected components each of size $\leq r$, by the removal of $O(tn/r)$ vertices. Given a tree decomposition of width $O(t)$, said separator decomposition can be constructed in $O(n^2)$ time.
\end{lemma}

To extend the previous work on approximation schemes for local Hamiltonians on planar graphs (\cite{Bansal2009ClassicalAS}) to H-minor free graphs, we use the following result by \cite{Demaine2005AlgorithmicGM} on the decomposition of said graphs.

\begin{theorem}
[\cite{Demaine2005AlgorithmicGM}\label{hminors}]
For a fixed graph $H$ of constant size, there exists a constant $c_H$ s.t. for every $k\geq 1$, and for every $H$-minor free graph $G$, the edges of $G$ can be partitioned into $k+1$ sets s.t. any $k$ of the sets induce a graph of tree-width at most $c_H k$. Furthermore, such a partition can be found in polynomial time $n^{O(1)}$.
\end{theorem}

To illustrate its application, as a warm-up, it is instructive to first discuss the approach of \cite{BakerBrenda1994ApproximationAF} on planar graphs. Consider an arbitrary starting vertex $s\in V$. Perform breadth-first-search, starting from $s$, within the planar graph. This procedure defines a sequence of layers, $L_1, L_2, \cdots $, where $L_i$ is the set of vertices at distance $i$ from $u$. Fix some integer parameter $t$. Consider the sets of interlacing layers $S_i = L_i, L_{i+t}, L_{i+2t}\cdots $. By construction, removing the set of vertices $S_i$ for some $i$ will partition the graph into sets of consecutive layers, e.g., $C_i = L_{i+1}, L_{i+2}\cdots L_{i+t-1}$. More importantly, if we connect all the nodes in $L_{i+1}$ with a source $v$, then $C_i\cup \{v\}$ has diameter $2t$, and since the graph is planar, this region has tree-width $O(t)$! This relation between diameter and tree-width is not unique to Planar Graphs. \cite{Eppstein2000DiameterAT} characterized the graphs that have this `linear local tree-width' property to be the set of apex-minor free graphs. In an incredible development, \cite{Demaine2005AlgorithmicGM} showed that analogous properties held for even more general classes, as we summarize in Theorem \ref{hminors} above.

\subsection{Bounded Degree $h$-Minor Free Graphs}
\label{subsection-boundeddeggraphs}
In this section, we discuss how to use the ideas from Bidimensionality theory to construct approximations to the ground state energy. In particular, we define a clustered product state, which is a tensor product of density matrices of clusters of vertices. As a warm-up, it is instructive to consider a simple divide and conquer algorithm for the bounded degree case. Let $H = \sum_{e\in E} H_e$ be a 2-Local Hamiltonian defined on an $h$-minor free graph $G = (V, E)$, where each vertex in $V$ has bounded degree $\Delta$. Let $J = \max_e \|H_e\|$ be the maximum interaction strength on this graph, $n = |V|, m = |E|$, and let us assume the minor has constant size $|h| = O(1)$. To construct an approximation scheme for the ground state energy of $H$, we first remove selected edges of $H$ in order the simplify the resulting optimization.

Let us first apply Theorem \ref{hminors} with $k = O(\epsilon^{-1})$. In this setting, we can pick a set $RE$ of edges of size $\leq m/(k+1) = O(\epsilon m)$ such that their removal defines components of tree-width $O(\epsilon^{-1})$ in time $n^{O(1)}$. Lemma \ref{treefind} tells us we can find a decomposition of tree-width $O(\epsilon^{-1})$ in the resulting graph in time $n\cdot 2^{O(\epsilon^{-1})}$. By now applying the Recursive Vertex Separator Lemma \ref{recvertexsep}, we remove a set $RV$ of vertices of size $|RV| = O(n/(r\epsilon))$ such that the resulting graph $G' = (V', E')$ can be arranged into at most $O(n)$ mutually disconnected components, each of maximum size $r$. Since the graph has bounded degree, this corresponds to removing $\Delta \cdot |RV| = O(\Delta n / r\epsilon)$ edges. Let $H'$ be the Hamiltonian defined simply by keeping the interactions in $G'$. Let us pick $r = O(\Delta / \epsilon^2)$, such that in total we have removed $O(\epsilon m)$ edges. Weyl's inequality tells us that the Hamiltonian $H'$ satisfies

\begin{equation}
    H' = \sum_{e\in E'}H_e, \text{ such that } \|H-H'\| = O(\epsilon m J)
\end{equation}

Consider a single component $C\subset V$ of $G'$, and recall $C$ has at most $r$ vertices, and is disconnected from the other components in $H'$. We can find its ground state $\sigma_C$ by explicitly diagonalizing

\begin{equation}
    \sigma_C = \argmin_\rho \text{Tr}\bigg[ \sum_{e\in E': e\in C} H_e \rho \bigg]
\end{equation}

\noindent in time $2^{O(r)} = 2^{O(\Delta / \epsilon^2)}$, as the support of the effective interaction above is only on $r$ qubits. Repeating for each component, we conclude with the following theorem.

\begin{theorem} \label{boundeddegreegsapproxtheorem}
Let $H$ be a 2-Local Hamiltonian defined on $n$ qubits arranged on an $h$-minor free graph of bounded degree $\Delta$, where the minor is constant sized $|h| = O(1)$. Fix a constant $\epsilon > n^{-1/2}$. Then there is an algorithm that finds a clustered product state $\sigma$ in time $n^{O(1)} + n \cdot 2^{O(\Delta / \epsilon^2)}$ such that 
\begin{equation}
    \text{Tr}[H\sigma] \leq \text{min}_\rho \text{Tr}[H\rho] + \epsilon m J
\end{equation}
\end{theorem}

\begin{proof}
As motivated in the previous discussion, one can construct a Hamiltonian $H'$ that approximates the spectra of $H$ via Theorem \ref{hminors} and Lemma \ref{recvertexsep} in runtime $n^{O(1)}+n\cdot 2^{O(\epsilon^{-1})}$. We can then optimize each component individually in total runtime $n\cdot 2^{O(\Delta/\epsilon^2)}$. We are guaranteed that the resulting state $\sigma = \otimes_C \sigma_C$ satisfies

\begin{gather}
    \text{Tr}[H\sigma]\leq \|H-H'\|_\infty + \text{Tr}[H'\sigma] \leq \|H-H'\|_\infty + \text{Tr}[H'\rho] \leq 2\|H-H'\|_\infty + \text{Tr}[H\rho] \\
    \Rightarrow \text{Tr}[H\sigma]\leq \text{min}_\rho \text{Tr}[H\rho]+O(\epsilon J m)
\end{gather}

by an adequate re-scaling of $\epsilon$ we obtain the desired approximation.

\end{proof}

\subsection{High-Degree Dynamic Programming}
\label{subsection-highdegreedp}

To extend the bounded-degree divide and conquer construction in the previous subsection to $h$-minor free graphs of arbitrary degree, we use the high-low degree technique of \cite{Brando2013ProductstateAT}. This technique is a slight modification to their general self-decoupling lemmas, which we briefly present here. Consider a generic density matrix $\rho$ on $n$ qubits. The goal is to construct a suitable approximation to $\rho$ via an ensemble of product states over the high-degree vertices, essentially breaking the entanglement over high-degree vertices. Similar to as we discussed in section 5, \cite{Brando2013ProductstateAT} pick a random subset $C$ of vertices to measure in a random Pauli basis $b$, obtaining a string $z$. Let $\mathcal{H}$ be the set of vertices of degree $\geq \Delta$, and $\mathcal{L}$ be the set of vertices of degree $\leq \Delta$, for some $\Delta$ to be chosen later. Once measured, the remaining un-measured qubits that have high-degree in the underlying graph, $\mathcal{H}\setminus C$, are placed in a tensor product, while the un-measured qubits of low degree, $\mathcal{L}\setminus C$, are simply unmodified. The resulting state is a classical distribution of these product states, one for each choice of $C$, basis measurement $b$, and measurement outcome $z$:

\begin{equation}
    \sigma = \mathbb{E}_{C, b}\mathbb{E}_{z}\bigotimes_{u\in C} \psi_{b_u, z_u} \otimes \sigma_{\mathcal{L}\setminus C}^{(C, b, z)} \bigotimes_{u\in \mathcal{H}\setminus C} \sigma_u^{(C, b,z)}
\end{equation}

\begin{remark}
In the interest of shortening notation, we let $\bigotimes_{u\in C} \psi_{b_u,z_u} = \psi_{C, b, z}$ and $\eta_{C, b} = \mathbb{E}_z \psi_{C, b, z}\otimes \sigma_{\mathcal{L}\setminus C}^{(C, b, z)} \bigotimes_{u\in \mathcal{H}\setminus C} \sigma_u^{(C, b,z)}$. More importantly, in this section, as $C$ doesn't play a role we shorten the description and let $\sigma = \mathbb{E}_{x} \sigma_{\mathcal{L}}^{(x)}\bigotimes_{u\in \mathcal{H}} \sigma_u^{(x)}$
\end{remark}

\begin{lemma}
[High-Low Degree Technique, \cite{Brando2013ProductstateAT}] \label{highlow} Let $H$ be a 2-Local Hamiltonian defined on a graph $G=(V, E)$ of $n$ particles of local dimension $d = O(1)$, and $|E| = m$. Let $\mathcal{H}$ be the set of vertices of degree $\geq \Delta$, and $\mathcal{L}$ be the set of vertices of degree $\leq \Delta$. Let $l$ be a positive integer parameter $< n$, and $C$ be a random set of vertices of uniformly random size at most $l$. Let $\rho$ be a generic state defined on these particles $V$. Then there exists a separable state $\sigma = \sigma = \mathbb{E}_{x} \sigma_{\mathcal{L}}^{(x)}\bigotimes_{u\in \mathcal{H}} \sigma_u^{(x)}$, that is an ensemble of states that are product over the high degree vertices, that approximates the energy of $\rho$ as
\begin{equation}
    \bigg|\text{Tr}[H\sigma] - \text{Tr}[H\rho]\bigg|= O\bigg(\frac{l}{n} + \sqrt{\frac{n}{l\Delta}}\bigg)Jm \equiv \delta_{l, \Delta}
\end{equation}
\end{lemma}

\begin{remark}
If we pick some $l = O(n \Delta^{-1/3})$,  we achieve an error of $\delta_{l, \Delta} = \Delta^{-1/3}Jm$. 
\end{remark}

\begin{corollary}
By an averaging argument, there exists a product state $\sigma^{(x)} = \sigma_{\mathcal{L}}^{(x)}\bigotimes_{u\in \mathcal{H}} \sigma_u^{(x)}$ that approximates the ground state energy of $H$:
\begin{equation}
    \text{Tr}[H\sigma^{(x)}] \leq \text{min}_{\rho}\text{Tr}[H\rho] + O(\Delta^{-1/3}Jm)
\end{equation}
\end{corollary}

The main takeaway of the above is that we can try to minimize the energy of $H$ among certain structured product states that break the entanglement of high degree vertices, up to certain error. 

Let us now consider a Hamiltonian defined on an $h$-minor free graph $G = (V, E)$, and similarly to the above partition the set of vertices into $V = \mathcal{L}\cup \mathcal{H}$ of degree $\leq \Delta$ and $> \Delta$ respectively. Consider the subgraph $G_{\mathcal{L}} = (\mathcal{L}, E_{\mathcal{L}})$ of $G$, defined by the vertices of $\mathcal{L}\subset V$, and the edges between them (basically delete $\mathcal{H}$ from the graph). This graph is still $h$-minor free, and by construction, has bounded degree $\leq \Delta$. In this manner, we can apply the separator framework of the previous section. We partition $\mathcal{L}$ into (disjoint subsets) clusters $C_1\cdots C_c$ of $c = O(n)$ clusters of low degree vertices, each of size $\leq r$, by deleting $O(\Delta n/r\epsilon + \epsilon m)$ edges. We can now define a Hamiltonian $H''$, corresponding to the interactions within the clusters, between clusters and high-degree vertices, and between the high-degree vertices - that is, we omit cluster-cluster interactions. 

\begin{equation}
    H'' = \sum_{e\in E, u, v\in \text{ some } C_i} H_e + \sum_{e\in E, u, v\in \mathcal{H}} H_e + \sum_{e\in E, u\in \mathcal{H}, v\in \mathcal{L}} H_e
\end{equation}

where once again Weyl's inequality guarantees these Hamiltonians have similar spectra

\begin{equation}
    \|H-H''\| \leq O(J\Delta n/r\epsilon + \epsilon Jm).
\end{equation}

Let us pick, for conciseness and later convenience, $\Delta = O(\epsilon^{-3})$, and $r = O(\epsilon^{-5})$, such that the error above is $O(\epsilon m J)$. Let us consider minimizing the ground state energy of $H''$ among clustered product states $ \bigotimes_{C_i \subset \mathcal{L}}\sigma_{C_i} \bigotimes_{u\in \mathcal{H}} \sigma_u$, where each of the high degree vertices is put in a product, and each of the clusters $C_i$ is also put in a product. Let $\sigma$ be said minimum. We will discuss how to find said $\sigma$ of such structure shortly, but for now let us discuss its guarantees. We claim

\begin{claim} \label{gssparse}
With $H''$ previously defined; If $\sigma$ is the minimizer of $\text{Tr}[H''\sigma]$ among product states of certain structure $\sigma = \bigotimes_{C_i \subset \mathcal{L}}\sigma_{C_i} \bigotimes_{u\in \mathcal{H}} \sigma_u$, then $\sigma$ also provides a good estimate for the ground state energy of $H$:
\begin{equation}
    \text{Tr}[H\sigma] \leq \min_\rho\text{Tr}[H\rho] + O(\epsilon J m).
\end{equation}
\end{claim}

\begin{proof} Let $D_2$ be the set of all density matrices over $n$ particles of the form $\sigma_{\mathcal{L}} \bigotimes_{u\in \mathcal{H}} \sigma_u$, and let $D_1$ be the set of all density matrices over $n$ particles of the form $\bigotimes_{C_i \subset \mathcal{L}}\sigma_{C_i} \bigotimes_{u\in \mathcal{H}} \sigma_u$, i.e. with a tensor product over the clusters of low degree vertices. Lemma \ref{highlow} tells us that if $\rho''$ is the ground state of $H''$, then there exists $\sigma_1\in D_1$ such that

\begin{equation}
   \text{Tr}[H''\sigma_1 ]  \leq \text{Tr}[H''\rho''] + O(\epsilon J m)
\end{equation}

We note that since $H''$ has no interactions between clusters, then one can construct a density matrix $\sigma_2\in D_2$ that is the marginalization of $\sigma_1$ in each cluster, i.e.

\begin{equation}
    \sigma_2 = \bigotimes_{C_i \subset \mathcal{L}}\sigma_{2, C_i} \bigotimes_{u\in \mathcal{H}} \sigma_u \text{ with }\sigma_{2, C_i} = \text{Tr}_{\mathcal{L}\setminus C_i} [\sigma_{1, \mathcal{L}}] \text{ such that }\text{Tr}[H''\sigma_1] = \text{Tr}[H''\sigma_2]
\end{equation}

with the guarantee that their energies over $H''$ are exactly the same. In this manner, 

\begin{equation}
    \min_{\sigma \in D_2}\text{Tr}[H''\sigma] \leq \text{Tr}[H''\sigma_2]=\text{Tr}[H''\sigma_1]\leq \text{Tr}[H''\rho''] + O(\epsilon J m)
\end{equation}

Let $\sigma^*\in D_2$ be the minimizer of the LHS above, corresponding to the minimum energy state among the clustered product states. We are thereby guaranteed that $\sigma^*$ is a good estimate for the ground state energy of the actual $H$, since

\begin{gather}
    \text{Tr}[H\sigma^*] \leq \text{Tr}[H''\sigma^*] +O(\epsilon J m) \leq \text{Tr}[H''\rho''] + O(\epsilon J m) \leq \\ \leq \text{Tr}[H''\rho] + O(\epsilon J m)  \leq \text{Tr}[H\rho] + O(\epsilon J m).
\end{gather}

Where, in sequence, we use the fact that $H, H''$ are close, then 
the previously derived relation with the ground state $\rho''$ of $H''$, next the fact that $\rho$ has larger energy than $\rho''$ on $H''$ by definition, and finally the fact that $H, H''$ are close once again.
\end{proof}

Now that we are guaranteed that the clustered product state $\sigma$ that minimizes the energy of $H''$ is also a good estimate for the ground state energy, let us reason how to find it.

\begin{claim} \label{dpgs}
Let $D_2$ be the set of clustered product density matrices of the form $\bigotimes_{i\in [c]}\sigma_{C_i} \bigotimes_{u\in \mathcal{H}} \sigma_u$. Then, there exists a deterministic algorithm that performs the minimization $\min_{\sigma\in D_2} \text{Tr}[H''\sigma]$ and returns a state $\sigma\in D_2$ in time $n\cdot 2^{O(\epsilon^{-9}\log 1/\epsilon)}$ such that 

\begin{equation}
    \text{Tr}[H''\sigma] \leq  \min_{\sigma \in D_2}\text{Tr}[H''\sigma]  + \epsilon m J
\end{equation}
    
\end{claim}

\begin{proof}
Let us consider the Hamiltonian $H''$ and the graph $G''$ it is defined on. Recall that $G''$ can be configured into a set of high degree vertices $\mathcal{H}$, and a set of disjoint clusters $C_1\cdots C_c$ of low degree vertices, where there are no interactions between clusters. For convenience, consider contracting each of these clusters into supernodes, each of local dimension $2^{|C_i|} = 2^{O(r)}$, each connected only to high-degree vertices. We note that given a tree decomposition of $G''$, we can construct another decomposition $T$ of the supernode graph in linear time, where the tree-width is still bounded by $O(\epsilon^{-1})$ by Lemma \ref{treeprops}. WLOG, by Lemma \ref{treebin} we can assume $T$ is a binary tree of size $O(n)$. 

As discussed by \cite{Bansal2009ClassicalAS} and \cite{Brando2013ProductstateAT}, we could construct a $\epsilon$-net over each of the clusters and particles defined in the decomposition. This would reduce the problem to that of a classical Hamiltonian, where we could perform the known dynamic programming algorithm on graphs of low tree-width to minimize the energy. The central caveat in this approach is that the number possible values for the spin of each cluster is doubly exponential in the size of the cluster, $2^{O(2^r \log 1/\epsilon)} = 2^{2^{\text{poly}(1/\epsilon)}}$, as the dimension of the Hilbert space of each cluster is already exponential in $r =\text{poly}(1/\epsilon)$, and the dimension of the net is exponential in the dimension of the Hilbert space. We emphasize that this is already an exponential improvement over the original construction of \cite{Brando2013ProductstateAT}, but we can further improve the runtime by refining the dynamic programming algorithm.

Consider augmenting the tree-decomposition $T$ as follows. For each bag $B\in T$, if there is a cluster-node $C_i\in B$, then we augment the bag $B$ by adding to it the entire neighborhood of the cluster $C_i$. We note that we are only adding high degree
vertices to the bag, as $C_i$ is only connected to high degree vertices in $G''$. By adding these neighborhoods, we argue that the resulting tree is still a valid tree decomposition, but now at a higher tree-width. It follows by carefully inspecting property 3 in the definition \ref{defdecomp} that the decomposition remains a valid tree-decomposition. More importantly, we note that we can bound the new width of the tree-decomposition by the sizes of the added neighborhoods. Since each cluster has size $\leq r$, each of degree $\leq \Delta$, the neighborhood of high degree vertices of each cluster must be bounded by $\leq r\cdot \Delta$. Since there are at most tree-width $t = O(k)= O(\epsilon^{-1})$ cluster-nodes in each bag, we have added at most $O(kr\Delta) = O(\epsilon^{-9})$ high degree vertices to each bag. This is the new tree-width of the augmented tree decomposition.

Let us now consider constructing an $\epsilon$-net over each high degree vertex in $\mathcal{H}$. Consider a fixed cluster $C_i$, and fix a configuration $\psi_{u_1}\cdots \psi_{u_{N(C_i)}}$ of its $O(r\Delta) = O(\epsilon^{-8})$ high degree vertex neighbors $u_1\cdots u_{|N(C_i)|}$, where here $|N(C_i)|$ denotes the size of the neighborhood. The key point in why the augmented decomposition is useful, is the fact that once the configuration of the neighbors $N(C_i)$ is fixed, the density matrix of the cluster $\sigma_{C_i}$ that minimizes the energy is well defined:

\begin{equation}
    \sigma_{C_i} = \text{argmin}_{\rho} \sum_{(u, v)\in E: u, v\in C_i} \text{Tr}_{u, v}\bigg[H_{uv}\rho_{uv}\bigg] + \sum_{(u, v)\in E: u\in C_i, v\in N(C_i)} \text{Tr}_{ u, v}\bigg[H_{uv} \rho_u \otimes \psi_{v}\bigg]
\end{equation}

Simply the ground state of an effective Hamiltonian. We emphasize that in this setting one can treat this interaction between the cluster $C_i$ and its high-degree neighborhood $N(C_i)$ as an effective $|N(C_i)|$-local classical interaction between the neighbors, as the minimum energy of the interactions with and within $C_i$ only depends on the spins of the neighborhood. Crucially, each of these interactions is covered in at least a bag via the augmentation. 

Now let us formulate the dynamic programming algorithm. We express the subproblems as follows. Consider a given bag $B$ in the tree decomposition $T$, let $u_1\cdots u_{|B|}$ be the high-degree vertices in the bag, and let $V_B$ be all the vertices of $G$ contained in the subtree of $T$ rooted at $B$. Our subproblem $E(B, \psi_{u_1}\cdots \psi_{u_{|B|}} )$ corresponds to the minimum energy of the classical hamiltonian $H_{V_B}$ defined by all the interactions fully contained in $V_B$, among states in the clustered product  convex set $D_2$, conditioned on fixing the states of the high degree vertices $u_1\cdots u_{|B|}$ in $B$ to be $\psi_{u_1}\cdots \psi_{u_{|B|}}$ (each a spin chosen from the $\epsilon$-net). We note that there are a total of $O(n)\cdot (1/\epsilon)^{O(|B|)} = O(n)\cdot 2^{\text{poly}(1/\epsilon)}$ such subproblems.

The base case of the dp is quite straightforward. It suffices to compute the minimum energy of $H_{V_B}$ for each configuration of high degree vertices in the bags $B$ that are leaves of the tree. Let us now recursively relate the subproblems defined on a bag $B$ with that of the at most two children $B_1, B_2$ of $B$ in $T$. We divide into cases on the interactions within $V_B$. Let $H_U$ be the Hamiltonian defined by interactions within $B$, that are not fully supported in neither $V_{B_1}$ or $V_{B_2}$. Let $H_{12}$ be the interactions fully contained within both $V_{B_1}$ and $V_{B_2}$, which by property 3 of the definition \ref{defdecomp}, we are guaranteed to be fully contained in $B$. Via inclusion-exclusion, we can express the classical Hamiltonian corresponding to the interactions within $V_B$ as $H_{V_B} = H_{V_{B_1}}+H_{V_{B_2}}+ H_U - H_{1, 2}$. 

Fix a configuration of the spins of $B$, $\psi_B = \psi_{u_1}\cdots \psi_{u_{|B|}}$. Let $B\cap B_i$ be the vertices that are both in $B$ and $B$'s $i$th child, $i\in \{1, 2\}$. The spins of this set must be consistent with those chosen in $B$. In this setting, one can express the recursive optimization as

\begin{gather}
    E(B, \psi_B) = \min_{\psi_{B_1\setminus B}, \psi_{B_2\setminus B}}\bigg[ E(B_1, \psi_{B_1\cap B}\otimes \psi_{B_1\setminus B}) + E(B_2, \psi_{B_2\cap B}\otimes \psi_{B_2\setminus B})\bigg] + \\ + \text{Tr}\bigg[\big(H_U-H_{1, 2}\big)\psi_B\bigg]
\end{gather}

For fixed subproblem $(B, \psi_B)$, computing the corresponding energy above involves iterating over the configurations in the $\epsilon$-net of the particles in $ B_1\setminus B$ and $B_2\setminus B$, and thereby takes $(1/\epsilon)^{O(|B|)} = 2^{O(\epsilon^{-9}\log 1/\epsilon)}$ time. Since there are $n\cdot 2^{O(\epsilon^{-9}\log 1/\epsilon)}$ subproblems, we can compute the entire dynamic programming table in time $n\cdot 2^{O(\epsilon^{-9}\log 1/\epsilon)}$. Finally, we return the minimum of $E(r, \psi_r)$ at the root $r$ of the tree $T$, over the choices $\psi_r$ in the $\epsilon$-net of the high degree vertices in $r$. 
\end{proof}

Claim \ref{dpgs} guarantees that we can find the minimum energy clustered product state $\sigma\in D_2$ that minimizes the energy of $H''$ in polynomial time, and Claim \ref{gssparse} guarantes that it is a good approximation for the ground state energy of $H$. This concludes the proof of the main theorem of this section, Theorem \ref{theorem-sparsegsapprox}.

% \newpage
\section{A Free Energy PTAS on Sparse Graph Classes}
\label{section-sparseFEPTAS}

In this section, we extend the divide-and-conquer and dynamming programming scheme to construct approximations to the free energy of 2-Local Hamiltonians defined on $h$-minor free graphs. We do so in 3 steps. First, we begin in subsection \ref{bounddegFE} by constructing a classical approximation scheme for the free energy of 2-Local Hamiltonians defined on $h$-minor free graphs of bounded degree. Parallel to our algorithms for the ground state energy of Hamiltonians on sparse graphs of bounded degree of section \ref{section-sparsegsPTAS}, we use ideas from Bidimensionality theory to construct Hamiltonians that are clustered in a sense, and simpler to optimize over. Our algorithms for the bounded degree case are randomized, they achieve extensive additive errors, and work in the full regime of temperature. We summarize: 

\begin{theorem} \label{theorem-degfeapprox}
 Let $H = \sum_{e\in E}H_e$ be a 2-Local Hamiltonian on $n$ qubits defined on an $h$-minor free graph $G = (V, E)$, where $|h|=O(1)$, $|E|=m$. Further assume that each particle only directly interacts with $\leq \Delta$ other particles, that is, the graph has bounded degree $\Delta$. Let the maximum interaction strength be $\max_e \|H_e\|_\infty = J$. Then there exists a classical, randomized algorithm that produces a clustered product state $\sigma$ that approximates the free energy of $H$ up to error
 
 \begin{equation}
    F \leq f(\sigma) \leq F + \epsilon m J
 \end{equation}
 
 The algorithm runs in time $n^{O(1)} + n\cdot 2^{O(\Delta/\epsilon^2)}$.
 \end{theorem}

We note that in the bounded degree case, both the quality of the approximation and the runtime are temperature independent.

To lift our approximation schemes to the arbitrary degree case, we first have to argue that separable and product states actually serve as good approximations to the free energy. To do so, in subsection \ref{clusteredapproxfe} we use an application of the entropy non-decreasing Theorem \ref{theorem-entropy} to argue that the Brandao-Harrow states produced by the high-low degree technique (Lemma \ref{highlow}) do also serve as good approximations to the free energy. 

 In the ensuing subsection \ref{findingclusteredapproxfe}, we construct a dynamic programming algorithm that approximately finds the minimum of the free energy among these clustered product states. As our approach hinges on a discretization of the Hilbert space of high degree vertices, we incur a small thermal error which limits our algorithms to the low temperature regime. To summarize, we prove:
 
 \begin{theorem} \label{theorem-feptas}
Let $H = \sum_{e\in E}H_e$ be a 2-Local Hamiltonian on $n$ qubits defined on an $h$-minor free graph $G = (V, E)$, where $|h|=O(1)$, $|E|=m$. Let the maximum interaction strength be $\max_e \|H_e\|_\infty = J$. Then there exists a classical, randomized algorithm that produces a clustered product state $\sigma$ that approximates the free energy of $H$ up to error
 
 \begin{equation}
    F \leq f(\sigma) \leq F + \epsilon m J 
 \end{equation}
 
 The algorithm runs in time $n^{O(1)} + n\cdot \max(2, \frac{1}{\beta J})^{\tilde{O}(\epsilon^{-9})}$ .
 
 \end{theorem}
 
 \begin{remark}
We emphasize that the algorithm of the theorem above approximates the free energy of $H$ up to constant error $\epsilon$ in the low temperature regime in polynomial time $n^{\text{poly}(1/\epsilon)}$, that is, for any polynomial temperature $T \leq O(Jn^c)$ for any integer $c$.
\end{remark}

 \subsection{Bounded Degree h-Minor Free Graphs} \label{bounddegFE}
 
 In this subsection we provide an algorithm for approximating the free energy of Hamiltonians defined on $h$-minor free graphs of bounded degree. We summarize in the theorem:
 
 \begin{theorem} [Theorem \ref{theorem-degfeapprox}, restatement]
 Let $H = \sum_{e\in E}H_e$ be a 2-Local Hamiltonian on $n$ qubits defined on an $h$-minor free graph $G = (V, E)$, where $|h|=O(1)$, $|E|=m$. Further assume that each particle only directly interacts with $\leq \Delta$ other particles, that is, the graph has bounded degree $\Delta$. Let the maximum interaction strength be $\max_e \|H_e\|_\infty = J$. Then there exists a classical, randomized algorithm that produces a clustered product state $\sigma$ that approximates the free energy of $H$ up to error
 
 \begin{equation}
    F \leq f(\sigma) \leq F + \epsilon m J
 \end{equation}
 
 The algorithm runs in time $n^{O(1)} + n\cdot 2^{O(\Delta/\epsilon^2)}$.
 \end{theorem}

Our approach to the bounded degree case follows that of the ground state approximation scheme of Theorem \ref{boundeddegreegsapproxtheorem}. We use ideas from Bidimensionality theory to decompose $G$ into disconnected clusters, that can be optimized independently. To recall the construction, we first apply Theorem $\ref{hminors}$ from the work of \cite{Demaine2005AlgorithmicGM} to decompose $G$ into components of tree-width $O(\epsilon^{-1})$, by the removal of $O(\epsilon m)$ edges. Next, we apply the recursive vertex separator Lemma \ref{recvertexsep} to further decompose the graph into disjoint components $C_1\cdots C_c$ of size $\leq r = O(\Delta/\epsilon^2)$, by removing another $O(\epsilon m)$ edges. We note that the number of components is $c=O(n)$. Let the resulting Hamiltonian defined over the interactions that weren't deleted be $H'$. We observe $\|H-H'\|_\infty \leq \epsilon J m$. 

\begin{remark} \label{remark-gibbsproduct}
The Gibbs state of $H'$ is a product state over the clusters, that is,

\begin{equation}
    e^{-\beta H'}/Z = \bigotimes_{i\in [c]} e^{-\beta H_{C_i}}/Z_{C_i},
\end{equation}

\noindent since the support of the Hamiltonians $H_{C_i}$ within each cluster $C_i$ are disjoint, and thereby commute.
\end{remark}

A simple but crucial technical lemma enables us to approximate the free energy of $H$ by using the Gibbs state of the simpler Hamiltonian $H'$, which is close to $H$.

\begin{lemma} \label{lemma-closefe}
Let $H$ and $H'$ be two 2-Local Hamiltonians, and $\rho, \rho'$ their corresponding Gibbs states, and $f(\cdot), f'(\cdot)$ their corresponding variational free energies. Then

\begin{equation}
 f(\rho)\leq  f(\rho') \leq f(\rho) + 2\cdot \|H-H'\|_\infty 
\end{equation}
\end{lemma}

\begin{proof}
The lower bound on $f(\rho')$ follows from the variational definition of the free energy. In turn, the upper bound satisfies:

\begin{equation}
    f(\rho') \leq f'(\rho') + \|H-H'\|_\infty \leq f'(\rho) + \|H-H'\|_\infty \leq f(\rho) + 2\cdot \|H-H'\|_\infty.
\end{equation}

\end{proof}

We now are in a position to prove our main theorem of this subsection. 

\begin{proof} 

[of Theorem \ref{theorem-degfeapprox}]

We note that given a generic hamiltonian on $r$ qubits, one can explicitly diagonalize the Hamiltonian and find all its eigenvectors and eigenvalues in time $2^{O(r)}$. One can therefore explicitly construct its Gibbs state at a certain temperature in time $2^{O(r)}$. 

In this setting, we describe our algorithm as follows. We first construct the decomposed Hamiltonian $H'$ as detailed above in time $n^{O(1)} + n\cdot 2^{O(\epsilon^{-1})}$. Let $C_1\cdots C_c$ be the clusters in the decomposed graph $G'$. For each cluster, we construct its Gibbs state explicitly by diagonalizing the Hamiltonian $H'_C$, the terms of $H'$ supported on $C$, in total time $O(n)\cdot 2^{O(r)} = n\cdot 2^{O(\Delta/\epsilon^2)}$. Finally, we return the product state $\sigma = \bigotimes_{i\in [c]} \sigma_{C_i}$.

We note that via remark \ref{remark-gibbsproduct}, the state we return is indeed the exact Gibbs state of $H'$. Via Lemma \ref{lemma-closefe}, we are guaranteed that the variational free energy $f(\sigma)$ is indeed close to the actual free energy of $H$. In fact, since $\|H-H'\| = O(\epsilon m J)$, by an appropriate choice of $\epsilon$ we obtain the approximation guarantee. 

\end{proof}

\subsection{Clustered Approximations to the Free Energy}
\label{clusteredapproxfe}

In this subsection we reason that clustered product states (as defined in Section \ref{section-sparsegsPTAS}), and ensembles of clustered product states, are good approximations to the free energy on $h$-minor free graphs. To do so, we first extend the high-low degree technique (Lemma \ref{highlow}) of \cite{Brando2013ProductstateAT} with the entropy non-decreasing theorem to argue that breaking the entanglement among high-degree vertices also provides a good approximation to the free energy on sparse graphs. We follow our approach in Section \ref{section-existence} to argue that in fact one can consider a single a product state over high degree vertices, that approximates the free energy up to twice the error in the high-low degree technique. Then, we reason as we did in section \ref{bounddegFE} by applying ideas from the Bidimensionality theory to construct a simpler, decomposed Hamiltonian that is easier to optimize over. Our main conclusion of this subsection is the idea that one can compute the minimum free energy clustered product state for the simpler Hamiltonian $H'$, and said state will be a good approximation for the true free energy. We formalize and summarize in Claim \ref{claim-feapproxexist}.

To begin, let us recall the setting of the high-low degree technique of \cite{Brando2013ProductstateAT}, presented in Lemma \ref{highlow}. Consider an $h$-minor free graph $G$, and let $\mathcal{H}, \mathcal{L}$ be the sets of vertices of degree $> \Delta$ and $\leq \Delta$ respectively. For any given state $\rho$ defined on qubits configured on a graph $G$, Lemma \ref{highlow} guarantees the existence of a separable state which is a product over both the states of the high-degree particles $\mathcal{H}$ and the measured particles $C$:

\begin{equation}
    \sigma =\mathbb{E}_{C, b} \eta_{C, b} = \mathbb{E}_{C, b} \mathbb{E}_z \psi_{C, b, z}\otimes \sigma_{\mathcal{L}\setminus C}^{(C, b, z)} \bigotimes_{u\in \mathcal{H}\setminus C} \sigma_u^{(C, b,z)}
\end{equation}

We note that the entropy non-decreasing theorem \ref{theorem-entropy} allows us to state that $S(\rho) \leq S(\eta_{C, b})$, for all possible choices of $C, b$. This is since our proof of Theorem \ref{theorem-entropy} does not require us to take marginals of every qubit, as is the structure of the high-low degree states. In fact, this observation enables us to carry the proof technique of theorems \ref{theorem-feseparable} and \ref{theorem-feproduct}, and to argue that a particular product state $\gamma_{C, b} = (\mathbb{I}_C/d^{|C|})\otimes \sigma_{\mathcal{L}\setminus C}^{(C, b, z^*)} \bigotimes_{u\in \mathcal{H}\setminus C} \sigma_u^{(C, b,z^*)}$ provides a good approximation to the free energy:

\begin{lemma} \label{highlowfe}
Let $H$ be a $2$-Local Hamiltonian defined on an $h$-minor free graph $G = (V, E)$. Let $\mathcal{H}, \mathcal{L}$ be the subsets of vertices in $G$ of degree $> \Delta$ and $\leq \Delta$ respectively. Let $C$ be a subset of vertices of size $\leq l$, and $b\in \{1, 2, 3\}^{|C|}$. Then there exists a product state $\gamma_{C, b} = (\mathbb{I}_C/d^{|C|})\otimes \sigma_{\mathcal{L}\setminus C}^{(C, b, z^*)} \bigotimes_{u\in \mathcal{H}\setminus C} \sigma_u^{(C, b,z^*)}$ for a particular choice of $C, b, z^*$, such that

\begin{equation}
    f(\gamma_{C, b}) \geq F \geq f(\gamma_{C, b}) - 2\delta_{l, \Delta}
\end{equation}

where $F$ is the free energy of $H$.
\end{lemma}

\begin{proof}
We remark that the entropy non-decreasing Theorem \ref{theorem-entropy} enables us follow the same proof of Theorem \ref{theorem-feseparable} and conclude that there exists $C, b$ such that the separable states $\eta_{C, b}$ are $\delta_{l, \Delta}$-additive approximations to the free energy. In particular, one in fact reasons that these states are good approximations in expectation
\begin{equation}
    \mathbb{E}_{C, b} f(\eta_{C, b})\leq F+ \delta_{l, \Delta}
\end{equation}

To prove the existence of a single product state over high-degree vertices that serves as a good approximation to the free energy, we follow the approach of the proof of Theorem \ref{theorem-feproduct}, in carefully constructing states $\gamma_{C, b} = (\mathbb{I}_C/d^{|C|})\otimes \sigma_{\mathcal{L}\setminus C}^{(C, b, z^*)} \bigotimes_{u\in \mathcal{H}\setminus C} \sigma_u^{(C, b,z^*)}$ by replacing the states of set of particles in $C$ with maximally mixed states, and optimally picking measurement outcomes $z^*$ for each $\eta_{C, b}$. Please refer to Theorem \ref{theorem-feproduct} for details. We obtain that the states $\gamma_{C, b}$ are good approximations to the free energy in expectation:

\begin{equation}
    \mathbb{E}_{C, b}f(\gamma_{C, b}) \leq F +2\cdot \delta_{l, \Delta}
\end{equation}

\end{proof}

The next question is whether we can simplify the Hamiltonian to impose more structure on the states of the un-measured low degree vertices, $\sigma_{\mathcal{L}\setminus C}$. Our intention will be to write these states of low degree vertices as a tensor product of its marginal density matrices of clusters of low degree vertices in the graph. To do so, let us now consider the Hamiltonian $H''$ of section \ref{section-sparsegsPTAS}. Let us briefly recall its construction: starting from $H$ defined on $G=(V, E)$, we first use Lemma \ref{hminors} to define a subgraph $G'$ of $G$ such that each connected component has tree width $O(\epsilon^{-1})$. Next, we partition the vertices of $G'$ into high and low degree vertices thresholded by a degree $\Delta = O(\epsilon^{-3})$, and use the recursive vertex separator Lemma \ref{recvertexsep} to decompose the subgraph of low degree vertices into components $C_1\cdots C_c$, $c=O(n)$, each of size at most $r = O(\epsilon^{-5})$. The total amount of edges removed from $G$ to form $G''$ was $O(\epsilon m)$. The following claim says that we can attempt to optimize the free energy by exploiting this clustering:

\begin{claim}\label{claim-feapproxexist}
Let $H$ be a 2-Local Hamiltonian defined on an $h$-minor free graph $G$, and let $H''$ be its corresponding Hamiltonian defined after the decomposition procedure of Section \ref{section-sparsegsPTAS}. Let $C_1\cdots C_c$ be the disjoint clusters of low degree particles in $H''$. Then, there exists clustered product state $\sigma = \bigotimes_{i\in [c]} \sigma_{C_i} \bigotimes_{u\in \mathcal{H}} \sigma_u$ that provides a good approximation to the free energy:

\begin{equation}
    f(\sigma_{C, b}) \geq F \geq f(\sigma_{C, b}) - 2\cdot \delta_{l, \Delta} - O(\epsilon m J)
\end{equation}
\end{claim}

The claim above is the key existence statement that enables us to later optimize over clusters independently. Our strategy to proving it follows closely to the proof of Claim \ref{gssparse}.

\begin{remark}
Let us pick $l = O(n\Delta^{-1/3})$, and as before $\Delta = O(\epsilon^{-3})$, such that an adequate rescaling of $\epsilon$ achieves total error $\epsilon m J$.
\end{remark}

\begin{proof}
Let $\rho, \rho''$ be the Gibbs states of $H, H''$ respectively. Lemma \ref{lemma-closefe} tells us that $F=f(\rho)\leq f(\rho'')\leq f(\rho) + O(\epsilon m J)$, since $H, H''$ are close. Let us now apply Lemma \ref{highlowfe}, the free energy variant of the high-low degree technique to $H'', \rho''$. We are guaranteed the existence of clustered product states $\gamma_{C, b}$ for each $C, b$ such that

\begin{equation}
    \mathbb{E}_{C, b}f''(\gamma_{C, b})\leq   f''(\rho'') + \delta_{l, \Delta}
\end{equation}

Fix $C, b$, and let us consider the state $\gamma_{C, b}$. Recall that it is of the form $\gamma_{C, b} = (\mathbb{I}_C/d^{|C|})\otimes \sigma_{\mathcal{L}\setminus C}^{(C, b, z^*)} \bigotimes_{u\in \mathcal{H}\setminus C} \sigma_u^{(C, b,z^*)}$, that is, the high degree vertices and the measured vertices are in a tensor product. Our intention is to leverage the structure of $H''$ to decompose the state $\sigma_{\mathcal{L}\setminus C}^{(C, b, z^*)}$ supported on the hilbert spaces of unmeasured low degree vertices, into a tensor product of low degree clusters. Similarly to the proof of Claim \ref{gssparse}, let us define reduced density matrices of clusters of un-measured low-degree vertices by a marginalization over each cluster: $\sigma^{(C, b, z)}_{C_i\setminus C} = \text{Tr}_{\mathcal{L}\setminus \{C_i, C\}}[\sigma^{(C, b, z)}_{\mathcal{L}\setminus C}]$. The key observation is that the marginalized state $\sigma_{C,b}$,

\begin{equation}
    \sigma_{C,b} = (\mathbb{I}_C/d^{|C|})\otimes\bigotimes_{i\in [c]}\sigma^{(C, b,  z)}_{C_i\setminus C}  \bigotimes_{u\in \mathcal{H}\setminus C} \sigma_u^{(C, b,z)},
\end{equation}

\noindent decreases the free energy, that is $f''(\sigma_{C, b})\leq f''(\gamma_{C, b})$. This is simply as the entropy is non-decreasing, via subadditivity, and the energy remains the same, as there are no interactions between different clusters in $H''$. In this manner, we have constructed clustered product states $\sigma_{C, b}$ for each $C, b$ such that 

\begin{equation}
    \mathbb{E}_{C, b }f''(\sigma_{C, b}) \leq f''(\rho'') + \delta_{l, \Delta}
\end{equation}

By an averaging argument, let us pick $C, b$ s.t. $f''(\sigma_{C, b})\leq \mathbb{E}_{C, b }f''(\sigma_{C, b})$. It follows straightforwardly that we can use such a choice to estimate the free energy of $H$:
\begin{gather}
    f(\rho)\leq f(\sigma_{C, b}) \leq f''(\sigma_{C, b}) + O(\epsilon m J) \leq   f''(\rho'') + \delta_{l, \Delta}+ O(\epsilon m J)\leq \\ \leq f''(\rho) + \delta_{l, \Delta}+ O(\epsilon m J)\leq f(\rho) + \delta_{l, \Delta}+ 2\cdot O(\epsilon m J)
\end{gather}

where we first used Lemma \ref{lemma-closefe} to reason that $f(\cdot)$ and $f''(\cdot)$ are close, then the derived relation between $f''(\sigma_{C, b})$ and $F''=f''(\rho'')$, and then the optimality $f''(\rho'')\leq f''(\rho)$. 
\end{proof}

In the next subsection, we discuss how to exploit the clustered structure of the states described in Claim \ref{claim-feapproxexist} to construct approximation algorithms.

\subsection{Finding Clustered Approximations to the Free Energy}
\label{findingclusteredapproxfe}

We dedicate this subsection to constructing a dynamic programming algorithm to find the clustered product state of minimum variational free energy. We leverage the discussion in Claim \ref{claim-feapproxexist} to argue that such states are indeed good approximations to the actual free energy. Similar to the ground state case of Section \ref{section-sparsegsPTAS}, our approach hinges on a discretization of the space of density matrices of the high-degree vertices, albeit in the free energy case we need to be slightly more careful in treating the error of the discretization.

\begin{theorem} [Theorem \ref{theorem-feptas}, restatement] 
Let $H = \sum_{e\in E}H_e$ be a 2-Local Hamiltonian on $n$ qubits defined on an $h$-minor free graph $G = (V, E)$, where $|h|=O(1)$, $|E|=m$. Let the maximum interaction strength be $\max_e \|H_e\|_\infty = J$. Then there exists a classical, randomized algorithm that produces a clustered product state $\sigma$ that approximates the free energy of $H$ up to error
 
 \begin{equation}
    F \leq f(\sigma) \leq F + \epsilon m J 
 \end{equation}
 
 The algorithm runs in time $n^{O(1)} + n\cdot \max(2, 1/\beta J)^{\tilde{O}(\epsilon^{-9})}$ .
 
 \end{theorem}

Let us now overview the proof of Theorem \ref{theorem-feptas}. We begin as we did in our algorithm to approximate the ground state energy in constructing the Hamiltonian $H''$ of Claim \ref{claim-feapproxexist} and Theorem \ref{theorem-sparsegsapprox}, in time $n^{O(1)}$. This defines a Hamiltonian $H''$ that is composed of interactions within high degree vertices $\mathcal{H}$, within a set of disjoint clusters $C_1\cdots C_c$ of low degree vertices, and between the clusters and high degree vertices. As previously argued in Claim \ref{claim-feapproxexist} and in the paragraph above Theorem \ref{theorem-feptas}, it suffices to approximately minimize the free energy of $H''$ among certain clustered product states. To do so, we use the high degree dynamic programming approach of Theorem \ref{theorem-sparsegsapprox} with two modifications. First, algorithmically, the objective we minimize is regularized by the entropy, second, within the analysis, we must bound the error to the objective by using a $\delta$-net over the high-degree vertices.

\begin{proof} 

[of Theorem \ref{theorem-feptas}]

Following the description above, let us study the structure of the objective on said clustered product states. We can describe the structure of the states we are optimizing over as $\sigma =  \bigotimes_{i\in [c]} \sigma_{C_i} \bigotimes_{u\in \mathcal{H}} \sigma_u$. Much like Section \ref{section-feptas-dense}, we can express the free energy of this product state as 

\begin{equation}
    f''(\sigma) = \text{Tr}[H''\sigma] - \sum_{i\in [c]}S(\sigma_{C_i}) - \sum_{u\in \mathcal{H}}S(\sigma_{u})
\end{equation}

 To find the minimum $\sigma$ among clustered product states for the objective above, let us consider defining a $\delta$-net over the high degree vertices. We choose to optimize over $\sigma_u$ chosen from within the discrete set of vectors in the net, for $u\in \mathcal{H}$. To bound the error of this partial discretization, we make the observation that at least one assignment $\sigma_u^\delta$ of spins in the set will be $\delta$-close to $\sigma_{u}^{*}$ for each $u\in \mathcal{H}$, where $\sigma^*$ is the clustered product state of minimum variational free energy $f''$ . Under this assignment, the error to the entropy of the high degree vertices becomes:

\begin{equation}
    \bigg|\sum_{u\in \mathcal{H}} S(\sigma_u^\delta) - \sum_{u\in \mathcal{H}} S(\sigma_{u}^{*})\bigg|\leq \sum_{u\in \mathcal{H}}\bigg| S(\sigma_u^\delta)-S(\sigma_{u}^{*})\bigg|\leq n\cdot O(\delta \log 1/ \delta )
\end{equation}

\noindent by the Fannes–Audenaert inequality. This is what incurs the temperature dependence at high temperatures. Meanwhile, we bound the error to the remaining part of the objective as follows. Let $\sigma^{\delta} = \bigotimes_{i\in [c]} \sigma_{C_i}^{\delta} \bigotimes_{u\in\mathcal{H}} \sigma_{u}^{\delta}$ be the minimizer among clustered product states when the particles of $\mathcal{H}$ have their spins $\sigma_{u}^{\delta, (C, b, z)}$ chosen from the $\delta$-net. Then, we can relate the variational free energy of this partial discretization as:

\begin{gather}
    \min_{\sigma_{u\in \mathcal{H}} \text{ in the net}}f''(\sigma)  = f''(\sigma^{\delta})  \leq f''\bigg(\bigotimes_{i\in [c]} \sigma_{C_i}^{*} \otimes \bigotimes_{u\in\mathcal{H}} \sigma_{u}^{\delta}\bigg) \leq \\ \leq O(\delta J m) + n/\beta \cdot O(\delta \log 1/ \delta ) + f''(\sigma^{*})
\end{gather}

Where, in sequence, we use the definition of $\sigma^\delta$, then the fact that it is the clustered product state of minimum variational free energy $f''$ to replace the states of the cluster by those of $\sigma^*$. This defines a `hybrid' state $\bigotimes_{i\in [c]} \sigma_{C_i}^{*} \otimes \bigotimes_{u\in\mathcal{H}} \sigma_{u}^{\delta}$, whose entropy is close to that of $\sigma^*$ via the previous relation on the entropies of $\sigma^\delta_u$ and $\sigma^*_u$ when $u\in \mathcal{H}$, and moreover also has energy close to that of $\sigma^*$, since

\begin{equation}
    \text{Tr}\bigg[H \bigg(\bigotimes_{i\in [c]} \sigma_{C_i}^{*} \otimes \bigotimes_{u\in\mathcal{H}} \sigma_{u}^{\delta} - \bigotimes_{i\in [c]} \sigma_{C_i}^{*} \otimes \bigotimes_{u\in\mathcal{H}} \sigma_{u}^{*}\bigg)\bigg] \leq 2\cdot J \cdot m \cdot \delta
\end{equation}

by Holders inequality and the $\delta$-net guarantees. That is, performing the partial discretization over high degree vertices incurs an error of $O(\delta Jm + n/\beta \cdot \delta \log 1/\delta)$ to the free energy. We note $\delta \log 1/\delta \leq \sqrt{\delta}$ since always $\delta < 1$, and therefore it suffices that $\delta = \epsilon \cdot \min(1, \Omega(\beta^2 J^2))$ to achieve an $O(\epsilon m J)$ approximation to the variational free energy of $H''$. We emphasize that this implies $\sigma^\delta$ is also an $O(\epsilon m J)$ approximation to the variational free energy of $H$, since $\|H-H''\|\leq O(\epsilon m J)$. Moreover, Claim \ref{claim-feapproxexist} thereby guarantees that the states $\sigma^{\delta}$ are also $O(\epsilon m J)$ additive approximations to the actual free energy of $H$.

To conclude, it suffices to discuss how to perform the actual minimization of $f''$ over the clustered product states. The key point is that since the entropy of product states is exactly additive, one can treat the interaction between a cluster $C_i$, and its high-degree neighborhood $N(C_i)$, as an effective $|N(C_i)|$-local classical interaction between the neighbors. This is simply since once the density matrices of the neighborhood is fixed, the minimum free energy state of the cluster is well determined, much like in the proof of \ref{dpgs} of the original high degree dynamic programming scheme. In this manner, we perform the dynamic programming algorithm with the new objective over the augmented tree decomposition in time $n\cdot (1/\delta)^{O(\epsilon^{-9})}$, given the previous bounds on the augmented width $t = O(\epsilon^{-9})$.

\end{proof}

\end{document}